\documentclass{article}
\usepackage{fullpage}
\usepackage{xcolor}
\usepackage{amsfonts,amssymb,amsmath,dsfont,amsthm}

\usepackage{subfigure}
\usepackage{tikz}
\usepackage{enumitem}
\usepackage{xspace}
\usepackage{graphicx}
\usepackage[colorlinks=true,citecolor=blue,breaklinks]{hyperref}
\usepackage[
    capitalize,
    noabbrev
]{cleveref}
\usepackage{cite}
\usepackage{bm}
\usepackage{comment}

\newtheorem{theorem}{Theorem}
\newtheorem{proposition}{Proposition}
\newtheorem{lemma}{Lemma}
\newtheorem*{lemmastar}{Lemma}
\newtheorem{corollary}{Corollary}
\newtheorem{definition}{Definition}
\newcommand{\paragraphproof}[1]{\noindent \textbf{#1}}

\newcommand{\proba}{\mathbb{P}}
\newcommand{\mean}{\mathds{E}}
\newcommand{\var}{\operatorname{Var}}
\newcommand{\probamalphawinner}{\proba(m \;\textnormal{is}  \; \valpha \textnormal{-winner})}
\newcommand{\probacondorcetwinner}[1]{\proba(#1 \;\textnormal{is CW})}
\newcommand{\bigO}{\mathcal{O}}

\newcommand{\integers}{\mathds{Z}}
\newcommand{\reals}{\mathbb{R}}

\newcommand{\indic}{\mathds{1}}
\newcommand{\diag}{\operatorname{diag}}
\newcommand{\hessian}{\mathcal{H}}
\newcommand{\transp}{\intercal}

\newcommand*{\eg}{\textit{e.g.,}\@\xspace}
\newcommand*{\ie}{\textit{i.e.,}\@\xspace}

\newcommand{\Real}{\operatorname{Re}}

\newcommand{\mA}{\mathcal{A}}

\newcommand{\mC}{\mathcal{C}}

\newcommand{\mH}{\mathcal{H}}

\newcommand{\mM}{\mathcal{M}}

\newcommand{\mS}{\mathcal{S}}

\newcommand{\mX}{\mathcal{X}}
\newcommand{\mY}{\mathcal{Y}}

\newcommand{\vect}[1]{\bm{#1}}
\newcommand{\vb}{\vect{b}}

\newcommand{\ve}{\vect{e}}

\newcommand{\vl}{\vect{\ell}}
\newcommand{\vm}{\vect{m}}
\newcommand{\vn}{\vect{n}}
\newcommand{\vp}{\vect{p}}
\newcommand{\vq}{\vect{q}}

\newcommand{\vs}{\vect{s}}
\newcommand{\vt}{\vect{t}}
\newcommand{\vu}{\vect{u}}
\newcommand{\vv}{\vect{v}}
\newcommand{\vx}{\vect{x}}
\newcommand{\vX}{\vect{X}}
\newcommand{\vy}{\vect{y}}
\newcommand{\vY}{\vect{Y}}
\newcommand{\vz}{\vect{z}}
\newcommand{\vzero}{\vect{0}}
\newcommand{\vone}{\vect{1}}
\newcommand{\vtwo}{\vect{2}}
\newcommand{\valpha}{\vect{\alpha}}
\newcommand{\vbeta}{\vect{\beta}}

\newcommand{\vkappa}{\vect{\kappa}}
\newcommand{\vlambda}{\vect{\lambda}}
\newcommand{\vtheta}{\vect{\theta}}
\newcommand{\vtau}{\vect{\tau}}
\newcommand{\vzeta}{\vect{\zeta}}

\DeclareMathOperator*{\argmax}{arg\,max}
\DeclareMathOperator*{\argmin}{arg\,min}
\DeclareMathOperator\erf{erf}

\usepackage{listings}
\usepackage{graphicx}

\newcommand{\elie}[1]{\textcolor{red}{(Elie) #1}}
\newcommand{\fanfan}[1]{\textcolor{red}{(Fanfan) #1}}

\usepackage{pgfplots}
\newcommand{\axisWidth}{10cm}
\newcommand{\axisHeight}{5cm}
\newcommand{\legendFont}{\small}

\begin{document}

\title{Probability of a Condorcet Winner for Large Electorates: \\ An Analytic Combinatorics Approach}
\author{%
Emma Caizergues$^{1, 2, }$\thanks{Authors presented in alphabetical order.},
Fran\c{c}ois Durand$^1$,
Marc Noy$^3$,\\
\'Elie de Panafieu$^1$,
Vlady Ravelomanana$^4$\\
{\normalsize%
$^1$ Nokia Bell Labs,
$^2$ CNRS - Lamsade, Université Paris Dauphine - PSL, Paris, France,}
\\
{\normalsize%
$^3$ Universitat Polit\`ecnica de Catalunya, Barcelona, Spain,}
\\
{\normalsize%
$^4$ IRIF - UMR CNRS 8243, Université Paris Denis Diderot}}
\date{}

\maketitle

\begin{abstract}
We study the probability that a given candidate is an $\valpha$-\emph{winner}, \ie a candidate preferred to each other candidate~$j$ by a fraction $\alpha_j$ of the voters. This extends the classical notion of \emph{Condorcet winner}, which corresponds to the case $\valpha=(\frac{1}{2},\dots,\frac{1}{2})$. Our analysis is conducted under the general assumption that voters have independent preferences, illustrated through applications to well-known models such as Impartial Culture and the Mallows model. 
While previous works use probabilistic arguments to derive the limiting probability as the number of voters tends to infinity, we employ techniques from the field of analytic combinatorics to compute convergence rates and provide a method for obtaining higher-order terms in the asymptotic expansion. 
In particular, we establish that the probability of a given candidate being the Condorcet winner in Impartial Culture is  
\(
a_0 + a_{1, n} n^{-1/2} + \bigO(n^{-1}),
\)  
where we explicitly provide the values of the constant \( a_0 \) and the coefficient \( a_{1, n} \), which depends solely on the parity of the number of voters \( n \).
Along the way, we derive technical results in multivariate analytic combinatorics that may be of independent interest.
\end{abstract}

\section{Introduction}

\paragraph{Motivation}
A long-standing tradition in social choice theory is to  compute the probability that a \emph{Condorcet winner} exists, \ie a candidate preferred to every other candidate by a majority of voters. A comprehensive overview of this research is presented in \cite{gehrlein2006condorcet}, which is entirely dedicated to this topic. \cite{niemi1968paradox} derived the limit of this probability as the number of voters tends to infinity, under the only assumption of independent preferences. More recently, \cite{krishnamoorthy2005condorcet} revisited and modernized this result.

However, to the best of our knowledge, all previous studies have relied on classical combinatorial or probabilistic methods. For instance, the results by \cite{niemi1968paradox} and \cite{krishnamoorthy2005condorcet} are primarily based on a Gaussian approximation. 
In this paper, we introduce techniques from \emph{analytic combinatorics} into this domain, drawing on the approach outlined in the reference book by \cite{flajolet2009analytic}. These methods encode combinatorial problems into analytic objects, allowing the application of powerful analytical tools to extract information, especially  limits, but also speeds of convergence, and more generally asymptotic expansions as certain parameters tend to infinity.

\paragraph{Contributions}
Similarly to \cite{niemi1968paradox} and \cite{krishnamoorthy2005condorcet}, we consider a general probability distribution over voter preferences which we call \emph{General Independent Culture (GIC)}, with the primary assumption that individual preferences are independent. Additionally, we assume that the distribution is \emph{generic}, in the sense that every possible ranking over the candidates has a positive probability.

To obtain results as general as possible, we introduce the notion of an \emph{$\valpha$-winner}, a generalization of the Condorcet winner, where the required victory threshold for each pairwise comparison is not necessarily one-half of the voters.

We analyze the probability that a given candidate is an $\valpha$-winner, with a particular focus on its asymptotic behavior as the number of candidates \( m \) remains fixed while the number of voters \( n \) tends to infinity. Our main contributions lie in providing a method that not only recovers previous results on limiting probabilities using a different approach but also enables the computation of the rate of convergence and higher-order terms in the asymptotic expansion.

These general theoretical results are illustrated through applications to the usual notion of Condorcet winner under probabilistic models known as the \emph{Impartial Culture} and the \emph{Mallows model}. In particular, we establish that the probability of a given candidate being the Condorcet winner in Impartial Culture is  
\(
a_0 + a_{1, n} n^{-1/2} + \bigO(n^{-1}),
\)  
where we explicitly provide the values of the constant \( a_0 \) and the coefficient \( a_{1, n} \), which depends solely on the parity of the number of voters~\( n \).

In the process, we also establish results in multivariate analytic combinatorics that are of independent interest.

\paragraph{Related Work}\label{sec_related-work}
Regarding the probability of a Condorcet winner, the standard reference is \cite{gehrlein2006condorcet}. In the General Independent Culture (GIC), which assumes independent voter preferences, the cases for finite values of voters \(n\) and candidates \(m\) were explored by \cite{gehrlein1976probability}, \cite{gillett1978recursion}, and \cite{gehrlein1983condorcet}.
The most relevant works to ours are \cite{niemi1968paradox} and \cite{krishnamoorthy2005condorcet}, which derive the limiting probability as \(n\) tends to infinity in this general setting. Our approach differs in methodology and results: we use analytic combinatorics and provide more precise asymptotic behaviors.

Among the specific cases of GIC, the most studied is \emph{Impartial Culture (IC)}, where all rankings are equally likely \cite{gehrlein2006condorcet}. For finite values of \( n \) and \( m \) or for the asymptotic regime where \( m \to \infty \), see \cite{sauermann2022probability}. Other works, like ours, focus on large electorates (\ie \( n \to \infty \)). An explicit formula exists for \( m = 3 \) \cite{guilbaud1952theories, niemi1968paradox, garman1968paradox}, while the limit for \( m=4 \) follows from a recurrence relation \cite{may1971some, fishburn1973proof}. Explicit formulas for \( m \in \{5, 6, 7\} \) have been derived in \cite{gehrlein1978probabilities, gehrlein1983condorcet}. The general case of arbitrary \( m \) with \( n \to \infty \) was addressed by \cite{niemi1968paradox} and \cite{krishnamoorthy2005condorcet}, whose results we recall in \Cref{th:impartial}. In IC, as in the broader GIC framework, our main contribution is to provide the rate of convergence and a method for computing higher-order terms in the asymptotic expansion.


Other specific cases of GIC have been studied, notably the \emph{Dual Culture} \cite{gehrlein2006condorcet} and the \emph{Perturbed Culture} \cite{williamson1967social,gehrlein2006condorcet}. Although the \emph{Mallows model} \cite{mallows1957non} is popular, the probability of a Condorcet winner in this framework has received little attention. This is likely because, for large electorates, the limiting behavior is trivial: the candidate favored by the culture becomes the Condorcet winner with probability tending to 1. The question of the convergence rate, which we explore here, is much more interesting.

More distant from our work, other models of culture fall outside the GIC framework as they do not assume voter independence. Notable examples include \emph{Impartial Anonymous Culture (IAC)} \cite{gehrlein1976condorcet,lepelley1989contribution,gehrlein2006condorcet} and, more generally, \emph{Pólya-Eggenberger models} \cite{berg1992note,gehrlein2006condorcet}.

One can also study the notion of a \emph{Weak Condorcet Winner}, a candidate who wins or ties in every pairwise comparison. \cite{krishnamoorthy2005condorcet} note that the probability of a tie in a pairwise comparison tends to zero as the electorate grows, implying that the limiting probability of a Weak Condorcet Winner is the same as for a Condorcet Winner. However, the convergence speeds may differ, as discussed in \Cref{sec_mallows-last-3}. Other works examine the probability that the \emph{pairwise majority relation} is transitive \cite{gehrlein2006condorcet,kelly1974voting,bell1978happens,bell1981random}. While this could be addressed with analytic combinatorics, it is beyond the scope of this paper.


Our analysis is based on representing voting configurations by polynomials, expressing probabilities via coefficient extraction, and deriving asymptotics using the saddle-point method. In Impartial Anonymous Culture, many voting questions \cite{gehrlein2002arithmetic,huang2000analytical,lepelley2008ehrhart,wilson2007probability} have been studied using Ehrhart polynomials \cite{barvinok1994polynomial}, reducing the enumeration of voting configurations to counting integral points in a polyhedron. \cite{maassen2002weak} use generating functions to express the probability of a Condorcet Winner in Impartial Culture or its variant with weak orders, obtaining exact expressions and asymptotics for several cases. However, none of these works apply complex analysis to derive asymptotics. \cite{may1971some} uses the saddle-point method to study the asymptotics as $m \to \infty$ of the limit as $n \to \infty$ of the probability that a candidate is the Condorcet Winner in Impartial Culture. This double limit problem is not addressed in this paper and leads to very different computations, involving only a univariate real integral.

\paragraph{Limitations}
A limitation of our work is that we focus on the probability that a \emph{given} candidate is an $\valpha$-winner, rather than the probability that at least one exists. For certain values of $\valpha$, multiple candidates may satisfy this condition simultaneously, requiring inclusion-exclusion techniques to compute the probability of existence. We exclude this issue from the scope of this paper, as it does not arise for the Condorcet winner, who, if existent, is unique; thus, its probability is simply the sum of individual probabilities.

Moreover, our main results concern the asymptotic behavior as \( n \) tends to infinity. While our approach also provides exact expressions for finite \( n \), these expressions are essentially the same as those obtained using classical combinatorial methods. The main advantage of our method, therefore, lies in its ability to analyze asymptotic behavior.

We assume that voters have independent preferences, excluding cases with more complex interactions. However, De Finetti's theorem \cite{definetti1929funzione,hewitt1955symmetric} asserts that \emph{exchangeable} random variables (i.e., roughly symmetric ones) are conditionally independent given a latent variable. In our case, the preferences of voters are exchangeable up to random relabeling. Therefore, the problem can, in principle, be reduced to independent preferences by conditioning on the latent variable and integrating over its possible values.

To avoid degenerate cases, we also assume that the preference distribution is \emph{generic}, meaning that all rankings have positive probability.
While analytic combinatorics methods could be applied without this assumption, doing so would introduce additional subcases that we prefer to avoid, where the \emph{saddle point} (which we will define shortly) has null or infinite coordinates.

Finally, in some situations, our method requires computing a sum that, in the worst case, may contain up to \( \mathcal{O}(2^m) \) terms (cf. \Cref{sec_supercritical-case}). However, it is worth noting that merely specifying the full probability distribution already requires \( m! \) real numbers in general.

\paragraph{Roadmap}

\Cref{sec_preliminaries} presents the necessary preliminaries. \Cref{sec_step-by-step} translates our combinatorial problem into an analytic question and introduces the notion of \emph{saddle point}. Depending on the coordinates of the saddle point, the analysis then branches into different cases, detailed in \Cref{sec_subcritical-case,sec_critical-case,sec_mixed-case,sec_supercritical-case}. \Cref{sec_error-terms} extends the analysis by showing how to derive higher-order terms in the asymptotic expansion. \Cref{sec_conclusion} concludes.

\Cref{sec_appendix-Elie} establishes the technical results of analytic combinatorics used throughout the paper.  
\Cref{sec_appendix-gic} proves general properties of the GIC model.  
\Cref{sec_appendix-ic-and-mallows} focuses on properties specific to Impartial Culture and the Mallows model.  

The paper is accompanied by the Python package \texttt{Actinvoting} (Analytic Combinatorics Tools In Voting), which was used for symbolic mathematics and numerical simulations. It is available at\\ \url{https://github.com/francois-durand/actinvoting}.

\section{Preliminaries}\label{sec_preliminaries}

In this section, we begin by introducing the main definitions and notations, and we formulate our research question.
We then briefly present the field of analytic combinatorics.

\subsection{Definitions, Notations, and Research Question}\label{sec_definitions-and-notations}

The number of voters is denoted by \( n \). We define the set of candidates as $\{1, \dots, m\}$, where \( m \) denotes the number of candidates.  A \emph{profile} represents the preferences of the voters and is defined as a list of \( n \) rankings over the set of candidates.  A \emph{culture} refers to a probability distribution over the \( (m!)^n \) possible profiles. We assume that voters have independent preferences, meaning the culture is characterized by a probability $p_r$ for each preference ranking~$r$.
Surprisingly, there is no standard terminology for this common assumption, which we refer to as a \emph{General Independent Culture (GIC)}.
We also assume that the culture is \emph{generic}, in the sense that \( p_r > 0 \) for every ranking~\( r \). The notation \( \proba \) refers to the probability under the given culture.

A \emph{Condorcet winner} is a candidate who is preferred to every other candidate by more than \( \frac{n}{2} \) voters. Without loss of generality, we study the probability $\probacondorcetwinner{m}$ that candidate \( m \) is the Condorcet winner.  
Our theoretical results extend this question to a more general notion, which we call an \( \valpha \)-\emph{winner}, defined as follows. Let \( \mathcal{A} := \{1, \dots, m-1\} \) be the set of adversaries of candidate \( m \). 
Consider a vector \( \valpha := (\alpha_1, \dots, \alpha_{m-1}) \in (0,1)^{m-1} \), where each \( \alpha_j \) represents the proportion of voters that candidate \( m \) needs on their side to win the pairwise comparison against candidate \( j \).
For an adversary \( j \in \mathcal{A} \), we write \( m \succ_{\valpha} j \), and we say that \( m \) wins against \( j \) in the sense of \( \valpha \), if candidate \( m \) is preferred to~\( j \) by more than \( \alpha_j n \) voters. For a subset of adversaries \( \mX = \{j_1, \dots, j_k\} \subseteq \mathcal{A} \), we write \( m \succ_{\valpha} \mX \), or equivalently \( m \succ_{\valpha} j_1, \dots, j_k \), if \( m \succ_{\valpha} j \) holds for every \( j \in \mX \). Similarly, we write \( m \preccurlyeq_{\valpha} j \), and we say that \( m \) does not win against \( j \) in the sense of \( \valpha \), if candidate \( m \) is preferred to \( j \) by at most \( \alpha_j n \) voters.  
We say that candidate \( m \) is an \( \valpha \)-\emph{winner} if \( m \succ_{\valpha} \mathcal{A} \), \ie \( m \) is preferred to each adversary \( j \) by more than \( \alpha_j n \) voters. The standard notion of a Condorcet winner corresponds to the special case \( \valpha = \left(\frac{1}{2}, \dots, \frac{1}{2}\right) \). The notion of a \emph{generalized Condorcet winner}, as introduced by \cite{sertel2004strong}, is recovered by considering a vector \( \valpha \) whose coordinates are equal.  

The goal of this paper is to study \(\probamalphawinner\), the probability that candidate \( m \) is an \mbox{\( \valpha \)-winner}, with a particular focus on its asymptotic behavior as the number of voters \( n \) tends to infinity. Note that the probability of the existence of a Condorcet winner can be obtained as the sum of the probabilities for all \( m \) candidates.

Our theoretical results will be illustrated using the \emph{Mallows model} \cite{mallows1957non}. This model is characterized by a reference ranking \( r_0 \) over the candidates and a concentration parameter \( \rho \in \mathbb{R}_{\ge 0} \). The probability of a ranking \( r \) is given by  
\[
p_r := \gamma e^{-\rho d(r,r_0)},
\]  
where \( \gamma \) is a normalization constant ensuring \( \sum_r p_{r} = 1 \), and \( d(r, r_0) \) denotes the Kendall-tau distance \cite{kendall1938new}, which counts the minimum number of adjacent swaps needed to transform \( r \) into~\( r_0 \). 
The Mallows model is commonly used in the field of \emph{epistemic democracy} \cite[Chapters 8 and 10]{brandt2016handbook}, where noisy evaluations from voters are collected with the aim of uncovering a hidden truth. The reference ranking models the hidden truth, and the concentration parameter \( \rho \) indicates the skill level of the voters.
When \( \rho = 0 \), all rankings are equally probable, recovering the classical \emph{Impartial Culture} model. We denote by \( \mathcal{M}_{m \;\textnormal{last}} \) (resp. \( \mathcal{M}_{m \;\textnormal{first}} \)) a Mallows culture with parameter \( \rho \), where the reference ranking \( r_0 \) places candidate \( m \) last (resp. first). We can assume, without loss of generality, that the reference ranking is \( (1,\dots,m) \) (resp. \( (m,\dots,1) \)).

Finally, we typeset vectors in boldface, \eg \( \vx := (x_1, \dots, x_{m-1}) \). If \( \mX \) is a set of indices, we define \( \vx_{{}_{\mX}} := (x_j)_{j \in \mX} \). We define $\log(\vx):=(\log(x_1),\dots,\log(x_{m-1}))$, and similarly for the exponential.
We write \( \vu \leq \vect{v}  \) if \( u_j \leq v_j \) for every coordinate \( j \), with similar notation for strict inequalities. 
We denote \( \vzero := (0, \ldots, 0) \) and \( \vone := (1, \dots, 1) \), where the vector’s size is understood from context. The diagonal matrix with diagonal elements $u_1, u_2, \ldots$ is denoted $\diag(\vu)$. Finally, for convenience, we set \( \vbeta := \vone - \valpha \), where \( \valpha \) is the vector of victory thresholds in the definition of an \( \valpha \)-winner.

\subsection{A Short Introduction to Analytic Combinatorics}

The techniques employed in this paper come from the field of analytic combinatorics \cite{flajolet2009analytic}. 
Within this framework, formal variables are introduced and linked to specific parameters that define a particular counting problem.
Once these variables are properly introduced, a formal power series is defined as a function of these variables, called the $\emph{generating function}$. It serves as a formal notation that succinctly represents the combinatorial structure of the problem. Furthermore, it can be interpreted as a function of complex variables, which can be studied with the powerful tools of complex analysis, hence giving insights into the original combinatorial problem.

To illustrate this, consider a classical combinatorial example: counting binary trees with \( k \) vertices. The problem is first encoded as a formal power series by introducing a variable \( z \) marking the number of vertices. The generating function \( T(z) \) is then defined as  
\[
T(z) := \sum\limits_{k=0}^{\infty} T_k z^k,
\]
where \( T_k \) is the number of binary trees with \( k \) vertices. The \emph{symbolic method} then translates the combinatorial structure into analytic properties.  
In this example, whose detailed analysis is beyond the scope of this paper, the combinatorial property that a binary tree is either an isolated leaf or a root with two binary subtrees translates into the equation  
\(
T(z) = z + z T(z)^2.
\)
To determine \( T_k \), we solve this equation and extract the coefficient of \( z^k \) in the Taylor expansion of \( T(z) \), which is denoted by the \emph{coefficient extraction}~\( [z^k] T(z) \). 

While this example is solved using an algebraic equation, other cases may involve different analytical techniques, such as differential equations or complex calculus.

\section{From Our Combinatorial Problem to Its Analytic Formulation}\label{sec_step-by-step}

As for binary trees, we first encode the problem of determining the probability that a candidate is an \(\valpha\)-winner. However, instead of a univariate infinite series, we now work with a multivariate polynomial.
We then show that the coefficient extractions can be achieved through \emph{Cauchy integrals}. Finally, to analyze their asymptotic behavior for large electorates, we introduce the general principle of the \emph{saddle-point method}. Its different sub-cases are later developed in \Cref{sec_subcritical-case,sec_critical-case,sec_mixed-case,sec_supercritical-case}.

\subsection{Symbolic Method}\label{sec_symbolic-method}

As a warm-up, consider the case \( m = 3 \).  
To determine whether candidate \( 3 \) is the Condorcet winner, or more generally an \( \valpha \)-winner, it suffices to know whether each voter ranks candidate \( 1 \) and/or \( 2 \) above candidate~\( 3 \), rather than their full ranking. Thus, we introduce formal variables \( x_1 \) (respectively \( x_2 \)) that indicates when a voter prefers candidate \( 1 \) (respectively \( 2 \)) over candidate \( 3 \).

To represent the probability distribution governing the preferences of a single voter, we introduce the characteristic polynomial \( P(x_1, x_2) \), as illustrated in \Cref{fig:symbolic-method}. In this polynomial, the probability of each ranking is multiplied by the appropriate formal variables, depending on whether candidate \( 1 \) and/or \( 2 \) is preferred over candidate \( 3 \). At this stage, the characteristic polynomial may appear to be merely a compact way of representing the probability distribution of a single voter.

\begin{figure}[ht!]
    \centering
    \begin{tikzpicture}[xscale=0.95,yscale=0.65]
        \coordinate (hg) at (-1,4.5) ;
        \coordinate (hd) at (11,4.5) ;
        
        \coordinate (a) at (-1,3.5)  ;
        \coordinate (b) at (11,3.5)  ;

        \coordinate (c) at (-1,1.3)  ;
	\coordinate (d) at (11,1.3)  ;
 
        \coordinate (bg) at (-1,0.3)  ; 
        \coordinate (bd) at (11,0.3)  ;

        \draw[thick] (hg) -- (hd) ;
	\draw[thick] (a) -- (b) ;
	\draw[thick] (c) -- (d) ;
        \draw[thick] (bg) -- (bd) ;

	\coordinate (a1) at (1,4.5)  ;
	\coordinate (a2) at (3,4.5)  ;
	\coordinate (a3) at (5,4.5)  ;
	\coordinate (a4) at (7,4.5)  ;
	\coordinate (a5) at (9,4.5)  ;
			
	\coordinate (b1) at (1,0.3)  ;
	\coordinate (b2) at (3,0.3)  ;
	\coordinate (b3) at (5,0.3)  ;
	\coordinate (b4) at (7,0.3)  ;
	\coordinate (b5) at (9,0.3)  ;
			
        \draw[thick] (hg) -- (bg) ;
	\draw (a1) -- (b1) ;
	\draw (a2) -- (b2) ;
	\draw (a3) -- (b3) ;
	\draw (a4) -- (b4) ;
	\draw (a5) -- (b5) ;
	\draw[thick] (hd) -- (bd) ;
		
	\node (p1) at (0,4) {$p_{123}$} ;
	\node (p2) at (2,4) {$p_{132}$} ;
	\node (p3) at (4,4) {$p_{213}$} ;
	\node (p4) at (6,4) {$p_{231}$} ;
	\node (p5) at (8,4) {$p_{312}$} ;
	\node (p6) at (10,4) {$p_{321}$} ;
			
	\node (a) at (0,3) {$1$};
	\node (b) at (0,2.4) {$2$};
	\node[draw, inner sep= 0pt, minimum width=4mm, minimum height=4mm] (c) at (0,1.8) {$3$};
			
	\node (a) at (2,3) {$1$};
	\node[draw, inner sep= 0pt, minimum width=4mm, minimum height=4mm] (b) at (2,2.4) {$3$};
	\node (c) at (2,1.8) {$2$};
			
	\node (a) at (4,3) {\strut $2$};
	\node (b) at (4,2.4) {\strut $1$};
	\node[draw, inner sep= 0pt, minimum width=4mm, minimum height=4mm] (c) at (4,1.8) {$3$};
			
	\node (b1) at (6,3) {$2$};
	\node[draw, inner sep= 0pt, minimum width=4mm, minimum height=4mm] (a1) at (6,2.4) {$3$};
	\node (c1) at (6,1.8) {$1$};
			
	\node[draw, inner sep= 0pt, minimum width=4mm, minimum height=4mm] (a2) at (8,3) {$3$};
	\node (b2) at (8,2.4) {$1$};
	\node (c2) at (8,1.8) {$2$};
			
	\node[draw, inner sep= 0pt, minimum width=4mm, minimum height=4mm] (a) at (10,3) {$3$};
	\node (b) at (10,2.4) {$2$};
	\node (c) at (10,1.8) {$1$};
			
			
	\node (x12) at (0,0.8) {\strut $x_1 x_2$};
	\node (x1) at (2,0.8) {\strut $x_1$};
	\node (x21) at (4,0.8) {\strut $x_1 x_2$};
	\node (x2) at (6,0.8) {\strut $x_2$};
	\node (1) at (8,0.8) {\strut $1$};
	\node (1) at (10,0.8) {\strut $1$};

        \begin{scope}[yshift=1.2cm]
	\node (poly) at (-2,-1.5) {\strut $P(x_1,x_2) \; :=$};

	\node (plus) at (1,-1.5) {\strut $+$};
        \node (plus) at (3,-1.5) {\strut $+$};
        \node (plus) at (5,-1.5) {\strut $+$};
        \node (plus) at (7,-1.5) {\strut $+$};
        \node (plus) at (9,-1.5) {\strut $+$};

        \node at (0,-1.5) {\strut $p_{123}\cdot x_1 x_2$};
       
        \node at (2,-1.5) {\strut $p_{132} \cdot x_1$};

        \node at (4,-1.5) {\strut $p_{213}\cdot x_1 x_2$};

        \node at (6,-1.5) {\strut $p_{231}\cdot x_2$};

        \node at (8,-1.5) {\strut $p_{312}\cdot 1$}; 

        \node at (10,-1.5) {\strut $p_{321} \cdot 1$}; 
        \end{scope}
		\end{tikzpicture}
    \caption{Definition of the characteristic polynomial \( P(x_1, x_2) \) encoding the probability distribution for a single voter when \( m = 3 \). The notation $p_{123}$, for example, is a shorthand for $p_{(1, 2, 3)}$. The formal variable $x_j$ marks rankings where candidate $j$ is preferred to candidate~$3$.}
    \label{fig:symbolic-method}
\end{figure}

Note that the characteristic polynomial can be rewritten as  
\begin{equation}\label{eq_polynomial-m-is-3}
  P(x_1,x_2) = p_{\emptyset} + p_{\{1\}} x_1 + p_{\{2\}} x_2 + p_{\{1,2\}} x_1 x_2,  
\end{equation}
where, for example, \( p_{\{1, 2\}} := p_{123} + p_{213} \) is the probability that a voter ranks both candidates \( 1 \) and \( 2 \) above candidate \( 3 \). 
This probability can also be expressed as a multivariate coefficient extraction, denoted by \( [{x_1}^{1}{x_2}^{1}]P(x_1, x_2) \).

More generally, for an arbitrary value of $m$ and a subset of adversaries $\mX \subseteq \mathcal{A}$, we overload the notation $p$ by defining $p_{{}_{\mX}}$ as the total probability that in a random ranking, the set of adversaries above $m$ is exactly $\mX$. We then encode the probability distribution for a single voter as follows.
\begin{definition}\label{def_def-of-P}
The \emph{characteristic polynomial} is defined as
    \[P(\vx) := \sum_{\mX \subseteq \mathcal{A}} \Big(p_{{}_{\mX}} \prod_{j \in \mX} x_j\Big).\]
\end{definition}
Note that $\sum_r p_r =1$ implies that $P(\vone)=1$.

Now, let us examine what happens with several voters, beginning with the case with $m=3$ candidates and $n=2$ voters.
The core observation, illustrated in \Cref{fig:treeproba}, is that the algebraic operation of developing $P(x_1, x_2)^{2}$ is isomorphic to the probability tree associated with the preferences of two voters. For example, the probability that, among the two voters, exactly two prefer candidate \( 1 \) over candidate \( 3 \), and exactly one prefers candidate \( 2 \) over candidate \( 3 \) is given by the coefficient extraction \( [{x_1}^2 {x_2}^{1}]P(x_1, x_2)^2 = 2 p_{\{1\}} p_{\{1, 2\}} \).

\begin{figure}[ht!]
\begin{center}
\begin{tikzpicture}[scale=0.43,yscale=0.8]





    \node (O2) at (0,-10) {1};
    \node (empty2) at (-12,-13) { $p_{\emptyset}$ };
    \node (a2) at (-4,-13) { $p_{\{1\}}x_1$ };
    \node (b2) at (4,-13) {$p_{\{2\}}x_2$ };
    \node (ab2) at (12,-13) {$p_{\{1, 2\}}x_1 x_2$ };

    \draw (O2) -- (empty2) ;
    \draw[very thick] (O2) -- (a2) ;
    \draw (O2) -- (b2) ;
    \draw[very thick] (O2) -- (ab2) ;

    \node (emptyT1a) at (-15,-16) { $ p_{\emptyset}$ };
    \node (emptyT1b) at (-13.5,-16.8) { $p_{\{1\}}x_1$ };
    \node (emptyT1c) at (-11.5,-16) { $p_{\{2\}}x_2$ };
    \node (emptyT1d) at (-9.5,-16.8) { $p_{\{1, 2\}}x_1 x_2$};

    \draw (empty2) -- (emptyT1a) ;
    \draw (empty2) -- (emptyT1b) ;
    \draw (empty2) -- (emptyT1c) ;
    \draw (empty2) -- (emptyT1d) ;

    \node (aT1a) at (-7,-16) { $ p_{\emptyset}$ };
    \node (aT1b) at (-5.5,-16.8) { $p_{\{1\}}x_1$ };
    \node (aT1c) at (-3.5,-16) { $p_{\{2\}}x_2$ };
    \node (aT1d) at (-1.5,-16.8) { $p_{\{1,2\}}x_1 x_2$ };

    \draw (a2) -- (aT1a) ;
    \draw (a2) -- (aT1b) ;
    \draw (a2) -- (aT1c) ;
    \draw[very thick] (a2) -- (aT1d) ;

    \node (bT1a) at (1,-16) { $p_{\emptyset}$ };
    \node (bT1b) at (2.5,-16.8) { $p_{\{1\}}x_1$ };
    \node (bT1c) at (4.5,-16) { $p_{\{2\}}x_2$ };
    \node (bT1d) at (6.5,-16.8) {$p_{\{1, 2\}}x_1 x_2$ };

    \draw (b2) -- (bT1a) ;
    \draw (b2) -- (bT1b) ;
    \draw (b2) -- (bT1c) ;
    \draw (b2) -- (bT1d) ;

    \node (abT1a) at (9,-16) { $p_{\emptyset}$ };
    \node (abT1b) at (10.5,-16.8) { $p_{\{1\}}x_1$ };
    \node (abT1c) at (12.5,-16) { $p_{\{2\}}x_2$ };
    \node (abT1d) at (14.5,-16.8) {$p_{\{1,2\}}x_1 x_2$ };

    \draw (ab2) -- (abT1a);
    \draw[very thick] (ab2) -- (abT1b);
    \draw (ab2) -- (abT1c) ;
    \draw (ab2) -- (abT1d) ;

\end{tikzpicture}
\caption{Tree representing the algebraic expansion of \( P(x_1, x_2)^2 \). The edges correspond to multiplications, and each path from the root to a leaf represents a term in the expansion. For example, the highlighted paths correspond to the terms involving \( {x_1}^2 {x_2}^{1} \), with a total coefficient given by \( [{x_1}^2 {x_2}^{1}]P(x_1, x_2)^2 = 2 p_{\{1\}} p_{\{1, 2\}} \). This coefficient extraction corresponds to the probability that, among the two voters, exactly two prefer candidate \( 1 \) over candidate \( 3 \), and exactly one prefers candidate \( 2 \) over candidate \( 3 \).}
\label{fig:treeproba}
\end{center}
\end{figure}
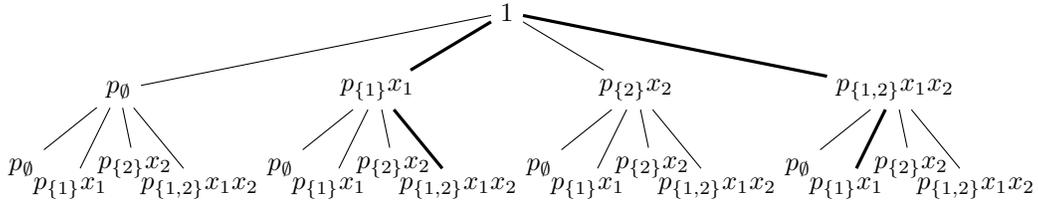

More generally, when $m$ and $n$ are arbitrary, for a vector $\vl \in \mathbb{N}^{m-1}$, the coefficient extraction 
\[ [\vx^{\vl}]P(\vx)^{n} := [{x_1}^{\ell_1}\cdots{x_{m-1}}^{\ell_{m-1}}]P(x_1,\dots, x_{m-1})^{n} \] corresponds to the probability that each adversary $j$ is preferred to candidate $m$ by exactly $\ell_j$ voters.

Now, for $m$ to be an $\valpha$-winner, every adversary $j$ must be preferred to $m$ by less than $(1-\alpha_j) n = \beta_j n$ voters. Summing over all such cases, the probability that candidate $m$ is an $\valpha$-winner is given~by
\[ \probamalphawinner = [\vx^{< \vbeta n}]P(\vx)^{n} := \sum_{\vl < \vbeta n} [\vx^{\vl}] P(\vx)^{n}. \]
However, the coordinates of \( \vbeta n \) are not necessarily integers, and we prefer to express the summation bounds in terms of integers:  
\begin{equation}\label{eq_alpha-winner-coeff-extraction}
    \probamalphawinner = [\vx^{\le \lceil \vbeta n \rceil-1}]P(\vx)^{n}  := \sum_{\vl \le \lceil \vbeta n \rceil-1} [\vx^{\vl}] P(\vx)^{n}.
\end{equation}

Note that if we define a \emph{weak $\valpha$-winner} as a candidate preferred to any other candidate~$j$ by \emph{at least} $\alpha_j n$ voters, then all the results of this paper extend to weak $\valpha$-winners by replacing  
$\lceil \vbeta n \rceil - 1$ with $\lfloor \vbeta n \rfloor$.
In particular, our findings for the Condorcet winner naturally extend to the \emph{weak Condorcet winner} (defined analogously) by replacing all occurrences of $\lceil n/2 \rceil - 1$ with $\lfloor n / 2 \rfloor$.

\subsection{Coefficient Extraction via Cauchy Integrals}

In analytic combinatorics, it is common to extract the coefficient \( f_\ell \) of a series \( F(z) = \sum_k f_k z^k \) with a positive radius of convergence by expressing it as a Cauchy integral \cite[Th. IV.4]{flajolet2009analytic}:  
\begin{equation}\label{eq_cauchy-formula}
    [z^\ell] F(z) = \frac{1}{2 i \pi} \oint F(z) \frac{dz}{z^{\ell+1}}.
\end{equation}
The symbol \( \oint \) indicates that the integral is taken over a positively oriented loop around \( 0 \) in the complex plane. 
To intuitively understand this equality, one can expand \( F(z) \) into its series, interchange sum and integral, and integrate over a circle of small radius centered at $0$. The residue theorem (see \eg \cite[Th. IV.3]{flajolet2009analytic}) then implies that all terms of the form \( z^{k-\ell-1} \) result in a zero integral unless \( k = \ell \).
Now, to extract \( [z^{\le L}]F(z) := f_0 +\dots + f_L \), we apply
\begin{equation}\label{eq_coeff-extraction-lower-than-L}
[z^{\leq L}] F(z) = [z^L] \frac{F(z)}{1-z},
\end{equation}
which is obtained by remarking that $\frac{F(z)}{1-z}= \left(\sum_k f_k z^k \right) \left( \sum_k z^k \right)=\sum_{k} \left(\sum_{\ell=0}^{k} f_{\ell} \right) z^k$.

\Cref{eq_cauchy-formula,eq_coeff-extraction-lower-than-L} extend to the multivariate case. Applying them to \Cref{eq_alpha-winner-coeff-extraction} gives

\begin{equation}\label{eq_alpha-winner-as-cauchy-integral}
     \probamalphawinner = \frac{1}{(2 i \pi)^{m-1}} \oint
		\frac{P(\vx)^n}{\prod_{j=1}^{m-1}(1-x_j)}
		\frac{d \vx}{ \prod_{j=1}^{m-1}x_j^{ \lceil \beta_j n \rceil }},
\end{equation}
where $\oint d\vx$ is a shorthand for $\oint \dots \oint d x_{1} \dots d x_{m-1}$.

\subsection{Saddle Point Method}\label{sec_saddle-point-method}

We study the asymptotic behavior of our complex integrals as $n \to +\infty$.
To build intuition, let us momentarily disregard the ceiling function in \Cref{eq_alpha-winner-as-cauchy-integral} and observe that
\[
\oint \frac{P(\vx)^n}{\prod_{j=1}^{m-1}(1-x_j)}
		\frac{d \vx}{ \prod_{j=1}^{m-1}x_j^{  \beta_j n }} = 
  \oint \frac{e^{-n\left(-\log(P(\vx)) + \vbeta^{\transp}\log(\vx)\right)}}{\prod_{j=1}^{m-1} \left(1-x_j\right)} d\vx =
  \oint A(\vx) e^{-n\psi(\log(\vx))} d\vx,
\]
for an appropriately defined function \( A(\vx) \) and $\psi(\vt) = - \log(P(e^{\vt})) + \vbeta^T \vt$, where $\vt=\log(\vx)$.


To analyze the asymptotic behavior of such integrals, it is standard to apply the \emph{saddle-point method} \cite[Chapter~VIII]{flajolet2009analytic}. The key idea is to find a contour of integration where, as $n$ approaches infinity, the integral's dominant contribution arises from the neighborhood of a specific point, while the contribution from the rest of the contour becomes negligible in comparison. 

For this approach to be valid, the chosen point must satisfy the condition that the gradient of the function~\( \psi \) vanishes. In the univariate case, since a holomorphic function cannot exhibit local extrema in modulus, the graph of the function's modulus at such a point resembles a saddle. This is why it is called a \emph{saddle point}, even in the general multivariate case.

In our case, the presence of the integer parts \( \lceil \beta_j n \rceil \) does not fundamentally alter the method. The necessary adaptations are detailed in \Cref{sec_appendix-Elie}. Moreover, \Cref{th_conditions-P-satisfied} in \Cref{sec_P_theorem_conditions} ensures the existence of a unique saddle point. The absolute value of $e^{-n \psi(\log(\vx))}$ at the saddle point is minimal on the real line and maximal on the integration circle.
These considerations lead to define the following objects, which we will use extensively in the rest of the paper. 
\begin{definition}\label{def_K-tau-etc}
 The \emph{cumulant generating function} of $P$  is
    \[
K: \vt \in  \mathbb{R}^{m-1} \mapsto \log\left(P(e^{\vt})\right).
\]
Observing that the function $\psi: \vt \mapsto - K(\vt)+ \vbeta^{\transp} \vt$ is strictly concave on $\reals^{m-1}$ (an immediate consequence of \Cref{th_convexity-cumulant-gf} in  \Cref{sec_appendix-Elie-preliminaries}), we define the \emph{log saddle point} $\vtau$ and the \emph{saddle point} $\vzeta$ as
\[
\vtau := \argmax_{\vt \in \mathbb{R}^{m-1}} \big( -K(\vt) + \vbeta^{\transp} \vt \big), \qquad \qquad \vzeta = e^{\vtau}.
\]
Remark that this implies that $\vzeta$ minimizes the function $\vx \in (\mathbb{R}_{>0})^{m-1} \mapsto \frac{P(\vx)}{\prod_{j}{x_j}^{\beta_j}}$.

\smallskip
To approximate the integral in \Cref{eq_alpha-winner-as-cauchy-integral}, a key ingredient is the Taylor expansion of \( \psi \) at its saddle point. Consequently, we will need the Hessian of \( \psi \) at \( \vtau \). Since the term \( \vbeta^{\transp} \vt \) is linear, this Hessian simplifies to \( \mathcal{H}_{K}(\vtau) \), the Hessian of $K(\vt)$ at $\tau$, which can be computed either directly or as
\begin{equation}\label{eq_hessian-of-K-in-tau}
     \mH_K(\vtau) =
     \diag(\vzeta) \frac{\mH_P(\vzeta)}{P(\vzeta)} \diag(\vzeta)
     + \diag(\vbeta)
     - \vbeta \vbeta^T.
\end{equation}
\end{definition}

\subsubsection*{Probabilistic interpretation.}

There is a nice probabilistic interpretation for the cumulant generating function. 
For any vector parameter $\vs \in \reals^{m-1}$,
we define the random vector $\vX^{(\vs)}$ as follows. For any subset $\mX \subseteq \mA$ of adversaries whose indicator vector is denoted $\vl \in \{0,1\}^{m-1}$:
\[
    \proba(\vX^{(\vs)} = \vl) :=
    \frac{[\vx^{\vl}] P(\vx) e^{\vl^T \vs}}{P(e^{\vs})} = \frac{p_{{}_\mX}e^{\vl^T \vs}}{P(e^{\vs})}.
\]
Note that $\vX^{(\vzero)}$ corresponds to the original probability distribution, as encoded by $P$.
For $\vs \neq \vzero$, the distribution of $\vX^{(\vs)}$ represents a perturbation of the original culture.

The \emph{cumulant generating function} of $\vX^{(\vs)}$ is classically defined as
\[
K_{\vs}: \vt \mapsto    \log(\mean(e^{\vt^T \vX^{(\vs)}})).
\]
Remark that $K = K_{\vzero}$.
It is well known that
the gradient and Hessian matrix of this function at $\vt = \vzero$
are equal to the mean and covariance matrix of $\vX^{(\vs)}$.
Actually, all the information contained in $K_{\vs}$ can be retrieved through $K$ thanks to the following observation.
\begin{align*}
    K_{\vs}(\vt) &=
    \log \left(
    \sum_{\vl} 
    \frac{\big([\vx^{\vl}] P(\vx)\big) e^{\vl^T \vs}}{P(e^{\vs})}
    e^{\vl^T \vt}
    \right)
    =
    \log \left(
    \sum_{\vl} 
    \big([\vx^{\vl}] P(\vx)\big) e^{\vl^T (\vs + \vt)}
    \right)
    - \log(P(e^{\vs}))
    \\&=
    \log(P(e^{\vs + \vt})) - \log(P(e^{\vs}))
    =
    K(\vs + \vt) - \log(P(e^{\vs})),
\end{align*}
so the mean and covariance matrix of $\vX^{(\vs)}$ are respectively equal to the gradient and the Hessian of $K(\vt)$
at $\vt = \vs$.

Since \( \vtau \) is the vector where the gradient of \( \psi: \vt \mapsto - K(\vt) + \vbeta^T \vt \) vanishes, it follows that the gradient of \( K \) at \( \vtau \), \ie the mean of $X^{(\vtau)}$, is equal to \( \vbeta \). In other words, the distribution of $X^{(\vtau)}$, \ie the perturbation of the original culture induced by $\vtau$, is such that, in expectation, candidate \( m \) is precisely at the threshold for being an \( \valpha \)-winner.

\subsubsection*{Criticality.}
In \Cref{eq_alpha-winner-as-cauchy-integral}, the terms of the form \( \frac{1}{1 - x_j} \) require special precautions. If \( \vzeta_j < 1 \), there is no issue, as we can integrate over a circle of radius \( \zeta_j \) that loops around \( 0 \) without enclosing the singularity at \( 1 \). In this case, we say that the coordinate \( \vzeta_j \) is \emph{subcritical}.
However, if \( \vzeta_j \) is \emph{critical}, i.e., if \( \vzeta_j = 1 \), then the integration path must be adjusted to slightly bypass the singularity. 
Finally, if \( \vzeta_j \) is \emph{supercritical}, i.e., if \( \vzeta_j > 1 \), then a circle of radius \( \zeta_j \) necessarily encloses a singularity, and the issue remains even with slight deformations of the integration path, requiring the singularity at~1 to be explicitly accounted for in the analysis.

In \Cref{sec_subcritical-case}, we analyze the case where all coordinates are subcritical, while \Cref{sec_critical-case} focuses on the case where all coordinates are critical. \Cref{sec_mixed-case} extends the analysis to the mixed case, where each coordinate is either subcritical or critical. To address supercritical coordinates, we introduce the necessary techniques in \Cref{sec_supercritical-case}, which makes it possible to tackle the most general setting. Finally, \Cref{sec_error-terms} provides higher-order terms in the asymptotic expansion of the previous formulas. Together, these sections establish the theoretical asymptotic behavior of an $\valpha$-winner in the GIC, followed by an application to the Condorcet winner in a particular culture.

\section{Subcritical Case}\label{sec_subcritical-case}

In this section, we study the case where all coordinates of the saddle point are subcritical.
We then illustrate this scenario with \( \mathcal{M}_{3 \textnormal{ last}} \), a Mallows culture where the reference ranking is \( (1,2,3) \), making candidate \( 3 \) particularly unlikely to be the Condorcet winner. 

\subsection{Theoretical Result in the Subcritical Case}\label{sec_subcritical-theoretical-result}

Since \( P(\vx) \) has strictly positive coefficients and \( \vbeta \in (0,1)^{m-1} \), it follows from \Cref{sec_P_theorem_conditions} that all the assumptions of \Cref{th_large-powers-limit} in \Cref{sec_appendix-large-powers} are satisfied, allowing its application. Recall that \( P(\vx) \) is introduced in \Cref{def_def-of-P}, and that \( \vtau \), \( \vzeta \), and \( \hessian_{K}(\vtau) \) are defined in \Cref{def_K-tau-etc}.  

\begin{theorem}\label{th_subcritical}

Assume that all coordinates of the saddle point are subcritical, \ie $\zeta_j < 1$
 for every adversary $j\in\mathcal{A}$. Then:
\[ \probamalphawinner \underset{n \to +\infty}{\sim} \frac{P(\vzeta)^n}{\prod\limits_{j\in \mathcal{A}}\big( (1-\zeta_j){\zeta_j}^{\lceil\beta_j n \rceil -1}\big)} \frac{1}{\sqrt{(2\pi n)^{m-1}\det( \hessian_{K}(\vtau) ) }}.\]
\end{theorem}
To interpret this asymptotic behavior, recall that $\vzeta$ is the unique global minimizer of \mbox{$ \vx \mapsto \frac{P(\vx)}{\prod_{j}{x_j}^{\beta_j}} $}. Hence, $\frac{P(\vzeta)}{\prod_{j}{\zeta_j}^{\beta_j}} < \frac{P(\vone)}{\prod_{j}1^{\beta_j}} = 1$. 
Thus, \Cref{th_subcritical} indicates that, in the subcritical case, $\probamalphawinner$ tends to $0$ exponentially fast.

\subsection{Application: Mallows Culture $\mM_{3 \text{ last}}$}\label{sec_mallows-last-3}




In the model \( \mM_{3 \textnormal{ last}} \), defined in \Cref{sec_definitions-and-notations}, candidate \( 3 \) is ranked last among the three options in the reference ranking \( (1,2,3) \), which can be interpreted as representing the hidden truth. We now analyze the probability of the \emph{a priori} undesirable event in which candidate \( 3 \) nevertheless emerges as the Condorcet winner.

To find this probability, the first step is to determine the saddle point. We provide the general formula for an arbitrary number of candidates \( m \). The proof is detailed in \Cref{sec_appendix-saddle-point-Mallows}.

\begin{lemma}\label{th_saddle-point-mallows}
    Under the Mallows culture $\mathcal{M}_{m \; \textnormal{last}}$, the log saddle point $\vtau$ is given by
    
    \[ \vtau = \left(\frac{-m\rho}{2},\frac{(-m+2)\rho}{2}, \dots, \frac{(m-4)\rho}{2} \right),\]
    where $\rho$ is the concentration parameter of the culture, defined in \Cref{sec_definitions-and-notations}.
\end{lemma}

In the case of three candidates, \Cref{th_saddle-point-mallows} yields \( \vtau = (\frac{-3\rho}{2}, \frac{-\rho}{2}) \), ensuring that both coordinates of the saddle point are subcritical and allowing the application of \Cref{th_subcritical}.
Standard algebraic calculations then lead to the following result.

\begin{theorem}\label{th_mallows-3-last}

Under the Mallows culture $\mathcal{M}_{3 \; \textnormal{last}}$, the probability that candidate $3$ is the Condorcet winner has asymptotic behavior
\[  \probacondorcetwinner{3} \underset{n \to +\infty}{\sim} \frac{P(\vzeta)^n}{\prod\limits_{j \in \mathcal{A}}\big( (1-\zeta_j){\zeta_j}^{\lceil n /2 \rceil -1}\big)} \frac{1}{2\pi n \sqrt{ \det(\mathcal{H}_{K}(\vtau))}}  , \]
 where  $\vzeta = (e^{-3\rho/2},e^{-\rho/2})$ , $P(\vzeta) =2 \frac{e^{-2\rho} ( 1 + e^{-\rho/2} + e^{-\rho} )}{(1 + e^{-\rho}) (1 + e^{-\rho} + e^{-2 \rho})} $ and \(\det(\mathcal{H}_{K}(\vtau))= \frac{1}{4}\frac{e^{-\rho/2}(1+e^{-\rho})}{(1+e^{-\rho/2} + e^{-\rho})^{2}}\).



\end{theorem}

As we saw in \Cref{sec_subcritical-theoretical-result}, this probability tends exponentially fast to 0. This can be viewed as a generalization of Condorcet's Jury Theorem \cite{condorcet1785essai}, which gives the same result for 2 candidates.

As mentioned in \Cref{sec_symbolic-method}, the counterpart of \Cref{th_mallows-3-last} for the weak Condorcet winner is obtained by replacing $\lceil n / 2 \rceil - 1$ with $\lfloor n / 2 \rfloor$.  
Notably, for even \( n \), the former simplifies to \( n/2 - 1 \), while the latter becomes \( n/2 \).  
Thus, although the limit remains 0 in both cases, the rates of convergence for even $n$ differ by a factor $\prod_{j \in \mA} \zeta_j$.

\begin{figure}
\begin{center}
\begin{tikzpicture}

\definecolor{darkgrey176}{RGB}{176,176,176}
\definecolor{darkorange25512714}{RGB}{255,127,14}
\definecolor{forestgreen4416044}{RGB}{44,160,44}
\definecolor{lightgrey204}{RGB}{204,204,204}
\definecolor{steelblue31119180}{RGB}{31,119,180}

\begin{axis}[
height=\axisHeight,
legend cell align={left},
legend style={font=\legendFont, fill opacity=1, draw opacity=1, text opacity=1, draw=lightgrey204},
log basis y={10},
tick align=outside,
tick pos=left,
width=\axisWidth,
x grid style={darkgrey176},
xlabel={Number of voters $n$},
xmin=-2.9, xmax=104.9,
xtick style={color=black},
y grid style={darkgrey176},
ylabel={$\mathbb{P}(3 \text{ is CW})$},
ymin=6.81017998419311e-11, ymax=1.03610728324648,
ymode=log,
ymode=log
]
\addplot [semithick, steelblue31119180]
table {%
2 0.318377900147517
3 0.356923675606596
4 0.112538466534534
5 0.151396097922717
6 0.0530393030591472
7 0.0764495125879401
8 0.0281220610612632
9 0.0420357040903004
10 0.0159046899984387
11 0.0243139895202693
12 0.00936983703068005
13 0.0145443270844925
14 0.00567771151715875
15 0.00891115095588899
16 0.00351212084043297
17 0.00555858451360971
18 0.00220701356460382
19 0.00351598981440965
20 0.00140422139765173
21 0.00224889919622051
22 0.000902466264109894
23 0.00145160856714447
24 0.000584831141213638
25 0.000944116145704879
26 0.000381642259523833
27 0.000618001805627908
28 0.000250530064255592
29 0.000406765007373475
30 0.000165304509451117
31 0.000269009714618354
32 0.000109557989084586
33 0.000178650317525268
34 7.28959144313867e-05
35 0.000119079708455935
36 4.86707097271569e-05
37 7.96328314106272e-05
38 3.25967767550739e-05
39 5.34093569966333e-05
40 2.18920453013265e-05
41 3.59158521009081e-05
42 1.47395801767964e-05
43 2.42097135206563e-05
44 9.9464889033124e-06
45 1.63543642730483e-05
46 6.7259358679223e-06
47 1.10697139229621e-05
48 4.55677296524094e-06
49 7.50630574631133e-06
50 3.09254977780783e-06
51 5.09847423261003e-06
52 2.1021874266335e-06
53 3.46834754979806e-06
54 1.4310970706217e-06
55 2.36278314858573e-06
56 9.75580026644568e-07
57 1.61175777297763e-06
58 6.6590305671862e-07
59 1.10080575682392e-06
60 4.55067506474211e-07
61 7.52698314380084e-07
62 3.11331785224552e-07
63 5.15226632218343e-07
64 2.13217698992352e-07
65 3.53031592097754e-07
66 1.46166337643353e-07
67 2.42125315675497e-07
68 1.00292963203401e-07
69 1.66208789537495e-07
70 6.88762367790171e-08
71 1.14191243787684e-07
72 4.73394304634349e-08
73 7.85156746438805e-08
74 3.25620605909007e-08
75 5.40264010540462e-08
76 2.24139336640093e-08
77 3.72018597464788e-08
78 1.5439212565924e-08
79 2.56339988036467e-08
80 1.0641867247809e-08
81 1.76744743654238e-08
82 7.33976316074336e-09
83 1.21938685237068e-08
84 5.06529434549527e-09
85 8.41760901295968e-09
86 3.49762783841792e-09
87 5.81401967726542e-09
88 2.41644798242074e-09
89 4.01785091551573e-09
90 1.67034310305083e-09
91 2.77798890599895e-09
92 1.15517676502139e-09
93 1.92166210300195e-09
94 7.99275486434468e-10
95 1.32991671210002e-09
96 5.53275180793296e-10
97 9.20798046524168e-10
98 3.83154932561016e-10
99 6.37806592166329e-10
100 2.65453570520244e-10
};
\addlegendentry{Theoretical equivalent}
\addplot [semithick, darkorange25512714]
table {%
2 0.0204081632653061
3 0.0709426627793975
4 0.0148549215604609
5 0.0405408681127132
6 0.0100421120829284
7 0.0245072310412306
8 0.00668640809480337
9 0.0152689892503339
10 0.00444048358533242
11 0.00969884782921564
12 0.00295277182264903
13 0.00624550051432034
14 0.00196855464858603
15 0.00406331911971191
16 0.00131622771504663
17 0.00266500544901309
18 0.000882617168925034
19 0.00175933366081399
20 0.00059347720061254
21 0.00116773167632655
22 0.000400072854528966
23 0.000778600722594852
24 0.000270326017924772
25 0.000521169040119847
26 0.0001830479725695
27 0.000350032980162291
28 0.000124191473790703
29 0.000235789233724621
30 8.44106827999181e-05
31 0.000159248214502529
32 5.74669468794744e-05
33 0.000107804519011099
34 3.9182934816963e-05
35 7.31319427532099e-05
36 2.6753649455854e-05
37 4.97046085182039e-05
38 1.82907206581063e-05
39 3.38399645673519e-05
40 1.25198616536082e-05
41 2.3074985142508e-05
42 8.57931357951837e-06
43 1.57570287171806e-05
44 5.88514101315197e-06
45 1.07740434417253e-05
46 4.04094324083337e-06
47 7.37582611049548e-06
48 2.77717482256044e-06
49 5.05511713985725e-06
50 1.91026585918475e-06
51 3.46821252652827e-06
52 1.3150197217302e-06
53 2.38178785001188e-06
54 9.0593754040482e-07
55 1.63717635111279e-06
56 6.24558605635246e-07
57 1.12630941809685e-06
58 4.30863834507065e-07
59 7.75473689490309e-07
60 2.97428873526482e-07
61 5.34321730745216e-07
62 2.05441424630015e-07
63 3.68422528596477e-07
64 1.41984650216692e-07
65 2.54202529899124e-07
66 9.81817881980387e-08
67 1.75504444840158e-07
68 6.79274362216725e-08
69 1.21242931313528e-07
70 4.70190282445318e-08
71 8.38052895063749e-08
72 3.25616461121752e-08
73 5.79589877750009e-08
74 2.25597426751533e-08
75 4.01044991938038e-08
76 1.56368201567465e-08
77 2.77637268431481e-08
78 1.08427946747929e-08
79 1.92293655723161e-08
80 7.52150960077239e-09
81 1.33243344432521e-08
82 5.21954620627371e-09
83 9.23657449487707e-09
84 3.62341289700349e-09
85 6.40550276956675e-09
86 2.51625147845193e-09
87 4.44390665789378e-09
88 1.74797635857953e-09
89 3.08417539134257e-09
90 1.21466593796596e-09
91 2.14125844281652e-09
92 8.44331015749118e-10
93 1.48712965329395e-09
94 5.8708139402276e-10
95 1.03317172094003e-09
96 4.08327786118141e-10
97 7.18016894115643e-10
98 2.84079701088237e-10
99 4.99148929327845e-10
100 1.97691483195951e-10
};
\addlegendentry{Exact results}
\addplot [semithick, forestgreen4416044]
table {%
2 0.0188
3 0.0692
4 0.0147
5 0.0351
6 0.0096
7 0.023
8 0.0059
9 0.0167
10 0.0042
11 0.0111
12 0.0028
13 0.0071
14 0.0022
15 0.0043
16 0.0013
17 0.0037
18 0.0011
19 0.0023
20 0.0004
21 0.0005
22 0.0005
23 0.0003
24 0.0001
25 0.0004
26 0.0001
27 0.0001
28 0.0001
29 0.0002
};
\addlegendentry{Monte-Carlo results}
\end{axis}

\end{tikzpicture}
\caption{Probability that candidate \(3\) is the Condorcet winner as a function of \(n\) in a culture \( \mM_{3 \textnormal{ last}} \) with \(\rho = \log(2)\), shown on a semilog scale. The theoretical equivalent is based on \Cref{th_mallows-3-last}, while exact results rely on \Cref{eq_alpha-winner-coeff-extraction}. Monte Carlo simulations use 10,000 profiles per point, yielding an error of order \(10^{-2}\). For \( n \geq 30 \), they return zero, which is not visible due to the logarithmic scale.}\label{fig_mallows-3-last}
\end{center}
\end{figure}
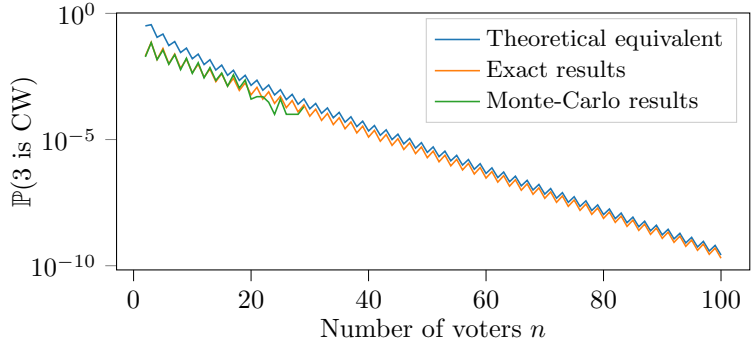

\paragraph{Numerical simulations}  
To validate our results numerically, we developed the Python package \texttt{Actinvoting}, available at \url{https://github.com/francois-durand/actinvoting}.  
This package notably compares our theoretical results with Monte Carlo simulations, which estimate probabilities by generating a large number of random profiles. Ranking samples in the Mallows model are obtained using the algorithms of \cite{doignon2004repeated} and \cite{lu2014effective}.

\Cref{fig_mallows-3-last} illustrates the case of \( \mM_{3 \textnormal{ last}} \), showing the probability that candidate~3 is the Condorcet winner as a function of the number of voters \(n\). We provide three estimations: the theoretical approximation given by \Cref{th_mallows-3-last}, the exact result obtained via coefficient extraction from \Cref{eq_alpha-winner-coeff-extraction}, and a Monte Carlo estimate based on 10,000 profiles per value of \(n\). The concentration parameter is set to \(\rho = \log(2)\). This choice balances two opposing constraints observed in practice. Lower concentration values require larger \(n\) for the theoretical approximation to be accurate, making direct comparison with exact results difficult, as their computational cost grows exponentially in \(n\) due to the algebraic expansion of \(P(\vx)^n\). Conversely, higher concentration values cause probabilities to decay rapidly, rendering Monte Carlo estimation unreliable.

The computational costs of these methods vary significantly. For instance\footnote{All of our numerical simulations were performed on an 11th Gen Intel Core i9-11980HK processor (8 cores, 16 logical processors) with 64 GB of RAM.}, at \(n=100\), Monte Carlo simulations take 46 seconds, the exact computation 17 seconds, and the theoretical approximation only 338~\(\mu\)s. As previously mentioned, the cost of the exact computation grows exponentially in \(n\). Monte Carlo simulations are also computationally expensive, as their error scales as \(\mathcal{O}(1/\sqrt{N})\), where \(N\) is the sample size. In contrast, the complexity of the theoretical approximation is polynomial in $\log(n)$, making it essentially constant time in practice.

A striking feature of \Cref{fig_mallows-3-last} is the non-monotonicity of the curves, which stems from parity effects in the number of voters \(n\). For small \(n\), the exact and Monte Carlo curves are nearly indistinguishable, while the theoretical approximation slightly overestimates the probability. When probabilities drop significantly below the order of magnitude of the Monte Carlo error, \(\frac{1}{\sqrt{N}} = 10^{-2}\), the Monte Carlo results naturally become more unstable. Below \(\frac{1}{N} = 10^{-4}\), they become undetectable due to sample size limitations. As expected, for a large number of voters \(n\), the theoretical approximation converges to the exact result, while being much faster computationally.

\section{Critical Case}\label{sec_critical-case}

In this section, we study the case where all coordinates of the saddle point are critical, which is exemplified by the notion of Condorcet winner in the Impartial Culture.

\subsection{Theoretical Result in the Critical Case}

When \( \vzeta =\vone \), it means that the expected proportion of voters who prefer candidate~$m
$ to candidate~$j$ is exactly \( \alpha_j \), as explained in the probabilistic interpretation of \Cref{sec_saddle-point-method}. In other words, in expectation, candidate $m$ is precisely at the threshold of being an $\valpha$-winner.

As soon as one coordinate \( \zeta_j \) is equal to 1, the term \( \frac{1}{1 - x_j} \) in \Cref{eq_alpha-winner-as-cauchy-integral} prevents integration along a circle of radius \( \zeta_j \). \Cref{th_large-powers-singularity} in \Cref{sec_appendix-Elie-singularity} is based on the idea that this difficulty can be circumvented by choosing an integration path that slightly deforms around the singularity at \( 1 \). Applying this result to the particular case where all coordinates are critical, we obtain the following theorem.

\begin{theorem}\label{th_critical}
Assume that all coordinates of the saddle point are critical, \ie $\zeta_j=1$ for every adversary $j \in \mathcal{A}$. Then:
	\[ \lim_{n \to +\infty}\probamalphawinner =
	\frac{1}{\sqrt{(2 \pi)^{m-1} \det(\mathcal{H}_{K}(\vtau))}}
	\int_{(0, + \infty)^{m-1}}
	e^{- \vu^{\transp} \mathcal{H_{K}(\vtau)}^{-1} \vu / 2}
	d \vu.
	\]
\end{theorem}
In this case, \Cref{eq_hessian-of-K-in-tau} simplifies to $\hessian_{K}(\vtau) = \hessian_{P}(\vone) + \diag(\vbeta) - \vbeta \vbeta^{\transp}$ because $\vzeta = \vone = P(\vone)$.

For the critical case, \Cref{th_critical} generalizes to $\valpha$-winners the result of \cite{niemi1968paradox} and \cite[Equation (12)]{krishnamoorthy2005condorcet} on Condorcet winners. While we derive it result via analytic combinatorics, their approach relies on a standard probabilistic method: the Gaussian approximation. This method allows for interpreting the vector \( \vu \) as the standardized outcome of each pairwise comparison between candidate \( m \) and their adversaries. Integrating over the positive orthant corresponds to requiring that each comparison outcome exceeds its expectation~\( \alpha_j n \).

The two approaches differ in the matrix involved. In their case, it is the correlation matrix of the pairwise comparisons against each adversary, denoted \( R_m \) by both \cite{niemi1968paradox} and  \cite{krishnamoorthy2005condorcet}.
\Cref{sec_link-matrices} recalls its definition and establishes the relation \( R_m = 4 D \mH_{K}(\vtau) D \), where $D$ is a diagonal matrix, which is the identity matrix when considering the notion of Condorcet winner. Since the result of \Cref{th_critical} is divided by the square root of the determinant, this rescaling has no impact.
However, our analytic combinatorics  approach makes it easier to deal with non-critical cases and extends naturally to the computation of an asymptotic expansion
(see \cref{sec_error-terms}). For instance, in the subcritical case, their formula only establishes that the limit is zero, whereas \Cref{th_subcritical} additionally provides the convergence rate.

\subsection{Application: Impartial Culture}
\label{sec_impartial_culture}

Computing the limiting probability for a candidate to be the Condorcet Winner in Impartial Culture
corresponds to applying \Cref{th_critical} with $\vbeta = \vone / \vtwo$.
We obtain the following result, proved in \Cref{sec_constants_IC}.

\begin{theorem}\label{th_impartial_culture}
Under Impartial Culture with $m$ candidates, the probability that candidate $m$ is the Condorcet winner has the following limit:
\[
    \lim_{n \to +\infty}
    \probacondorcetwinner{m}
    =
    \frac{1}{\sqrt{(2 \pi)^{m-1} \det(\mH_K(\vzero))}}
    \int_{(0, +\infty)^{m-1}}
    e^{- \vu \mH_K(\vzero)^{-1} \vu / 2}
    d \vu
\]
where,
letting $I$ denote the identity matrix and $J$ the matrix where all elements are $1$s, both of dimension $(m-1) \times (m-1)$,
\begin{equation}
\label{eq:mHKIC}
    \mH_K(\vzero) = \frac{1}{6} \left(I + \frac{1}{2} J \right),
    \qquad
    \mH_K(\vzero)^{-1} = 6 \left(I - \frac{1}{m+1} J \right),
    \qquad
    \det(\mH_K(\vzero)) = \frac{m+1}{2} \frac{1}{6^{m-1}}.
\end{equation}
\end{theorem}

We show in \cref{sec_link-matrices}
that this result is the same as Equation~(10) from \cite{niemi1968paradox} and Equation~(12) from \cite{krishnamoorthy2005condorcet}.
For completeness, we recall a simpler univariate integral formulation
introduced by \cite{ruben1954moments} to study moments of Gaussian distributions,
and linked to the voting setting first implicitly by \cite{niemi1968paradox},
then explicitly by \cite{may1971some} (Equation~(7)).

\begin{theorem}[\cite{may1971some}]\label{th:impartial}
Under Impartial Culture,
\begin{equation} \label{eq_int_erf_vote}
\lim_{n \to +\infty}\probacondorcetwinner{m} = 
\frac{1}{\sqrt{\pi}}
\int_{-\infty}^{+\infty}
e^{-u^2}
(1 - \Phi(u))^{m-1}
d u
\end{equation}
where $\Phi(u) =
\frac{1}{2} \left(
1 + \erf \left( \frac{u}{\sqrt{2}} \right)
\right)$ and $\erf$ denotes the error function: \(
\erf(u) = \frac{2}{\sqrt{\pi}} \int_{0}^{u} e^{-v^2} dv.
\)
\end{theorem}

\begin{figure}
\begin{center}
\begin{tikzpicture}

\definecolor{darkgrey176}{RGB}{176,176,176}
\definecolor{darkorange25512714}{RGB}{255,127,14}
\definecolor{forestgreen4416044}{RGB}{44,160,44}
\definecolor{lightgrey204}{RGB}{204,204,204}
\definecolor{steelblue31119180}{RGB}{31,119,180}

\begin{axis}[
height=\axisHeight,
legend cell align={left},
legend style={font=\legendFont, 
  fill opacity=1,
  draw opacity=1,
  text opacity=1,
  at={(0.97,0.03)},
  anchor=south east,
  draw=lightgrey204
},
tick align=outside,
tick pos=left,
width=\axisWidth,
x grid style={darkgrey176},
xlabel={Number of voters $n$},
xmin=-2.9, xmax=104.9,
xtick style={color=black},
y grid style={darkgrey176},
ylabel={$\mathbb{P}(3 \text{ is CW})$},
ymin=0, ymax=0.4,
ytick={0.0, 0.1, 0.2, 0.3, 0.4},
ytick style={color=black}
]
\addplot [semithick, steelblue31119180]
table {%
2 0.304086723984094
3 0.304086723984094
4 0.304086723984094
5 0.304086723984094
6 0.304086723984094
7 0.304086723984094
8 0.304086723984094
9 0.304086723984094
10 0.304086723984094
11 0.304086723984094
12 0.304086723984094
13 0.304086723984094
14 0.304086723984094
15 0.304086723984094
16 0.304086723984094
17 0.304086723984094
18 0.304086723984094
19 0.304086723984094
20 0.304086723984094
21 0.304086723984094
22 0.304086723984094
23 0.304086723984094
24 0.304086723984094
25 0.304086723984094
26 0.304086723984094
27 0.304086723984094
28 0.304086723984094
29 0.304086723984094
30 0.304086723984094
31 0.304086723984094
32 0.304086723984094
33 0.304086723984094
34 0.304086723984094
35 0.304086723984094
36 0.304086723984094
37 0.304086723984094
38 0.304086723984094
39 0.304086723984094
40 0.304086723984094
41 0.304086723984094
42 0.304086723984094
43 0.304086723984094
44 0.304086723984094
45 0.304086723984094
46 0.304086723984094
47 0.304086723984094
48 0.304086723984094
49 0.304086723984094
50 0.304086723984094
51 0.304086723984094
52 0.304086723984094
53 0.304086723984094
54 0.304086723984094
55 0.304086723984094
56 0.304086723984094
57 0.304086723984094
58 0.304086723984094
59 0.304086723984094
60 0.304086723984094
61 0.304086723984094
62 0.304086723984094
63 0.304086723984094
64 0.304086723984094
65 0.304086723984094
66 0.304086723984094
67 0.304086723984094
68 0.304086723984094
69 0.304086723984094
70 0.304086723984094
71 0.304086723984094
72 0.304086723984094
73 0.304086723984094
74 0.304086723984094
75 0.304086723984094
76 0.304086723984094
77 0.304086723984094
78 0.304086723984094
79 0.304086723984094
80 0.304086723984094
81 0.304086723984094
82 0.304086723984094
83 0.304086723984094
84 0.304086723984094
85 0.304086723984094
86 0.304086723984094
87 0.304086723984094
88 0.304086723984094
89 0.304086723984094
90 0.304086723984094
91 0.304086723984094
92 0.304086723984094
93 0.304086723984094
94 0.304086723984094
95 0.304086723984094
96 0.304086723984094
97 0.304086723984094
98 0.304086723984094
99 0.304086723984094
100 0.304086723984094
};
\addlegendentry{Theoretical equivalent}
\addplot [semithick, darkorange25512714]
table {%
2 0.111111111111111
3 0.314814814814815
4 0.148148148148148
5 0.310185185185185
6 0.169581618655693
7 0.308327617741198
8 0.183965858862978
9 0.30734128118173
10 0.194465480998831
11 0.306728641149723
12 0.202566406713069
13 0.306310528605833
14 0.209065535158771
15 0.306006806614097
16 0.214432326840295
17 0.305776127352619
18 0.218963456454917
19 0.305594952029487
20 0.222856822586211
21 0.305448882480342
22 0.226250199713895
23 0.305328613342689
24 0.229242879851164
25 0.305227862088046
26 0.231908492866705
27 0.305142233320958
28 0.234302962190104
29 0.305068559518588
30 0.236469637388801
31 0.305004500170059
32 0.238442716436943
33 0.304948288365469
34 0.240249593171544
35 0.304898565199271
36 0.24191250760723
37 0.304854268418318
38 0.243449731441944
39 0.304814555654353
40 0.244876436057957
41 0.304778750324587
42 0.246205338932428
43 0.304746302756819
44 0.247447192404867
45 0.304716761764114
46 0.248611158347832
47 0.304689753532354
48 0.249705098964687
49 0.304664965715857
50 0.250735805056098
51 0.30464213530132
52 0.251709177062998
53 0.304621039237945
54 0.252630370023913
55 0.304601487125072
56 0.253503910657579
57 0.304583315448787
58 0.254333792697468
59 0.304566382997766
60 0.255123555100848
61 0.304550567186113
62 0.25587634665643
63 0.304535761080526
64 0.25659497970299
65 0.304521870979254
66 0.257281975065402
67 0.304508814426946
68 0.257939599857595
69 0.304496518576493
70 0.258569899454288
71 0.30448491882911
72 0.25917472466649
73 0.304473957699034
74 0.259755754949281
75 0.304463583860719
76 0.260314518309419
77 0.304453751345188
78 0.260852408453892
79 0.304444418859005
80 0.26137069962064
81 0.304435549204567
82 0.261870559453116
83 0.304427108784563
84 0.262353060216757
85 0.304419067176684
86 0.262819188604172
87 0.304411396767217
88 0.263269854334393
89 0.304404072434249
90 0.263705897717798
91 0.304397071272792
92 0.264128096330704
93 0.304390372355518
94 0.264537170920994
95 0.304383956523825
96 0.264933790647399
97 0.304377806204839
98 0.2653185777396
99 0.304371905250683
100 0.265692111653369
};
\addlegendentry{Exact results}
\addplot [semithick, forestgreen4416044]
table {%
2 0.1136
3 0.3166
4 0.1488
5 0.3064
6 0.1692
7 0.3052
8 0.1903
9 0.3071
10 0.1975
11 0.3077
12 0.208
13 0.3039
14 0.2082
15 0.2997
16 0.2147
17 0.3033
18 0.22
19 0.3039
20 0.2173
21 0.3006
22 0.228
23 0.3086
24 0.225
25 0.2976
26 0.229
27 0.3033
28 0.2299
29 0.3027
30 0.2324
31 0.2977
32 0.2358
33 0.299
34 0.2386
35 0.3003
36 0.2392
37 0.2976
38 0.2403
39 0.3026
40 0.2468
41 0.2984
42 0.2369
43 0.3066
44 0.2511
45 0.3051
46 0.2456
47 0.3004
48 0.2472
49 0.2991
50 0.2515
51 0.301
52 0.2458
53 0.3011
54 0.2493
55 0.304
56 0.2465
57 0.2975
58 0.2515
59 0.2964
60 0.2534
61 0.3027
62 0.2562
63 0.3044
64 0.2499
65 0.2987
66 0.2523
67 0.2983
68 0.2501
69 0.3007
70 0.2578
71 0.2943
72 0.2578
73 0.2967
74 0.2592
75 0.2959
76 0.2564
77 0.2965
78 0.2587
79 0.2989
80 0.2625
81 0.3036
82 0.2564
83 0.2994
84 0.2597
85 0.2964
86 0.2575
87 0.2999
88 0.2597
89 0.295
90 0.2578
91 0.3003
92 0.2583
93 0.302
94 0.2632
95 0.304
96 0.2628
97 0.3033
98 0.2547
99 0.2978
100 0.2621
};
\addlegendentry{Monte-Carlo results}
\end{axis}

\end{tikzpicture}
\caption{Probability that a given candidate is the Condorcet winner as a function of \(n\) in the Impartial Culture with three candidates. The theoretical equivalent is based on \Cref{th_impartial_culture}, while exact results rely on \Cref{eq_alpha-winner-coeff-extraction}. Monte Carlo simulations use 10,000 profiles per point, yielding an error of order \(10^{-2}\). When the exact results curve is not visible, it is overlapped by the Monte Carlo results curve.}\label{fig_ic-3}
\end{center}
\end{figure}
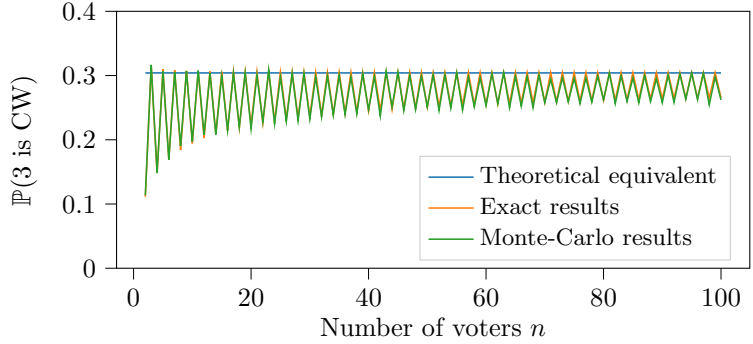

\paragraph{Numerical simulations}
\Cref{fig_ic-3} illustrates our results for the Impartial Culture with three candidates. Computational times vary significantly: for \(n=100\), Monte Carlo simulations take 5 seconds (faster than for the Mallows model due to quicker random profile generation in IC), the exact computation takes 22 seconds, and the theoretical approximation requires only 358~\(\mu\)s.  As before, the Monte Carlo and exact result curves exhibit sawtooth patterns due to parity effects in the number of voters \(n\) and remain nearly indistinguishable. In contrast, the theoretical approximation, which is just a numerical limit here, is simply represented as a horizontal line.
Notably, for odd \(n\), the exact values closely match the limit, whereas for even \(n\), they are significantly lower, although they converge to the same value. In \Cref{sec_error-terms}, we present the next term in the asymptotic expansion of the probability, which improves the fit between the theoretical approximation and the exact curve, and explains why the limit approximation is already very accurate for an odd number of voters~\(n\).

\section{Mixed Case: Subcritical and Critical}\label{sec_mixed-case}

We study the case where some of the coordinates are subcritical and some are critical, combining the ideas developed in \Cref{sec_subcritical-case,sec_critical-case}. 

\subsection{Theoretical Result in the Mixed Case}

As in the critical case, we apply \Cref{th_large-powers-singularity} from \Cref{sec_appendix-Elie-singularity}, now in its more general form.

\begin{theorem}\label{th_mixed-case}
Let $\mathcal{S}:=\{j \in \mathcal{A}, \; \zeta_j < 1\}$ and $\mathcal{C}:=\{j \in \mathcal{A}, \; \zeta_j = 1\}$ respectively denote the subset of subcritical and critical coordinates of the saddle point.
Assume $\mA = \mS \cup \mC$.
Then:
\[  \probamalphawinner
\underset{n \to +\infty}{\sim}
\frac{P(\vzeta)^n}{\prod\limits_{j\in \mathcal{S}}\big( (1-\zeta_j){\zeta_j}^{\lceil\beta_j n \rceil -1}\big)}
 \frac{1}{  \sqrt{(2\pi)^{m-1}n^{|\mathcal{S}|}\det(\mathcal{H}_{K}(\vtau) ) }}
 \int_{(0, + \infty)^{|\mathcal{C}|}}
	e^{- \vu^{\transp} M \vu / 2}
	d \vu,
	\]
where $M$ is the submatrix of  $ {\hessian_{K}(\vtau)}^{-1}$ that corresponds to the rows and columns of $\mathcal{C}$.
\end{theorem}

As in the subcritical case, this result establishes that the probability converges exponentially fast to zero whenever at least one coordinate is subcritical. The Gaussian integral associated with the critical coordinates is just a multiplicative constant in this asymptotic behavior.

\subsection{Application: Mallows Culture $\mM_{4 \textnormal{ last}}$}\label{sec_mallows-last-4}

In \Cref{sec_mallows-last-3}, we established that for the Mallows culture \( \mM_{3 \textnormal{ last}} \), all coordinates of the saddle point are subcritical. In the case of 4 candidates, \ie for \( \mM_{4 \textnormal{ last}} \), \Cref{th_saddle-point-mallows} yields \( \vzeta = (e^{-2\rho}, e^{-\rho}, 1) \). The first two coordinates are subcritical, while the last one is critical. We then apply \Cref{th_mixed-case}. Since the Gaussian integral is univariate in that case, we compute it explicitly.

\begin{theorem}\label{th_mallows-4-last}
Under the Mallows culture $\mathcal{M}_{4 \; \textnormal{last}}$, the probability that candidate $4$ is the Condorcet Winner has asymptotic behavior
\[\probacondorcetwinner{4} \underset{n\to +\infty}{\sim}
 \frac{P(\vzeta)^n}{\prod\limits_{j\in \mathcal{S}}\big( (1-\zeta_j){\zeta_j}^{\lceil n/2 \rceil -1}\big)}
 \frac{1}{ 4\pi n \sqrt{ \det(\mathcal{H}_{K}(\vtau) )  M_{33} }},
 \]
 where $M_{33}$ is the bottom right coefficient of $\mathcal{H}_{K}(\vtau)^{-1}$.
\end{theorem}

In this expression, computing \( P(\vzeta) \) by hand remains relatively straightforward and yields the elegant result:  
\(
P(\vzeta) = 8 \gamma e^{-3\rho} (1+e^{-\rho})\left(\frac{1}{2} + e^{-\rho}\right)^2
\).
However, the expressions of \( \mathcal{H}_{K}(\vtau) \) and its inverse being large, we omit them in this paper. More generally, as \( m \) increases, the computation of \( P(\vzeta) \) and \( \mathcal{H}_{K}(\vtau) \) is best left to the computer. The package \texttt{Actinvoting} provides symbolic algorithms to efficiently compute these quantities.

\begin{figure}
\begin{center}
\begin{tikzpicture}

\definecolor{darkgrey176}{RGB}{176,176,176}
\definecolor{darkorange25512714}{RGB}{255,127,14}
\definecolor{forestgreen4416044}{RGB}{44,160,44}
\definecolor{lightgrey204}{RGB}{204,204,204}
\definecolor{steelblue31119180}{RGB}{31,119,180}

\begin{axis}[
height=\axisHeight,
legend cell align={left},
legend style={font=\legendFont, fill opacity=1, draw opacity=1, text opacity=1, draw=lightgrey204},
log basis y={10},
tick align=outside,
tick pos=left,
width=\axisWidth,
x grid style={darkgrey176},
xlabel={Number of voters $n$},
xmin=-0.4, xmax=52.4,
xtick style={color=black},
y grid style={darkgrey176},
ylabel={$\mathbb{P}(4 \text{ is CW})$},
ymin=1.51644497674266e-11, ymax=0.104243331462616,
ymode=log,
ymode=log
]
\addplot [semithick, steelblue31119180]
table {%
2 0.0281915573674291
3 0.0372307551265095
4 0.0069142830949232
5 0.0109574924475735
6 0.00226106987013422
7 0.00383920435092178
8 0.000831827609366385
9 0.00146472078517002
10 0.000326423489266461
11 0.000587844898419254
12 0.000133431461069767
13 0.000243988957384717
14 5.61008623102275e-05
15 0.000103724260982465
16 2.40788462995008e-05
17 4.48932215600777e-05
18 1.04988380579653e-05
19 1.97030825508882e-05
20 4.63491560958989e-06
21 8.74433058317186e-06
22 2.06684177420426e-06
23 3.91630143219821e-06
24 9.29344546212115e-07
25 1.76734779988224e-06
26 4.20797095210057e-07
27 8.02705711074423e-07
28 1.91666465705525e-07
29 3.6658965394712e-07
30 8.77487616114693e-08
31 1.68218824103093e-07
32 4.03524917342539e-08
33 7.75140505002004e-08
34 1.86294272630734e-08
35 3.5849608055634e-08
36 8.63046119857855e-09
37 1.66344642097699e-08
38 4.01061518040067e-09
39 7.74112390375748e-09
40 1.86892848533574e-09
41 3.61195934912737e-09
42 8.73094935639405e-10
43 1.68933717779531e-09
44 4.08804970730554e-10
45 7.91827109444663e-10
46 1.91809051356161e-10
47 3.71880243845136e-10
48 9.01662019799119e-11
49 1.74969749381427e-10
50 4.24593258498929e-11
};
\addlegendentry{Theoretical equivalent}
\addplot [semithick, darkorange25512714]
table {%
2 0.00444444444444444
3 0.0185871936076018
4 0.00194968145988554
5 0.00657844628230301
6 0.000848050684255617
7 0.00256714404209404
8 0.000372482357068678
9 0.00105315543865573
10 0.000165153102248213
11 0.000445481677633842
12 7.38488102494915e-05
13 0.000192435778451525
14 3.3268335336217e-05
15 8.44216602252766e-05
16 1.50845951872274e-05
17 3.74800764416713e-05
18 6.87838802560402e-06
19 1.67988296560971e-05
20 3.1519530659138e-06
21 7.58821219530236e-06
22 1.45061309949923e-06
23 3.45003763483218e-06
24 6.70166394809406e-07
25 1.57726541116229e-06
};
\addlegendentry{Exact results}
\addplot [semithick, forestgreen4416044]
table {%
2 0.0055
3 0.021
4 0.0022
5 0.0074
6 0.0007
7 0.0027
8 0.0006
9 0.0014
10 0.0002
11 0.0003
13 0.0004
};
\addlegendentry{Monte-Carlo results}
\end{axis}

\end{tikzpicture}
\caption{Probability that candidate \(4\) is the Condorcet winner as a function of \(n\) in a culture \( \mM_{4 \textnormal{ last}} \) with \(\rho = \log(2)\), shown on a semilog scale. The theoretical equivalent is based on \Cref{th_mallows-4-last}, while exact results rely on \Cref{eq_alpha-winner-coeff-extraction}. Monte Carlo simulations use 10,000 profiles per point, yielding an error of order \(10^{-2}\). For \( n \geq 14 \), they return zero, which is not visible due to the logarithmic scale. Exact results are not computed for \( n \geq 26 \) due to prohibitive runtime (e.g., 7 minutes for $n = 25$).}\label{fig_mallows-4-last}
\end{center}
\end{figure}
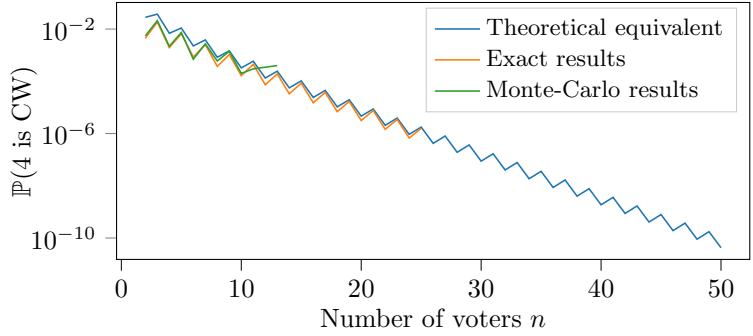

\paragraph{Numerical simulations}
\Cref{fig_mallows-4-last} illustrates our results for \( \mM_{4 \textnormal{ last}} \). 
The exact computation requires eliciting the polynomial \(P(\vx)\), which has \(2^{m-1}\) terms corresponding to all possible subsets of adversaries. This is followed by the algebraic expansion of \(P(\vx)^n\), resulting in a computational cost of \(\bigO(2^{(m-1)n})\). For \(m=4\), the runtime quickly becomes prohibitive: for instance, for \(n = 25\), the exact computation takes over 7 minutes. Therefore, we did not perform this calculation for larger values of \(n\).
In contrast, Monte Carlo simulations remain manageable, taking 24 seconds for \(n = 25\). However, as usual, their limitation lies in the precision of the calculation: in particular, when the exact probabilities are small, Monte Carlo simulations often return an empirical probability of zero. Finally, the theoretical approximation fits the exact curve well, while remaining computationally inexpensive: it takes only 437~\(\mu\)s for \(n = 25\).

\section{Dealing With Supercriticality}\label{sec_supercritical-case}

To compute \( \proba(m \succ_{\valpha} \mathcal{A}) \), our general approach relies on calculating the characteristic polynomial~\( P(\vx) \) and using the saddle point method.  
In \Cref{sec_subcritical-case,sec_critical-case,sec_mixed-case}, the singularity at 1 arising from the terms \( \frac{1}{1-x_j} \) in \Cref{eq_alpha-winner-as-cauchy-integral} was not an issue, as each coordinate \( \zeta_j \) of the saddle point was either subcritical or critical. Specifically, when \( \zeta_j \) is subcritical, the corresponding integration contour is simply a circle that passes through \( \zeta_j \), and hence does not encircle 1. When $\zeta_j$ is critical, the integration contour for that coordinate can be chosen to slightly bypass the singularity at 1.  

We now consider the case where at least one coordinate \( \zeta_j \) is supercritical, \ie \( \zeta_j > 1 \). In this scenario, the integration contour, which must pass sufficiently close to \( \zeta_j \), inevitably encloses the singularity at 1. In general, the contribution from this singularity is non-negligible and can be determined via a residue calculation. However, rather than relying solely on calculus, we provide an equivalent, more intuitive interpretation based on the complement rule.

\subsection{Theoretical Result for Supercriticality}\label{sec_supercritical-theoretical-result}

Assume that the coordinate \( \zeta_j \) is supercritical.
The complement rule gives
\begin{equation}\label{eq_complement-rule-particular}
 \proba\big(m \succ_{\valpha} \mathcal{A}\big) = \proba\big(m \succ_{\valpha} \mathcal{A}\setminus\{j\}\big) - \proba\big(m \succ_{\valpha} \mathcal{A}\setminus\{j\} \land m \preccurlyeq_{\valpha} j\big).   
\end{equation}

In the first term on the right-hand side, the event \( \{ m \succ_{\valpha} \mathcal{A} \setminus \{j\} \} \) means that \( m \) is an \( \valpha \)-winner after removing adversary \( j \), reducing the competition to \( m-2 \) adversaries. As we will elaborate, the associated culture arises by setting \( x_j = 1 \) in \( P(\vx) \), allowing the application of the saddle point method in a lower-dimensional space.

The second term is the probability that \( m \) defeats all candidates in \( \mathcal{A} \setminus \{j\} \) but not \( j \) in the sense of \( \valpha \). As we will see, this probability is expressed by using a modified characteristic polynomial and a new saddle point derived from \( \vzeta \), which is modified solely by inverting its \( j \)-th coordinate, thereby making it subcritical.

At this stage, the saddle point of the new characteristic polynomial may still have supercritical coordinates. Thus, the process might need to be iterated, using the following generalization of \Cref{eq_complement-rule-particular}. Let \( \mathcal{X} \) and \( \mathcal{Y} \) be disjoint sets of adversaries, and let \( j \in \mathcal{X} \). Then,
\begin{equation}\label{eq_complement-rule}
     \proba\big(m \succ_{\valpha} \mathcal{X} \land m \preccurlyeq_{\valpha} \mathcal{Y}  \big) = \proba\big(m \succ_{\valpha} \mathcal{X} \setminus \{j\} \land m \preccurlyeq_{\valpha} \mathcal{Y}  \big) - \proba\big(m \succ_{\valpha} \mathcal{X}\setminus \{j\}  \land m \preccurlyeq_{\valpha} \mathcal{Y}\cup \{j\}  \big).
\end{equation}
Once again, the first term reduces the dimensionality of the problem, while the second term enables the inversion of a saddle point coordinate, which will be used in practice to transform a supercritical coordinate into a subcritical one.  
By iterating this process, we systematically reduce the problem to terms of the form  
\(
\proba\big(m \succ_{\valpha} \mathcal{X} \land m \preccurlyeq_{\valpha} \mathcal{Y}\big),
\) 
for well-chosen sets \( \mathcal{X} \) and \( \mathcal{Y} \).

To compute such terms, we define a new characteristic polynomial 
\( P_{\mathcal{Y}}^{\mathcal{X}}(\vx_{{}_\mathcal{X}},\vy_{{}_\mathcal{Y}}) \), with variables \( \vx_{{}_{\mX}} = (x_j)_{j \in \mX} \) and \( \vy_{{}_{\mY}} = (y_j)_{j \in \mY} \),  
derived from the original polynomial \( P(\vx) \) through the following transformations.
 For each adversary \( j \in \mathcal{X} \), the formal variable \( x_j \) keeps its interpretation and does not lead to a modification in the polynomial. 
  For each adversary \( j \in \mathcal{Y} \), instead of using \( x_j \), which encodes when \( j \) is preferred to \( m \), we introduce a new formal variable \( y_j \) that encodes when \( m \) is preferred to \( j \). This variable will appear in all monomials that do not contain \( x_j \) in \( P(\vx) \). 
  Finally, for each adversary \( j \in \mathcal{A} \setminus (\mathcal{X} \cup \mathcal{Y}) \), we disregard the outcome of the pairwise comparison, hence we eliminate the corresponding variable by setting \( x_j =1 \).  
This corresponds to the algebraic operation:
\[
P_{\mathcal{Y}}^{\mathcal{X}}\big(\vx_{{}_\mathcal{X}},\vy_{{}_\mathcal{Y}}\big) = \Big(\prod_{j \in \mathcal{Y}} y_j\Big) P\Big(\vx_{{}_\mathcal{X}},\frac{\vone}{\vy_{{}_\mathcal{Y}}},\vone_{{}_{\mathcal{A}\setminus (\mathcal{X} \cup \mathcal{Y})}}\Big),
\]
where $\frac{\vone}{\vy_{{}_\mathcal{Y}}} :=({y_j}^{-1})_{j \in \mY}$.
As an example, consider \( m=3 \), \( \mX=\{1\} \), and \( \mY=\{2\} \). From the expression of $P$ in \Cref{eq_polynomial-m-is-3}, we deduce:
\[
P_{\{2\}}^{\{1\}}(x_1,y_2)= p_{\emptyset} y_2 + p_{\{1\}} x_1 y_2 + p_{\{2\}} + p_{\{1,2\}} x_1.
\]
Note that $P_{\mY}^{\mX}\big(\vx_{{}_\mathcal{X}},\vy_{{}_\mathcal{Y}}\big)$ generalizes $P(\vx)$, since $P(\vx)=P_{\emptyset}^{\mA}\big(\vx_{{}_\mathcal{A}},\vy_{{}_\emptyset}\big)$.
As in \Cref{eq_alpha-winner-coeff-extraction}, the probability of interest is expressed as a coefficient extraction:
\[
 \proba\big(m \succ_{\valpha} \mathcal{X} \land m \preccurlyeq_{\valpha} \mathcal{Y} \big) = \Big[ \prod_{j \in \mathcal{X}}{x_j}^{< \beta_j n}\prod_{j \in \mathcal{Y}}{y_j}^{\le \alpha_j n}\Big] P_{\mathcal{Y}}^{\mathcal{X}}\big(\vx_{{}_\mathcal{X}},\vy_{{}_\mathcal{Y}}\big).
\]
Given two disjoint sets of adversaries \( \mathcal{X} \) and \( \mathcal{Y} \), \Cref{def_K-tau-etc} generalizes to the cumulant generating function $\tilde{K}(\vt_{{}_{\mX \cup \mY}})=\log\big(P_{\mathcal{Y}}^{\mathcal{X}}(\exp(\vt_{{}_{\mX \cup \mY}}))\big)$, the log saddle point $\tilde{\vtau} = \argmax_{\vt_{{}_{\mX \cup \mY}} \in \mathbb{R}^{|\mX \cup \mY|}} \big( -K(\vt_{{}_{\mX \cup \mY}}) + (\vbeta_{{}_{\mX}},\valpha_{{}_{\mY}})^{\transp} \vt_{{}_{\mX \cup \mY}} \big) $, and the saddle point $\tilde{\vzeta}= e^{\tilde{\vtau}}$.

As previously mentioned, transferring an adversary from \( \mathcal{X} \) to \( \mathcal{Y} \) results in the inversion of the corresponding coordinate in the saddle point. The following lemma formalizes this process and is proved in \Cref{sec_appendix-flip-flop}.
\begin{lemma}\label{th_supercritical-becomes-subcritical}
Let $\mX$ and $\mY$ be two disjoint sets of adversaries. Let $\tilde{\vzeta}$ be the saddle point associated with $\proba\big(m \succ_{\valpha} \mathcal{X} \land m \preccurlyeq_{\valpha} \mathcal{Y} \big)$ . Let $j \in \mX$, $\mX' = \mX \setminus \{j\}$, and $\mY' = \mY \cup \{j\}$. 
Then the saddle point associated with $\proba\big(m \succ_{\valpha} \mathcal{X}' \land m \preccurlyeq_{\valpha} \mathcal{Y}' \big)$ has its $k$-th coordinate equal to $\tilde{\zeta}_k$ if $k \neq j$ and $\frac{1}{\tilde{\zeta}_j}$ if $k=j$.
\end{lemma}
This lemma, combined with the iterated application of \Cref{eq_complement-rule}, allows expressing the probability $\probamalphawinner$ as a sum of terms of the form $\proba\big(m \succ_{\valpha} \mathcal{X} \land m \preccurlyeq_{\valpha} \mathcal{Y} \big)$, where all the coordinates of all saddle points are either subcritical or critical.
The following theorem provides the asymptotic behavior of such terms. It directly follows from \Cref{th_large-powers-singularity} in \Cref{sec_appendix-Elie-singularity}.

\begin{theorem}\label{th_proba-X-Y}
Let $\mX$ and $\mY$ be disjoint sets of adversaries.
Let $\tilde{P}\big(\vx_{{}_\mathcal{X}},\vy_{{}_\mathcal{Y}}\big) :=P_{\mathcal{Y}}^{\mathcal{X}}\big(\vx_{{}_\mathcal{X}},\vy_{{}_\mathcal{Y}}\big)$ and $\tilde{K}$ its cumulant generating function. Let $\tilde{\vtau}$ and $\tilde{\vzeta}$ respectively be the log saddle point and saddle point associated with $\proba\big(m \succ_{\valpha} \mathcal{X} \land m \preccurlyeq_{\valpha} \mathcal{Y}\big)$.
We assume that (1) $\mathcal{X}=\mathcal{S} \cup \mathcal{C}$, where $  \tilde{\zeta}_j <1, \  \forall j \in \mathcal{S}$ and $\tilde{\zeta}_j =1, \ \forall j \in \mC$, and (2) $\tilde{\zeta}_j < 1, \ \forall j \in \mY$.
Then    
\begin{align*}
    & \proba\big(m \succ_{\valpha} \mathcal{X} \land m \preccurlyeq_{\valpha} \mathcal{Y}\big) \underset{n \to +\infty}{\sim} \\
    & \frac{\tilde{P}\big(\tilde{\vzeta}\big)^n}{\prod\limits_{j \in \mathcal{S}} {\tilde{\zeta}_{j}}^{\lceil \beta_j n \rceil -1} \prod\limits_{j \in \mathcal{Y}} {\tilde{\zeta}_{j}}^{\lfloor \alpha_j n \rfloor}}
 \frac{1}{ \prod\limits_{j\in \mathcal{S}\cup\mathcal{Y}}\big(1-\tilde{\zeta}_j\big)}
 \frac{1}{  \sqrt{(2\pi)^{|\mathcal{X} \cup \mathcal{Y}|}n^{|\mathcal{S} \cup \mathcal{Y}|}\det\big(\mathcal{H}_{\tilde{K}}(\tilde{\vtau}) \big) }}
 \int_{(0, + \infty)^{|\mathcal{C}|}}
	e^{- \vu^{\transp} \tilde{M} \vu / 2}
	d \vu,
\end{align*}
where $\tilde{M}$ is the submatrix of  $\mathcal{H}_{\tilde{K}}(\tilde{\vtau})^{-1}$ that corresponds to the rows and columns of $\mathcal{C}$.
\end{theorem}

In summary, to handle supercriticality, the process consists in applying \Cref{eq_complement-rule} to one of the supercritical coordinates of the original saddle point, and then iterating this step on each term whose saddle point still contains supercritical coordinates, until every term satisfies the assumptions of \Cref{th_proba-X-Y}. The procedure is guaranteed to terminate, as each iteration either reduces the dimension by one or decreases the number of supercritical coordinates. Once \Cref{th_proba-X-Y} is used to compute the asymptotics of each term, some may be negligible compared to others, leading to a simplified sum with fewer terms, as illustrated in the following application.

%
%
%
%
%
%
%
%
%
%

\subsection{Application:  Mallows Culture $\mM_{3 \textnormal{ first}}$}

Analogously to \Cref{th_saddle-point-mallows} for the culture \( \mM_{m \textnormal{ last}} \), \Cref{th_saddle-point-mallows-m-first} provides the expression of the saddle point for \( \mM_{m \textnormal{ first}} \). The proof is similar and is provided in \Cref{sec_appendix-saddle-point-Mallows}.
\begin{lemma}\label{th_saddle-point-mallows-m-first}
     Under the culture $\mathcal{M}_{m \; \textnormal{first}}$, the log saddle point $\vtau$ is given by
    \[ \vtau = \left(\frac{m\rho}{2},\frac{(m-2)\rho}{2}, \dots, \frac{(-m+4)\rho}{2} \right),\]
    where $\rho$ is the concentration parameter of the culture, defined in \Cref{sec_definitions-and-notations}.
\end{lemma}

In the case of three candidates, \Cref{th_saddle-point-mallows-m-first} gives \( \vzeta = \left(e^{3 \rho/2}, e^{\rho/2}\right) \), hence both coordinates are supercritical. Consequently, we apply the procedure outlined in \Cref{sec_supercritical-theoretical-result}. In this specific case, we show in \Cref{sec_appendix-mallows-3-first} that this approach effectively reduces to applying the inclusion-exclusion principle and identifying the dominant term in the sum. Unsurprisingly, the dominant term that may prevent candidate~3 from being the Condorcet winner corresponds to the event of not winning the pairwise comparison against the strongest adversary, candidate~2.

\begin{theorem}\label{th_mallows-3-first}
Under $\mathcal{M}_{3 \; \textnormal{first}}$,
the probability that candidate $3$ is the Condorcet Winner has asymptotic behavior
\[ 1-\probacondorcetwinner{3} \underset{n \to +\infty}{\sim}  \sqrt{\frac{2}{\pi n}} \frac{e^{-\rho \lceil n/2 \rceil} }{1-e^{-\rho}}\left(\frac{2}{1+e^{-\rho}} \right)^n .
\]
\end{theorem}

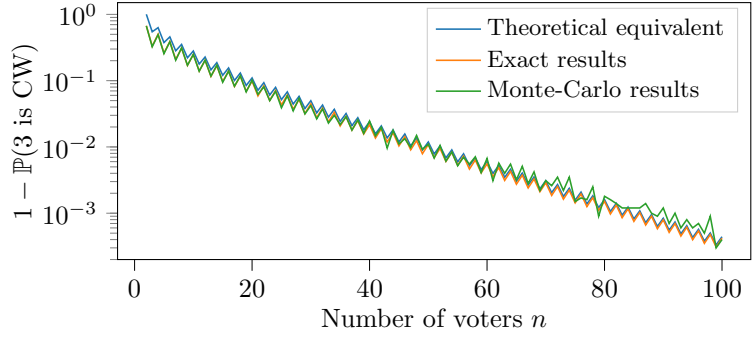
\begin{figure}
\begin{center}
\begin{tikzpicture}

\definecolor{darkgrey176}{RGB}{176,176,176}
\definecolor{darkorange25512714}{RGB}{255,127,14}
\definecolor{forestgreen4416044}{RGB}{44,160,44}
\definecolor{lightgrey204}{RGB}{204,204,204}
\definecolor{steelblue31119180}{RGB}{31,119,180}

\begin{axis}[
height=\axisHeight,
legend cell align={left},
legend style={font=\legendFont, fill opacity=1, draw opacity=1, text opacity=1, draw=lightgrey204},
log basis y={10},
tick align=outside,
tick pos=left,
width=\axisWidth,
x grid style={darkgrey176},
xlabel={Number of voters $n$},
xmin=-2.9, xmax=104.9,
xtick style={color=black},
y grid style={darkgrey176},
ylabel={$1 - \mathbb{P}(3 \text{ is CW})$},
ymin=0.000199945756126065, ymax=1.50491371787716,
ymode=log,
ymode=log
]
\addplot [semithick, steelblue31119180]
table {%
2 1.0030037040849
3 0.54596606336211
4 0.630427307301029
5 0.375914217065934
6 0.457548214164234
7 0.282404888858856
8 0.352220779486604
9 0.221384623688702
10 0.280031859918226
11 0.17799993481926
12 0.227229283699761
13 0.145543222757181
14 0.186998611649919
15 0.120438548108148
16 0.15548549702142
17 0.100562059211133
18 0.13030500660149
19 0.0845530534053881
20 0.109882829702396
21 0.0714897743482394
22 0.093128148730381
23 0.0607207508388358
24 0.0792563680796717
25 0.051770042843041
26 0.0676862714554361
27 0.0442806639576121
28 0.0579769984024744
29 0.0379790839039169
30 0.0497876460126507
31 0.0326520245579622
32 0.0428503830786369
33 0.0281307602249866
34 0.0369519802418614
35 0.0242801795681305
36 0.0319207732162323
37 0.0209909711630102
38 0.0276172437052266
39 0.0181739181622974
40 0.02392707509025
41 0.0157556534852648
42 0.0207559416231289
43 0.0136754491470288
44 0.0180255375554556
45 0.0118827523864761
46 0.0156705097876791
47 0.0103352706224549
48 0.0136360598814602
49 0.00899746610656791
50 0.0118760493552406
51 0.00783936072624078
52 0.0103514884456004
53 0.00683557857149762
54 0.00902932054934311
55 0.00596457284912256
56 0.00788143712216856
57 0.00520799719594536
58 0.00688387394988034
59 0.00455019114877953
60 0.00601615142130029
61 0.00397775661979774
62 0.00526073004543612
63 0.00347920747313564
64 0.00460255886675716
65 0.00304467822650364
66 0.00402869925851936
67 0.00266568087404519
68 0.00352801024478961
69 0.00233490109862255
70 0.00309088432071467
71 0.00204602689436732
72 0.00270902492469328
73 0.0017936039840374
74 0.00237525842242425
75 0.00157291348526937
76 0.00208337480610324
77 0.00137986812487742
78 0.00182799237711265
79 0.00121092397264078
80 0.00160444253062673
81 0.00106300520428834
82 0.00140867144322078
83 0.000933439836911192
84 0.00123715601594497
85 0.000819904731128475
86 0.00108683187302494
87 0.000720378440132164
88 0.000955031581683426
89 0.000633100719502288
90 0.000839431558069237
91 0.000556537703740399
92 0.00073800637085741
93 0.00048935191391483
94 0.000648989357908807
95 0.000430376392033983
96 0.000570838640493053
97 0.000378592366833856
98 0.000502207760384769
99 0.000333109946628292
100 0.000441920282760544
};
\addlegendentry{Theoretical equivalent}
\addplot [semithick, darkorange25512714]
table {%
2 0.673469387755102
3 0.331389698736637
4 0.50242954324587
5 0.260696050367021
6 0.389037636434554
7 0.208879919246569
8 0.308371752058135
9 0.169813258899794
10 0.248553942361773
11 0.139656794391601
12 0.202930755039498
13 0.115940040947108
14 0.167401154871251
15 0.0970035811682791
16 0.139273613761614
17 0.0816924651121835
18 0.116703380507902
19 0.0691800092325987
20 0.0983852490119846
21 0.0588609308161647
22 0.0833726824698777
23 0.0502835505815715
24 0.0709652832800871
25 0.0431052262880784
26 0.0606358406630061
27 0.0370622211402186
28 0.0519815192048608
29 0.0319488610153653
30 0.0446903732361986
31 0.0276028389130062
32 0.0385179001543008
33 0.0238946813981454
34 0.0332703362326289
35 0.020720086975191
36 0.0287925747013836
37 0.0179942786421006
38 0.0249593077052145
39 0.0156477893345395
40 0.0216684505521267
41 0.0136232798954782
42 0.0188362031636876
43 0.0118731099190063
44 0.0163933001585521
45 0.0103574636850584
46 0.0142821335727698
47 0.00904288971173228
48 0.0124545230038901
49 0.00790115166346217
50 0.0108699709479579
51 0.00690831596104236
52 0.00949428529841079
53 0.00604402107362301
54 0.00829848231787023
55 0.00529088756729734
56 0.00725790583091801
57 0.004634038192597
58 0.00635151459201877
59 0.0040607047490997
60 0.00556130158794477
61 0.00355990395657702
62 0.0048718177042264
63 0.00312216863949299
64 0.00426977860401578
65 0.00273932358379647
66 0.00374373845752718
67 0.00240429772808626
68 0.00328381776243347
69 0.00211096610315342
70 0.0028814752258236
71 0.00185401627721604
72 0.00252931576372484
73 0.00162883510234724
74 0.00222092827933651
75 0.0014314123662682
76 0.00195074812604523
77 0.00125825858843098
78 0.0017139401341012
79 0.00110633470130306
80 0.00150629884561559
81 0.000972991757656527
82 0.00132416320968065
83 0.000855919125455906
84 0.00116434347412664
85 0.000753099891005693
86 0.00102405839998976
87 0.000662772401563871
88 0.000900881239843443
89 0.000583397050764711
90 0.000792693177532167
91 0.000513627551723128
92 0.000697643136658654
93 0.000452286059657303
94 0.000614113037795017
95 0.000398341603004559
96 0.000540687727116018
97 0.000350891363029682
98 0.000476128917711205
99 0.000309144409797235
100 0.000419352583739818
};
\addlegendentry{Exact results}
\addplot [semithick, forestgreen4416044]
table {%
2 0.6706
3 0.325
4 0.4938
5 0.2559
6 0.3858
7 0.2037
8 0.3124
9 0.1707
10 0.2486
11 0.1388
12 0.2035
13 0.117
14 0.1712
15 0.0946
16 0.1391
17 0.0824
18 0.1184
19 0.0677
20 0.1016
21 0.0621
22 0.0807
23 0.0499000000000001
24 0.0689
25 0.039
26 0.0590000000000001
27 0.035
28 0.0538999999999999
29 0.0315
30 0.042
31 0.0265
32 0.0374
33 0.0229
34 0.03
35 0.0217000000000001
36 0.0286
37 0.0179
38 0.0256
39 0.016
40 0.0246
41 0.0147
42 0.0197000000000001
43 0.00960000000000005
44 0.0178
45 0.0111
46 0.0133
47 0.00980000000000003
48 0.0149
49 0.00919999999999999
50 0.011
51 0.00680000000000003
52 0.0105
53 0.00609999999999999
54 0.00839999999999996
55 0.00519999999999998
56 0.00700000000000001
57 0.00549999999999995
58 0.0071
59 0.00409999999999999
60 0.00670000000000004
61 0.00309999999999999
62 0.00570000000000004
63 0.004
64 0.00549999999999995
65 0.00319999999999998
66 0.00509999999999999
67 0.00280000000000002
68 0.00419999999999998
69 0.00219999999999998
70 0.00309999999999999
71 0.00260000000000005
72 0.00349999999999995
73 0.00219999999999998
74 0.00349999999999995
75 0.00149999999999995
76 0.00170000000000003
77 0.00160000000000005
78 0.00249999999999995
79 0.000900000000000012
80 0.00180000000000002
81 0.00160000000000005
82 0.00139999999999996
83 0.00119999999999998
84 0.00119999999999998
85 0.00119999999999998
86 0.00119999999999998
87 0.00139999999999996
88 0.001
89 0.000900000000000012
90 0.00119999999999998
91 0.000700000000000034
92 0.001
93 0.000600000000000045
94 0.000800000000000023
95 0.000600000000000045
96 0.000700000000000034
97 0.000499999999999945
98 0.000900000000000012
99 0.000299999999999967
100 0.000399999999999956
};
\addlegendentry{Monte-Carlo results}
\end{axis}

\end{tikzpicture}
\caption{Probability that candidate \(3\) fails to be the Condorcet winner as a function of \(n\) in a culture \( \mM_{3 \textnormal{ first}} \) with \(\rho = \log(2)\), shown on a semilog scale. The theoretical equivalent is based on \Cref{th_mallows-3-first}, while exact results rely on \Cref{eq_alpha-winner-coeff-extraction}. Monte Carlo simulations use 10,000 profiles per point, yielding an error of order \(10^{-2}\). When the exact results curve is not visible, it is overlapped by the Monte Carlo results curve.}\label{fig_mallows-3-first}
\end{center}
\end{figure}

\paragraph{Numerical simulations}  
\Cref{fig_mallows-3-first} illustrates these results. Since the probability that candidate~3 is the Condorcet winner tends towards~1, we now represent the complement probability on a semi-logarithmic scale to visualize the difference with the limit. The computation times are similar to those in \Cref{fig_mallows-3-last} as they involve the same culture and parameters. As in the previous cases, Monte Carlo simulations become unreliable when probabilities are small. We observe good agreement between the theoretical approximation, the exact results, and the Monte Carlo estimate when the latter is relevant.




\section{Asymptotic Expansion}\label{sec_error-terms}

We now explain how to improve our asymptotic results.
%
Both \Cref{th_subcritical} for the subcritical case and \Cref{th_critical} for the critical case are encompassed by \Cref{th_mixed-case}, which addresses the \emph{mixed case}. As detailed in \Cref{sec_supercritical-theoretical-result}, the supercritical case also reduces to the mixed case. Therefore, we focus on this scenario, involving a set $\mS$ of subcritical coordinates and a set $\mC$ of critical coordinates. Essentially, \cref{th_mixed-case} provides a constant $a_0$ such that
\[
    \probacondorcetwinner{m}
    \underset{n \to +\infty}{\sim}
    \frac{P(\vzeta)^n}
    {\prod_{j \in \mS} \zeta_j^{\lceil \beta_j n \rceil - 1}}
    \frac{1}{n^{|\mS|/2}}
    a_0.
\]
\cref{th:asympt_expansion} from \Cref{sec_asymptotic_expansion} extends this result, establishing the existence of coefficients $(a_{k, n})_{k, n \geq 1}$ such that, for any $r \geq 0$, we have
\[
    \probacondorcetwinner{m} =
    \frac{P(\vzeta)^n}
    {\prod_{j \in \mS} \zeta_j^{\lceil \beta_j n \rceil - 1}}
    \frac{1}{n^{|\mS|/2}}
    \left(
    a_0 + a_{1, n} n^{-1/2} + \cdots + a_{r-1, n} n^{-(r-1)/2} + \bigO(n^{-r/2}) 
    \right).
\]
In general, the coefficients $a_{k, n}$ may depend on $n$, but only through the bounded vector $\lceil \vbeta n \rceil - \vone - \vbeta n$. For example, in the case of the Condorcet winner, they depend only on the parity of~$n$.

Such an approximation is called an \emph{asymptotic expansion}.  
In the subcritical case, it is well known \cite[Section~VIII.3]{flajolet2009analytic, analytic2013pemantle}[Chapter~5] and proven in \Cref{th:asympt_expansion}, that $a_{k,n} = 0$ for odd $k$, leading to an expansion in powers of $n^{-1}$ instead of $n^{-1/2}$. However, this does not hold in general, as the following result illustrates.  

For the Impartial Culture, \Cref{th_impartial_culture} gives the limiting probability but not the convergence rate, which we now address by computing \( a_{1,n} \).

\begin{theorem}
\label{th_IC_error_term}
Under Impartial Culture with $m$ candidates, the probability that candidate $m$ is the Condorcet winner has asymptotic behavior
\[
    \probacondorcetwinner{m} =
    a_0 + a_{1,n} n^{-1/2} + \bigO(n^{-1}),
\]
where $a_0$ and $a_{1,n}$ are given by
\begin{align*}
    a_0 &=
    \frac{1}{\sqrt{(2 \pi)^{m-1} \det(\mH_K(\vzero))}}
    \int_{(0,+\infty)^{m-1}}
    e^{-\vu \mH_K(\vzero)^{-1} \vu / 2}
    d \vu,
\\
    a_{1,n} &=
    \begin{cases}
        0 & \text{if $n$ is odd},
        \\
        - \frac{6}{m+1}
        \frac{1}{ \sqrt{(2 \pi)^{m-1} \det(\mH_K(\vzero))} }
        \int_{(0, +\infty)^{m-1}}
        \vone^T \vu
        e^{- \vu^T \mH_K(\vzero)^{-1} \vu / 2}
        d \vu.
    & \text{if $n$ is even,}
    \end{cases}
\end{align*}
and $\mH_K(\vzero)$, its inverse and determinant are given in \Cref{eq:mHKIC}.
\end{theorem}

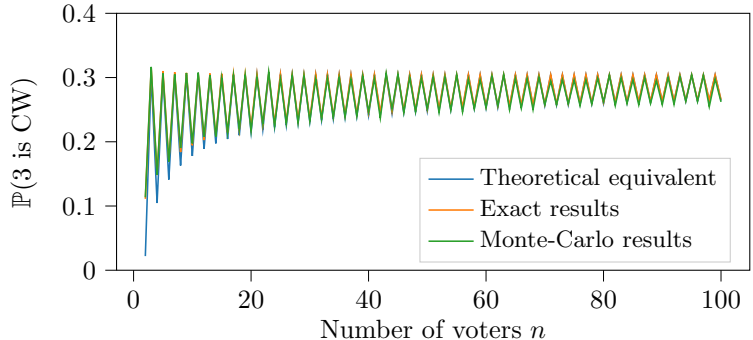
\begin{figure}
\begin{center}
\begin{tikzpicture}

\definecolor{darkgrey176}{RGB}{176,176,176}
\definecolor{darkorange25512714}{RGB}{255,127,14}
\definecolor{forestgreen4416044}{RGB}{44,160,44}
\definecolor{lightgrey204}{RGB}{204,204,204}
\definecolor{steelblue31119180}{RGB}{31,119,180}

\begin{axis}[
height=\axisHeight,
legend cell align={left},
legend style={font=\legendFont, 
  fill opacity=1,
  draw opacity=1,
  text opacity=1,
  at={(0.97,0.03)},
  anchor=south east,
  draw=lightgrey204
},
tick align=outside,
tick pos=left,
width=\axisWidth,
x grid style={darkgrey176},
xlabel={Number of voters $n$},
xmin=-2.9, xmax=104.9,
xtick style={color=black},
y grid style={darkgrey176},
ylabel={$\mathbb{P}(3 \text{ is CW})$},
ymin=0, ymax=0.4,
ytick={0.0, 0.1, 0.2, 0.3, 0.4},
ytick style={color=black}
]
\addplot [semithick, steelblue31119180]
table {%
2 0.021991932217439
3 0.304086723984094
4 0.104615583788485
5 0.304086723984094
6 0.141219220020624
7 0.304086723984094
8 0.163039328100767
9 0.304086723984094
10 0.177930097886316
11 0.304086723984094
12 0.188922007496598
13 0.304086723984094
14 0.197464914675363
15 0.304086723984094
16 0.20435115388629
17 0.304086723984094
18 0.210055126728543
19 0.304086723984094
20 0.21488051817874
21 0.304086723984094
22 0.219031944020146
23 0.304086723984094
24 0.222652972002359
25 0.304086723984094
26 0.225847705810555
27 0.304086723984094
28 0.228693719599511
29 0.304086723984094
30 0.231250161946487
31 0.304086723984094
32 0.23356302604243
33 0.304086723984094
34 0.235668687343474
35 0.304086723984094
36 0.237596343918891
37 0.304086723984094
38 0.239369740309884
39 0.304086723984094
40 0.241008410935205
41 0.304086723984094
42 0.242528593678213
43 0.304086723984094
44 0.243943912299257
45 0.304086723984094
46 0.245265893770318
47 0.304086723984094
48 0.246504365740346
49 0.304086723984094
50 0.247667765630763
51 0.304086723984094
52 0.248763383680207
53 0.304086723984094
54 0.249797555996271
55 0.304086723984094
56 0.250775819329729
57 0.304086723984094
58 0.251703036229
59 0.304086723984094
60 0.252583497048988
61 0.304086723984094
62 0.253421003708664
63 0.304086723984094
64 0.254218938935192
65 0.304086723984094
66 0.254980323876044
67 0.304086723984094
68 0.255707866320042
69 0.304086723984094
70 0.256404001284426
71 0.304086723984094
72 0.257070925356318
73 0.304086723984094
74 0.257710625893541
75 0.304086723984094
76 0.258324905970121
77 0.304086723984094
78 0.258915405780469
79 0.304086723984094
80 0.259483621081417
81 0.304086723984094
82 0.260030919144713
83 0.304086723984094
84 0.260558552607641
85 0.304086723984094
86 0.261067671541406
87 0.304086723984094
88 0.26155933400212
89 0.304086723984094
90 0.262034515284835
91 0.304086723984094
92 0.26249411606491
93 0.304086723984094
94 0.262938969581451
95 0.304086723984094
96 0.263369847993227
97 0.304086723984094
98 0.263787468017429
99 0.304086723984094
100 0.264192495944973
};
\addlegendentry{Theoretical equivalent}
\addplot [semithick, darkorange25512714]
table {%
2 0.111111111111111
3 0.314814814814815
4 0.148148148148148
5 0.310185185185185
6 0.169581618655693
7 0.308327617741198
8 0.183965858862978
9 0.30734128118173
10 0.194465480998831
11 0.306728641149723
12 0.202566406713069
13 0.306310528605833
14 0.209065535158771
15 0.306006806614097
16 0.214432326840295
17 0.305776127352619
18 0.218963456454917
19 0.305594952029487
20 0.222856822586211
21 0.305448882480342
22 0.226250199713895
23 0.305328613342689
24 0.229242879851164
25 0.305227862088046
26 0.231908492866705
27 0.305142233320958
28 0.234302962190104
29 0.305068559518588
30 0.236469637388801
31 0.305004500170059
32 0.238442716436943
33 0.304948288365469
34 0.240249593171544
35 0.304898565199271
36 0.24191250760723
37 0.304854268418318
38 0.243449731441944
39 0.304814555654353
40 0.244876436057957
41 0.304778750324587
42 0.246205338932428
43 0.304746302756819
44 0.247447192404867
45 0.304716761764114
46 0.248611158347832
47 0.304689753532354
48 0.249705098964687
49 0.304664965715857
50 0.250735805056098
51 0.30464213530132
52 0.251709177062998
53 0.304621039237945
54 0.252630370023913
55 0.304601487125072
56 0.253503910657579
57 0.304583315448787
58 0.254333792697468
59 0.304566382997766
60 0.255123555100848
61 0.304550567186113
62 0.25587634665643
63 0.304535761080526
64 0.25659497970299
65 0.304521870979254
66 0.257281975065402
67 0.304508814426946
68 0.257939599857595
69 0.304496518576493
70 0.258569899454288
71 0.30448491882911
72 0.25917472466649
73 0.304473957699034
74 0.259755754949281
75 0.304463583860719
76 0.260314518309419
77 0.304453751345188
78 0.260852408453892
79 0.304444418859005
80 0.26137069962064
81 0.304435549204567
82 0.261870559453116
83 0.304427108784563
84 0.262353060216757
85 0.304419067176684
86 0.262819188604172
87 0.304411396767217
88 0.263269854334393
89 0.304404072434249
90 0.263705897717798
91 0.304397071272792
92 0.264128096330704
93 0.304390372355518
94 0.264537170920994
95 0.304383956523825
96 0.264933790647399
97 0.304377806204839
98 0.2653185777396
99 0.304371905250683
100 0.265692111653369
};
\addlegendentry{Exact results}
\addplot [semithick, forestgreen4416044]
table {%
2 0.1136
3 0.3166
4 0.1488
5 0.3064
6 0.1692
7 0.3052
8 0.1903
9 0.3071
10 0.1975
11 0.3077
12 0.208
13 0.3039
14 0.2082
15 0.2997
16 0.2147
17 0.3033
18 0.22
19 0.3039
20 0.2173
21 0.3006
22 0.228
23 0.3086
24 0.225
25 0.2976
26 0.229
27 0.3033
28 0.2299
29 0.3027
30 0.2324
31 0.2977
32 0.2358
33 0.299
34 0.2386
35 0.3003
36 0.2392
37 0.2976
38 0.2403
39 0.3026
40 0.2468
41 0.2984
42 0.2369
43 0.3066
44 0.2511
45 0.3051
46 0.2456
47 0.3004
48 0.2472
49 0.2991
50 0.2515
51 0.301
52 0.2458
53 0.3011
54 0.2493
55 0.304
56 0.2465
57 0.2975
58 0.2515
59 0.2964
60 0.2534
61 0.3027
62 0.2562
63 0.3044
64 0.2499
65 0.2987
66 0.2523
67 0.2983
68 0.2501
69 0.3007
70 0.2578
71 0.2943
72 0.2578
73 0.2967
74 0.2592
75 0.2959
76 0.2564
77 0.2965
78 0.2587
79 0.2989
80 0.2625
81 0.3036
82 0.2564
83 0.2994
84 0.2597
85 0.2964
86 0.2575
87 0.2999
88 0.2597
89 0.295
90 0.2578
91 0.3003
92 0.2583
93 0.302
94 0.2632
95 0.304
96 0.2628
97 0.3033
98 0.2547
99 0.2978
100 0.2621
};
\addlegendentry{Monte-Carlo results}
\end{axis}

\end{tikzpicture}
\caption{Probability that a given candidate is the Condorcet winner as a function of \(n\) in the Impartial Culture with three candidates. The theoretical equivalent is based on \Cref{th_IC_error_term}, while exact results rely on \Cref{eq_alpha-winner-coeff-extraction}. Monte Carlo simulations use 10,000 profiles per point, yielding an error of order \(10^{-2}\). Not all curves being visible mean that they overlap.}\label{fig_ic-3-with-error-term}
\end{center}
\end{figure}

\paragraph{Numerical simulations}  
In \Cref{fig_ic-3}, we observed that the approximation by the limit value was highly accurate for odd values of \(n\). \Cref{th_IC_error_term} provides an explanation for this: the term of the asymptotic expansion in \( n^{-1/2} \) vanishes in this case. In contrast, for even \(n\), the term is non-zero and negative, which is consistent with the previous observation that the exact values are significantly lower than the limit in these cases. To illustrate this, \Cref{fig_ic-3-with-error-term} shows the same curves as \Cref{fig_ic-3} for the Monte Carlo simulations and exact values. However, this time, the theoretical approximation includes the term in \( n^{-1/2} \). As before, this approximation is computationally inexpensive: for instance, it takes only 55~\(\mu\)s for \(n = 100\). Now, the three curves become nearly indistinguishable visually: the addition of this term in \( n^{-1/2} \) significantly improves the accuracy of the approximation, the error term being $\bigO(n^{-1})$.


\section{Conclusion}\label{sec_conclusion}

\paragraph{Summary of the Contributions}
In this paper, we present a method for calculating the probability that a candidate is an $\valpha$-winner under the GIC model, assuming each ranking has a strictly positive probability. Our approach, based on analytic combinatorics, first computes the characteristic polynomial, the cumulant generating function, and the saddle point. The asymptotic analysis then splits into cases based on whether each coordinate of the saddle point is subcritical, critical, or supercritical. In all cases, we derive the limiting probability and rate of convergence. Additionally, our method computes higher-order terms in the asymptotic expansion. We apply this approach to the Impartial Culture and Mallows models, providing explicit formulas supported by numerical simulations.

\paragraph{Future Work}
A natural direction for future work is to investigate cases where the saddle point has null or infinite coordinates, which may arise in non-generic cultures. Additionally, it would be valuable to perform a finer analysis of the complexity of the algorithm presented in \Cref{sec_supercritical-case} for handling supercritical coordinates, particularly to identify conditions under which it runs in polynomial time. Another avenue is to examine the limiting behavior of the probability as the number of candidates \( m \) tends to infinity. 
Finally, it would be insightful to extend our method to other events, such as the transitivity of the majority relation, different kinds of monotonicity failures, or the manipulability, \ie the susceptibility to strategic voting.

\bibliographystyle{apalike}
\bibliography{biblio}

\newpage

\appendix

\section{Multivariate Coefficient Extraction of Large Powers}\label{sec_appendix-Elie}

This Appendix presents the main technical contributions of the paper in analytic combinatorics. We assume here that the reader is familiar with the main concepts of this field. For an introduction, less familiar readers may refer to \cite{flajolet2009analytic} or \Cref{sec_preliminaries,sec_step-by-step} of this paper.

The \emph{large power theorem}
\cite[Section~VIII.8]{flajolet2009analytic}
provides sufficient conditions
to derive the asymptotics
of coefficients extraction of the form
\[
    [z^{\lambda n}]
    A(z)
    B(z)^n
\]
as $n$ tends to infinity.
%
Our goal here is to extend this result to the multivariate setting
\[
    [\vz^{\vlambda n}] A(\vz) B(\vz)^n.
\]
We also cover variants,
including a case where $A(\vz)$ has a singularity
at the \emph{saddle point} (defined below).

To this end, we rely on the following multivariate version of the Laplace method, presented by \cite[Theorem 5.1.2, p. 98]{analytic2013pemantle}.

\begin{theorem}[Multivariate Laplace method]
\label{th:multivariate_laplace_method}
Consider a compact neighborhood $C \subset \reals^d$ of $\vzero$
and two functions $A(\vx)$ and $\phi(\vx)$
analytic on $C$.
Suppose that the real part of $\phi(\vx)$ is strictly positive
except at the origin,
that its gradient vanishes at the origin
and that its Hessian matrix $\mH$
is non-singular there.
Assume $A(\vzero) \neq 0$, then
\[
    \int_C
    A(\vx)
    e^{- n \phi(\vx)}
    d \vx
    \sim
    \left( \frac{2 \pi}{n} \right)^{d/2}
    \frac{A(\vzero)}
        {\sqrt{\det(\mH)}}
\]
and the choice of sign of the square root is defined by taking the product
of the principal square roots of the eigenvalues of $\mH$.
\end{theorem}

Variants of this theorem exist for the case where $A(\vzero) = 0$
and where the Hessian of $\phi(\vx)$ is singular at $\vzero$.

If the assumptions of the theorem are satisfied,
except $\phi(\vzero) \neq 0$,
then we rewrite
\[
\int_C A(\vx) e^{- n \phi(\vx)} d \vx =
e^{- n \phi(\vzero)}
\int_C A(\vx) e^{- n (\phi(\vx) - \phi(\vzero))} d \vx.
\]
and the theorem is appicable with $\phi(\vx)$ replaced by $\phi(\vx) - \phi(\vzero)$.

If the point where the gradient of $\phi(\vx)$ vanishes
is some other point than $\vzero$,
as long as the Hessian of $\phi$ is non-singular at this point
and $A$ does not vanish there,
the theorem is applicable after a translation
sending this point to $\vzero$.

\paragraph{Road map.}
To reduce the multivariate coefficient extraction
of large powers to a Laplace integral,
we first establish a few lemmas in \cref{sec_appendix-Elie-preliminaries}
ensuring that the assumptions of the Laplace method
are satisfied.
We will derive the multivariate large powers theorem
in \cref{sec_appendix-large-powers}.
It is used in this article for the subcritical case from \cref{sec_subcritical-case}.
Then we consider in \cref{sec_appendix-Elie-singularity}
the case where the saddle point meets a singularity,
which is applied both in \cref{sec_critical-case}
and \cref{sec_mixed-case}.

\paragraph{Notation.}
We use bold letters to denote vectors, \eg $\vx = (x_1, \ldots, x_d)$ with $d \in \mathbb{N}_{>0}$.
Given two vectors $\vx$ and $\vy$ of the same dimension $d$ and a scalar $z$, we define the notations
\begin{align*}
    \vx^{\vy}
&:=
    {x_1}^{y_1} \cdots {x_d}^{y_d},
\\
    \vx z^{\vy}
&:=
    (x_1 z^{y_1}, \ldots, x_d z^{y_d}),
\\
    \log(\vx)
&:=
    (\log(x_1), \ldots, \log(x_d)).
\end{align*}
In particular, we write $\vzeta e^{i \vtheta}$
for the vector $(\zeta_1 e^{i \theta_1}, \ldots, \zeta_d e^{i \theta_d})$.
The \emph{support} of a multivariate power series $B(\vz)$
is the set
\[
    S_B :=
    \{ \vn,\ [\vz^{\vn}] B(\vz) \neq 0 \}
\]
and we define
\[
    \Delta(S_B) :=
    \{ \vn - \vm,\ \vn, \vm \in S_B \}.
\]
The \emph{cumulant generating function}
associated to the generating function $B(\vz)$
is
\[
    K(\vt) := \log\big(B(e^{\vt})\big)
    = \log\big(B(e^{t_1}, \ldots, e^{t_d})\big).
\]
The gradient of a function $B$ at $\vx$ is denoted by $\nabla_B(\vx)$.

We write $r + q \integers_{\geq 0}$ for the set $\{r + q \ell,\ \ell \in \integers_{\geq 0}\}$.

A \emph{polydisc} is a Cartesian product of discs.
A \emph{torus} is a Cartesian product of circles.

\subsection{Rank and Periodicity of a Multivariate Generating Function}\label{sec_appendix-Elie-preliminaries}

This section presents general results on multivariate formal power series,
defining their \emph{rank} and \emph{period}.
We will use them to prove the existence and uniqueness of the saddle point,
defined in the following section.
In particular, we prove a multivariate version
of the Daffodil lemma \cite[Lemma~IV.1, page~266]{flajolet2009analytic}.

\begin{lemma}
\label{th:smithnormalform}
The kernel of any integer matrix
has a basis composed of integer vectors.
\end{lemma}

\begin{proof}
Let $M$ denote our matrix and let its Smith normal form be
\[
    M = U D V,
\]
where $U$ and $V$ are unimodular integer matrices
and $D$ is an integer diagonal matrix.
The dimension of the kernel of $M$
is equal to the number of zero elements
on the diagonal of $D$.
Consider such an element of index $j$
and let $\ve^{(j)}$, denote the vector with a $1$
in position $j$ and $0$'s everywhere else.
Set $\vp^{(j)} = V^{-1} \ve^{(j)}$.
Since $V$ is unimodular,
$\vp^{(j)}$ has integer coefficients and is nonzero.
We have
\[
    M \vp^{(j)} = U D V \vp^{(j)} = U D \ve^{(j)} = \vzero,
\]
so $\vp^{(j)}$ belongs to the kernel of $M$.
The invertibility of $V$ ensures
that the vectors $\vp^{(j)}$ are linearly independent,
hence they form a basis of the kernel.
\end{proof}

\begin{definition}
\label{def:rank}
We define the \emph{rank}
of a formal power series $B(\vz)$
as the dimension of the real vector space generated by $\Delta(S_B)$.
We say that $B(\vz)$ has \emph{full rank}
if its rank is equal to its number of variables.
\end{definition}

In particular, a formal power series
on at least one variable
and that has full rank cannot be a monomial.

\begin{lemma}
\label{th:rank_characterization}
If the formal power series $B(z_1, \ldots, z_d)$
has rank $k$,
then there exist $d-k$ integer vectors
$\vp^{(1)}, \ldots, \vp^{(d-k)}$
generating a real vector space of dimension $d-k$
and integers $r_1, \ldots, r_{d-k}$
such that for all $j \in [1, d-k]$
\[
    B(\vz y^{\vp^{(j)}})
    =
    y^{r_j}
    B(\vz).
\]
Reciprocally, if such vectors $\vp^{(j)}$
and integers $r_j$ exist,
then $B(\vz)$ has rank at most $k$.
\end{lemma}

\begin{proof}
If $B(z_1, \ldots, z_d)$ has rank $k$,
by \cref{th:smithnormalform},
the vector space orthogonal to $\Delta(S_B)$
has a basis of $d-k$ integer vectors
$(\vp^{(1)}, \ldots, \vp^{(d-k)})$.
Fix $\vm$ in the support $S_B$ of $B(\vz)$
and define $r_j = \vm^T \vp^{(j)}$,
then for all $\vn$ in $S_B$,
we have
\[
    \vn^T \vp^{(j)} = r_j,
\]
so
\[
    B(\vz y^{\vp^{(j)}})
    =
    y^{r_j}
    B(\vz).
\]

Reciprocally, if there are $d-k$ linearly independent
integer vectors $(\vp^{(1)}, \ldots, \vp^{(d-k)})$
and integers $r_1, \ldots, r_{d-k}$ such that
\[
    B(\vz y^{\vp^{(j)}})
    =
    y^{r_j}
    B(\vz),
\]
then each $\vp^{(j)}$ is in the vector space orthogonal to
$\Delta(S_B)$,
so this vector space has dimension at most $k$.
\end{proof}

In the current article, we will consider a polynomial $P$
having full rank.
It is however worth noticing that coefficient extraction
in a power series that does not have full rank
can be simplified as follows.
Say we are interested in the coefficient extraction
$[\vz^{\vn}] B(\vz)$ of a power series $B(\vz)$
that does not have full rank.
Let $\vp$ and $r$ denote
a nonzero integer vector and an integer such that
\[
    B(\vz y^{\vp}) = y^r B(\vz).
\]
Since $\vp$ is nonzero, there exists $j$
such that $p_j \neq 0$.
Then for any $\vn \in S_B$, we have
\[
    n_j = \frac{r - \sum_{\ell \neq j} p_{\ell} n_{\ell}}{p_j},
\]
so $n_j$ is uniquely determined by $(n_{\ell})_{\ell \neq j}$
and
\[
    [\vz^{\vn}] B(\vz) =
    \begin{cases}
    0 & \text{if $\vn^T \vp \neq r$},
    \\
    [z_1^{n_1} \cdots z_{j-1}^{n_{j-1}}
    z_{j+1}^{n_{j+1}} \cdots z_d^{n_d}]
    B(z_1, \ldots, z_{j-1}, 1, z_{j+1}, \ldots, z_d)
    & \text{otherwise}.
    \end{cases}
\]
Thus, coefficient extraction in a power series
that does not have full rank
reduces to a coefficient extraction
in a power series with fewer variables.

\begin{lemma}[Convexity of the cumulant generating function]
\label{th_convexity-cumulant-gf}
Consider a series $B(\vz)$
with nonnegative coefficients,
full rank,
and analytic on a neighborhood $\Omega$ of $\vzero$.
Then its cumulant generating function
\[
    K : \vt \mapsto
    \log\big(B(e^{\vt})\big)
\]
is strictly convex
on $\{\vt \in \reals^d,\ e^{\vt} \in \Omega\}$.
\end{lemma}

\begin{proof}
Let $\vX$ denote the vector of random variables
parametrized by $\vx \in \reals_{> 0}^d \cap \Omega$
with distribution
\[
    \proba(\vX = \vn) =
    \frac{[\vz^{\vn}] B(\vz) \vx^{\vn}}{B(\vx)}.
\]
By construction, those probabilities sum to $1$.
This is the so-called \emph{Boltzmann distribution}
associated to the generating function $B(\vz)$
(see \cite{DFLS04}).
We now follow the classical proof of convexity
of the cumulant generating function
of a multivariate random variable,
then prove that the assumption that $B(\vz)$ has full rank
implies strict convexity.

The moment generating function of $\vX$ is
\[
    \mean(e^{\vt^T \vX}) =
    \frac{B(\vx e^{\vt})}{B(\vx)}.
\]
Observe that the cumulant generating function of $\vX$
is equal to
\[
    \log\big(\mean(e^{\vt^T \vX})\big)
    =
    \log\big(B(\vx e^{\vt})\big)
    - \log\big(B(\vx)\big)
    =
    K\big(\vt + \log(\vx)\big) - \log\big(B(\vx)\big).
\]
The vector of the means is
\[
    \mean(\vX) =
    \frac{\diag(\vx) \nabla_B(\vx)}{B(\vx)}.
\]
It is equal to the gradient of $K(\vt)$ at $\vt = \log(\vx)$.
The covariance matrix is
\[
    \var(\vX)
    =
    \mean\big((\vX - \mean(\vX))^T (\vX - \mean(\vX))\big)
    =
    \mean\big(\vX \vX^T\big)
    - \mean(\vX) \mean(\vX)^T.
\]
It is equal to the Hessian of $K(\vt)$ at $\vt = \log(\vx)$.

For any vector $\vu \in \reals^d$,
\[
    \vu^T \var(\vX) \vu =
    \mean\Big(\big((\vX - \mean(\vX))^T \vu\big)^2\Big),
\]
so $\var(\vX)$ is positive semi-definite.
Consider a vector $\vu$ such that
\[
    0
    =
    \vu^T \var(\vX) \vu
    =
    \mean\Big(\big((\vX - \mean(\vX))^T \vu\big)^2\Big).
\]
This implies $(\vX - \mean(\vX))^T \vu = \vzero$
so, for any $\vn$, $\vm$ in the support of $B(\vx)$,
\[
    (\vn - \vm)^T \vu = 0.
\]
Our assumption that $B(\vz)$ has full rank
implies $\vu = \vzero$,
which proves that $\var(\vX)$ is positive definite.
Since $\var(\vX)$ is equal to the Hessian of $K(\vt)$
at $\log(\vx)$,
we deduce that $K(\vt)$ is strictly convex
on $\{\vt \in \reals^d,\ e^{\vt} \in \Omega\}$.
\end{proof}

We recall below the definition of $q$-periodicity
from \cite[Section 4.2]{flajolet1991automatic}
for univariate formal power series.

\begin{definition}
\label{def:periodicity_univariate}
A univariate generating function $B(z)$
is said to be \emph{$q$-periodic}
if there exist an integer $r$ and a formal power series $C(z)$
such that
\[
    B(z) = z^r C(z^q).
\]
Equivalently, the support $S_B$ of $B(z)$ satisfies
\[
    S_B \subseteq r + q \integers_{\geq 0}. 
\]
The series $B(z)$ is said to be \emph{aperiodic}
if it is not $q$-periodic for any $q \geq 2$.
\end{definition}

Let us now extend this definition
to multivariate formal power series.

\begin{definition}
\label{def:periodicity}
A multivariate generating function $B(z_1, \ldots, z_d)$
is said to be \emph{$q$-periodic} if there exist
nonnegative integers $p_1, \ldots, p_d$,
not all divisible by $q$,
such that $B(y^{p_1}, \ldots, y^{p_d})$
is a $q$-periodic univariate formal power series.
An equivalent formulation on the support $S_B$ of $B(\vz)$
is that there exists a nonnegative integer $r$ such that
\begin{equation}
\label{eq:multivariate_periodicity_support}
    \{\vp^T \vn,\ \vn \in S_B\} \subseteq r + q \integers_{\geq 0}.
\end{equation}
Yet another equivalent formulation is that
there exists a formal power series $E(\vz, y)$,
integers $p_1, \ldots, p_d$ not all divisible by $q$
and an integer $r$ such that
\[
    B(z_1 y^{p_1}, \ldots, z_d y^{p_d}) =
    y^r E(\vz, y^q).
\]
We say that $B(\vz)$ is \emph{aperiodic}
if it is not $q$-periodic for any $q \geq 2$.
\end{definition}

Our next result presents a simple sufficient condition
for a formal multivariate power series to have full rank
and be aperiodic.

\begin{proposition}
\label{th_full_rank_aperiodic}
Let $\ve_j$ denote the vector with a $1$ at position $j$,
and $0$s everywhere else.
Consider a formal power series $B(z_1, \ldots, z_d)$
whose support contains $\vzero, \ve_1, \ldots, \ve_d$.
They $B(\vz)$ has full rank and is aperiodic.
\end{proposition}

\begin{proof}
Since the support of $B(\vz)$ contains $\vzero$,
we have $S_B \subseteq \Delta(S_B)$.
Thus, $\Delta(S_B)$ contains $d$ independent vectors
$\ve_1, \ldots, \ve_d$, so it has rank $d$ and $B(\vz)$ has full rank.
Consider an integer vector $\vp$ and integers $r$ and $q \geq 2$.
Assume
\[
    \{ \vp^T \vn,\ \vn \in S_B\}
    \subseteq
    r + q \integers_{\geq 0}.
\]
For $\vn = \vzero$, we deduce that $q$ divides $r$.
For $\vn = \ve_j$, we deduce that $q$ divides $p_j - r$,
so $p_j$ is a multiple of $q$.
Thus, $\vp$ is necessarily a multiple of $q$
and $B(\vz)$ cannot be $q$-periodic, so $B(\vz)$ is aperiodic.
\end{proof}

The Daffodil lemma \cite[Lemma~IV.1, page~266]{flajolet2009analytic}
states that the absolute value of a univariate generating function
with nonnegative coefficients
reaches its maximum at a unique point
on any disk contained in its domain of convergence
if and only if it is aperiodic.
Our following lemma extends this result to the multivariate setting.

\begin{lemma}[Multivariate daffodil lemma]
\label{th:multivariate_daffodil}
Consider a multivariate power series $B(z_1, \ldots, z_d)$
with nonnegative coefficients
and a vector $\vzeta \in \reals_{> 0}^d$
such that the domain of convergence of $B(\vz)$ contains
the closed polydisc of radius $\vzeta$
\[
    D(\vzeta) :=
    \{ \vz,
    |z_1| \leq \zeta_1 \land \ldots \land |z_d| \leq \zeta_d \}.
\]
Then for all $\vz \in D(\vzeta)$, we have
$|B(\vz)| \leq B(\vzeta)$.

Assume furthermore that $B(\vz)$ has full rank.
Then $|B(\vz)|$ has a unique maximum on $D(\vzeta)$
if and only if it is aperiodic.
\end{lemma}

\begin{proof}
Using the positivity of the coefficients of $B(\vz)$,
triangular inequality yields for any $\vx \in D(\vzeta)$
\[
    |B(\vx)|
    =
    \left|
    \sum_{\vn}
    b_{\vn}
    \vx^{\vn}
    \right|
    \leq
    \sum_{\vn}
    b_{\vn}
    |\vx|^{\vn}
    \leq
    \sum_{\vn}
    b_{\vn}
    \vzeta^{\vn}
    =
    B(\vzeta).
\]
Assume $|B(\vx)| = B(\vzeta)$,
then $(|x_1|, \ldots, |x_d|) = \vzeta$,
so there exists a vector $\vt \in [0, 1)^d$
such that for all $j$, we have $x_j = \zeta_j e^{2 i \pi t_j}$.
The case $\vt = \vzero$ corresponds to $\vx = \vzeta$.
Let us assume first $\vt \neq \vzero$, and prove that there exists $q \geq $
such that $B(\vz)$ is $q$-periodic.
The equality $|B(\vx)| = B(\vzeta)$ implies
\[
    \left|
    \sum_{\vn}
    b_{\vn}
    \vzeta^{\vn}
    e^{2 i \pi \vn^T \vt}
    \right|
    =
    \sum_{\vn}
    b_{\vn}
    \vzeta^{\vn},
\]
so by the triangular inequality,
for all $\vn$ in the support $S_B$ of $B(\vz)$,
the complex numbers $e^{2 i \pi \vn^T \vt}$ are aligned.
This is equivalent with the existence of $t_0 \in [0,1)$ such that
for all $\vn \in S$, we have
\[
    \vn^T \vt = t_0 \mod 1.
\]
Then for all $\vn \in \Delta(S_B)$, we have
\[
    \vn^T \vt = 0 \mod 1.
\]
Thus, for all $\vn \in \Delta(S_B)$, there exist integers $k_{\vn}$ such that
\[
    \vn^T \vt = k_{\vn}.
\]
Let $\Delta(S_B)'$ denote the set
of vectors $\vn$ from $\Delta(S_B)$
with an additional element $- k_{\vn}$ at the bottom,
and let $\vt'$ denote the vector $\vt$
with an additional element $1$ at the bottom,
then the previous equality is equivalent with
\[
    \vn^T \vt' = 0
\]
for all $\vn \in \Delta(S_B)'$.
Since $B(\vz)$ has full rank,
$\Delta(S_B)$ generates a vector space of dimension $d$,
so the vector space orthogonal to $\Delta'(S_B)$ has dimension at most $1$.
It contains a nonzero vector, so the dimension is exactly $1$.
By \cref{th:smithnormalform},
there exists a nonzero integer vector $\vu$
orthogonal to $\Delta(S_B)'$.
We have $t'_{d+1} = 1$ and $\vt'$ proportional to the integer vector $\vu$,
so $\vt$ is a vector of rational values.
Recall that $\vt \in [0,1)^d$ and is nonzero.
Thus, there exist
$q \in \integers_{\geq 2}$
and $\vp \in \integers_{\geq 0}^d$
with $q$ that does not divide all $p_j$,
such that $\vt = \vp / q$, so for all $\vn \in \Delta(S_B)$,
\[
    \vn^T \vp = q\, k_{\vn}.
\]
Fix $r = \min_{\vn \in S_B} (\vn^T \vp)$, then
for all $\vn \in S_B$,
\[
    \vn^T \vp - r = q\, k_{\vn}.
\]
Thus, $q$ divides $\vn^T \vp - r$.
This implies the existence of a formal univariate power series $C(y)$
such that
\begin{equation}
\label{eq:BpyrC}
    B(y^{p_1}, \ldots, y^{p_d}) = y^r C(y^q),
\end{equation}
so $B(\vz)$ is $q$-periodic.

Reciprocally, if $B(\vz)$ is $q$-periodic for some $q \geq 2$,
there exist nonnegative integers $r, p_1, \ldots, p_d$,
with $q$ not dividing all $p_j$,
and a formal power series $C(y)$
such that \cref{eq:BpyrC} holds.
Then for all $\vn \in S_B$, the integer $q$ divides $\vn^T \vp - r$, so $e^{2 i \pi \vn^T \vp / q} = e^{2 i \pi r / q}$.
For any $\vzeta \in \reals_{> 0}^d$ such that the domain of convergence of $B(\vz)$ contains the polydisc of radius $\vzeta$, setting $\vx = \vzeta e^{2 i \pi \vp / q}$, we have
\[
    |B(\vx)|
    =
    \left|
    \sum_{\vn \in S_B}
    b_{\vn}
    \vzeta^{\vn}
    e^{2 i \pi \vp^T \vn / q}
    \right|
    =
    \left|
    \sum_{\vn \in S_B}
    b_{\vn}
    \vzeta^{\vn}
    e^{2 i \pi r}
    \right|
    =
    \left|
    \sum_{\vn \in S_B}
    b_{\vn}
    \vzeta^{\vn}
    \right|
    =
    B(\vzeta),
\]
where the last equality comes from the fact
that the coefficients of $B(\vz)$ are nonnegative.
Thus, on the torus of radius $\vzeta$,
at the point $\vx$, which is distinct from $\vzeta$,
we have $|B(\vx)| = B(\vzeta)$.
\end{proof}

In the current article, we will consider a polynomial $P$
that is aperiodic.
It is however worth noticing that coefficient extraction
in a periodic power series
can be simplified as follows.
Suppose we are interested in the coefficient extraction
$[\vz^{\vn}] B(\vz)$.
If $B(\vz)$ does not have full rank,
we saw after \cref{th:rank_characterization}
how to simplify the coefficient extraction.
Assume now $B(\vz)$ has full rank,
is analytic at $\vzero$,
and is $q$-periodic for some $q \geq 2$.
Then there exist nonnegative integers
$r, p_1, \ldots, p_d$,
with $q$ not dividing all $p_j$,
and a power series $E(\vz, y)$ such that
\[
    B(\vz y^{\vp}) =
    y^r E(\vz, y^q).
\]
We have
\[
    [\vz^{\vn}] B(\vz)
    =
    [\vz^{\vn} y^{\vp^T \vn}]
    B(z_1 y^{p_1}, \ldots, z_d y^{p_d})
    =
    [\vz^{\vn} y^{\vp^T \vn}]
    y^r E(\vz,y^q).
\]
Let $j$ be such that $p_j \neq 0$
(which must exist, otherwise $q$ would divide all $p_i$).
Let $\vn_{\setminus j}$ denote the vector
$(n_1, \ldots, n_{j-1}, n_{j+1}, \ldots, n_d)$
and similarly for $\vz_{\setminus j}$.
Then $B(\vz)$ having full rank implies
\[
    [\vz^{\vn} y^{\vp^T \vn}]
    B(\vz y^{\vp})
    =
    [\vz_{\setminus j}^{\vn_{\setminus j}} y^{\vp^T \vn}]
    B(z_1 y^{p_1}, \ldots, z_{j-1} y^{p_{j-1}}, y^{p_j}, z_{j+1} y^{p_{j+1}}, \ldots z_d y^{p_d}).
\]
Define
\[
    D(\vz_{\setminus j}, y) :=
    E(z_1, \ldots, z_{j-1}, 1, z_{j+1}, \ldots, z_d, y),
\]
then
\[
    [\vz^{\vn}] B(\vz)
    =
    \begin{cases}
    0 & \text{if $q$ does not divide $\vp^T \vn - r$},
    \\
    [\vz_{\setminus j}^{\vn_{\setminus j}} y^{(\vp^T \vn - r) / q}]
    D(\vz_{\setminus j}, y)
    & \text{otherwise.}
    \end{cases}
\]
Thus, when $|B(\vz)|$ reaches its maximum on a polydisc
at several points, the coefficient extraction
$[\vz^{\vn}] B(\vz)$ can be simplified.

\begin{corollary}
Consider a power series $B(\vz)$
and two complex torii $T_1$, $T_2$
centered at the origin
and contained in the domain of convergence of $B(\vz)$.
If $|B(\vz)|$ reaches its maximum on $T_1$ at a unique point,
then it reaches its maximum on $T_2$ at a unique point as well.
\end{corollary}

    \subsection{Large Powers}\label{sec_appendix-large-powers}

\Cref{th_large-powers-limit} is the key result used in \Cref{sec_subcritical-case}, which corresponds to the case where the saddle point has all coordinates subcritical. Consequently, the contour of integration can pass through the saddle point without intersecting the singularity at \(1\).  

We first present a preliminary theorem, which corresponds to extracting the coefficient of index \(\vlambda n\) for some fixed vector \(\vlambda\) of rational values. Then, we state the theorem in its more general form, considering a sequence of vectors of real values \(\vkappa_n = \vlambda n + \bigO(n^{\delta})\).

\begin{theorem}[Multivariate large powers theorem]
\label{th:multivariate_large_powers}
Consider a power series $B(z_1, \ldots, z_d)$
with nonnegative coefficients,
full rank (\cref{def:rank}), aperiodic (\cref{def:periodicity})
and convergent on a neighborhood $\Omega$ of $0$.
Let $K(\vt) = \log\big(B(e^{\vt})\big)$ denote
its cumulant generating function.
Let $\vlambda$ denote a vector of positive rational numbers
such that the function $K(\vt) - \vlambda^T \vt$
does not tend to its infimum on the boundary
of $\Psi = \{ \vt \in \reals^d,\ e^{\vt} \in \Omega \}$.
Let $\vtau$ denote the unique point of $\Psi$
where $K(\vt) - \vlambda^T \vt$ reaches its minimum
and $\mH$ denote the Hessian matrix of $K(\vt)$ at $\vtau$.
Set $\vzeta = e^{\vtau}$.
Consider a power series $A(\vz)$
analytic on the torus of radius $\vzeta$
and satisfying $A(\vzeta) \neq 0$.
Then, considering large values of $n$ such that $\vlambda n$ is a vector of integers,
\[
    [\vz^{\vlambda n}]
    A(\vz)
    B(\vz)^n
    \sim
    \frac{A(\vzeta)}
        {\sqrt{(2 \pi n)^d \det(\mH)}}
    \frac{B(\vzeta)^n}{\vzeta^{\vlambda n}}.
\]
\end{theorem}

\begin{proof}
\cref{th_convexity-cumulant-gf} ensures
that $K(\vt)$ is strictly convex,
so $K(\vt) - \vlambda^T \vt$ is strictly convex as well.
Since it does not reach its minimum on the boundary of $\Psi$,
there exists a unique point $\vtau$ in $\Psi$
where $K(\vt) - \vlambda^T \vt$ reaches its minimum.
We write the coefficient extraction as a Cauchy integral.
\[
    [\vz^{\vlambda n}]
    A(\vz)
    B(\vz)^n
    =
    \frac{1}{(2 i \pi)^d}
    \oint
    A(\vz)
    B(\vz)^n
    \frac{d \vz}{\vz^{\vlambda n + \vone}}.
\]
The analycity of $B(\vz)$ at $\vzeta$
and its nonnegative coefficients
ensure that $B(\vz)$ is analytic on the torus of radius~$\vzeta$.
We choose for the integration domain
the torus of radius $\vzeta$
\[
    [\vz^{\vlambda n}]
    A(\vz)
    B(\vz)^n
    =
    \frac{1}{(2 \pi)^d}
    \int_{\vtheta \in [-\pi, \pi]^d}
    A(\vzeta e^{i \vtheta})
    B(\vzeta e^{i \vtheta})^n
    \frac{d \vtheta}
    {\vzeta^{\vlambda n} e^{i n \vlambda^T \vtheta}}.
\]
Introducing
\[
    \phi(\vtheta) :=
    i \vlambda^T \vtheta -
    \log \left(
    \frac{B(\vzeta e^{i \vtheta})}
    {B(\vzeta)}
    \right),
\]
we have
\[
    [\vz^{\vlambda n}]
    A(\vz)
    B(\vz)^n
    =
    \frac{B(\vzeta)^n}{\vzeta^{\vlambda n}}
    \frac{1}{(2 \pi)^d}
    \int_{\vtheta \in [-\pi, \pi]^d}
    A(\vzeta e^{i \vtheta})
    e^{- n \phi(\vtheta)}
    d \vtheta.
\]
Since $|B(\vzeta e^{i \vtheta})|$
reaches its maximal value $B(\vzeta)$ only at $\vtheta = 0$,
the real part of $\phi(\vtheta)$
is positive except at $\vtheta = \vzero$.
The Hessian matrix of $\phi(\vtheta)$ at $\vtheta = \vzero$
is equal to the Hessian matrix $\mH$
of $K(\vt)$ at $\vt = \log(\vzeta)$,
and is nonsingular
according to \cref{th_convexity-cumulant-gf}.
The result of the theorem
follows by application of the Laplace method,
recalled in \cref{th:multivariate_laplace_method}.
\end{proof}

Applications often require more flexibility
than provided by our previous result.
Instead of extracting the coefficient of index $\vlambda n$
for some fixed vector $\vlambda$ of rational values,
we might consider an index $\vkappa_n \sim \vlambda n$,
which allows us to cover the case where $\vlambda$
contains irrational values.
To achieve this result,
we follow and modify slightly
the proof of \cref{th:multivariate_laplace_method}.

\begin{theorem}
\label{th_large-powers-limit}
Consider $\delta < 1/2$
and a sequence of vectors of real values
\[
    \vkappa_n = \vlambda n + \bigO(n^{\delta}).
\]
Under the assumptions of \cref{th:multivariate_large_powers},
with $\vtau$ defined as the minimum
of $K(\vt) - \vlambda^T \vt$
and $\vzeta = e^{\vtau}$,
we have
\[
    [\vz^{\vkappa_n}]
    A(\vz)
    B(\vz)^n
    \sim
    \frac{A(\vzeta)}
        {\sqrt{(2 \pi n)^d \det(\mH)}}
    \frac{B(\vzeta)^n}{\vzeta^{\vkappa_n}}.
\]
\end{theorem}

\begin{proof}
We write again the coefficient extraction
as a Cauchy integral on the torus of radius $\vzeta$
\[
    [\vz^{\vkappa_n}]
    A(\vz)
    B(\vz)^n
    =
    \frac{1}{(2 \pi)^d}
    \int_{(-\pi, \pi]^d}
    A(\vzeta e^{i \vtheta})
    B(\vzeta e^{i \vtheta})^n
    \frac{d \theta}
    {\vzeta^{\vkappa_n} e^{i \vkappa_n^T \vtheta}}
\]
and introduce
\begin{equation}
\label{eq:phi}
    \phi(\vtheta) :=
    i \vlambda^T \vtheta
    - \log \left(
    \frac{B(\vzeta e^{i \vtheta})}{B(\vzeta)}
    \right)
\end{equation}
to rewrite it as
\begin{equation}
\label{eq:central_part_large_powers_limit}
    [\vz^{\vkappa_n}]
    A(\vz)
    B(\vz)^n
    =
    \frac{B(\vzeta)^n}
    {\vzeta^{\vkappa_n}}
    \frac{1}{(2 \pi)^d}
    \int_{(-\pi, \pi]^d}
    A(\vzeta e^{i \vtheta})
    e^{- n \phi(\vtheta)}
    \frac{d \theta}
    {e^{i (\vkappa_n - \vlambda n)^T \vtheta}}.
\end{equation}

\paragraphproof{Central part.}
We introduce a positive sequence $\varepsilon_n$
converging to $0$ and specified later.
The central part of the integral is defined as
\[
    C_n :=
    \frac{1}{(2 \pi)^d}
    \int_{(-\varepsilon_n, \varepsilon_n)^d}
    A(\vzeta e^{i \vtheta})
    e^{- n \phi(\vtheta)}
    \frac{d \theta}
    {e^{i (\vkappa_n - \vlambda n)^T \vtheta}}.
\]
Since $\phi(\vzero) = \vzero$
and, by definition of $\vzeta$,
the gradient of $\phi(\vtheta)$ also vanishes at $\vzero$,
the Taylor expansion of $\phi(\vtheta)$ is
\begin{equation}
\label{eq:phi_gaussian}
    \phi(\vtheta) =
    - \frac{1}{2}
    \vtheta^T \mH \vtheta
    + \bigO(\| \vtheta \|^3).
\end{equation}
For $\vtheta \in (-\varepsilon_n, \varepsilon_n)$, we have
\[
    n \phi(\vtheta) =
    - \frac{n}{2}
    \vtheta^T \mH \vtheta
    + \bigO(n {\varepsilon_n}^3).
\]
We will choose $\varepsilon_n$ to ensure $n {\varepsilon_n}^3 \to 0$.
We will also ensure $n^{\delta} \varepsilon_n \to 0$,
so $e^{i (\vkappa_n - \vlambda n)^T \vtheta} \sim 1$
uniformly for $\theta \in (-\varepsilon_n, \varepsilon_n)$.
The change of variable $\vt = \sqrt{n} \vtheta$ is applied
\[
    C_n =
    \frac{n^{-d/2}}{(2 \pi)^d}
    \int_{(- \sqrt{n} \varepsilon_n, \sqrt{n} \varepsilon_n)^d}
    A(\vzeta e^{i n^{-1/2} \vt})
    e^{- \vt^T \mH \vt / 2 + \bigO(n {\varepsilon_n}^3)}
    d \vt.
\]
We have $A(\vzeta e^{i n^{-1/2} \vt}) \to A(\vzeta)$
and $n {\varepsilon_n}^3 \to 0$, so
\[
    C_n \sim
    A(\vzeta)
    \frac{n^{-d/2}}{(2 \pi)^d}
    \int_{(- \sqrt{n} \varepsilon_n, \sqrt{n} \varepsilon_n)^d}
    e^{- \vt^T \mH \vt / 2}
    d \vt.
\]
We recognise a Cauchy integral
and choose $\varepsilon_n$ such that $\sqrt{n} \varepsilon_n \to +\infty$,
so
\[
    C_n
    \sim
    A(\vzeta)
    \frac{n^{-d/2}}{(2 \pi)^d}
    \int_{(-\infty, +\infty)^d}
    e^{- \vt^T \mH \vt / 2}
    d \vt
    \sim
    \frac{A(\vzeta)}{\sqrt{(2 \pi n)^d \det(\mH)}}.
\]
The three constraints on $\varepsilon_n$ are satisfied
by choosing, for example, $\varepsilon_n = n^{- \beta}$
with $\beta \in (\max(\delta, 1/3), 1/2)$.

\paragraphproof{Tail.}
The tail corresponds to
$(-\pi, \pi]^d \setminus (-\varepsilon_n, \varepsilon_n)^d$.
The tail part of the integral~\eqref{eq:central_part_large_powers_limit}
is bounded by
\[
    \int_{(-\pi, \pi]^d \setminus (-\varepsilon_n, \varepsilon_n)^d}
    \left|
    A(\vzeta e^{i \vtheta})
    e^{- n \phi(\vtheta)}
    \right|
    d \vtheta
    \leq
    \left(
    \sup_{\vtheta \in (-\pi, \pi)^d} |A(\vzeta e^{i \vtheta})| 
    \right)
    \int_{(-\pi, \pi]^d \setminus (-\varepsilon_n, \varepsilon_n)^d}
    e^{- n \Real(\phi(\vtheta))}
    d \vtheta.
\]
We will prove in the next paragraph
that for all large enough $n$,
the minimum of $\Real(\phi(\vtheta))$
on the tail is reached at some points from
$[-\varepsilon_n, \varepsilon_n]^d \setminus (-\varepsilon_n, \varepsilon_n)^d$.
Thus, the tail is bounded by
\[
    \bigO(1)
    \sup_{\vtheta \in [-\varepsilon_n, \varepsilon_n]^d \setminus (-\varepsilon_n, \varepsilon_n)^d}
    e^{- n \Real(\phi(\vtheta))}.
\]
We replace $\phi(\vtheta)$ with its Taylor expansion
\[
    - n \Real(\phi(\vtheta)) =
    - \frac{n}{2} \vtheta^T \mH \vtheta
    + \bigO(n {\varepsilon_n}^3).
\]
Our choice of $\varepsilon_n$ ensures the $\bigO$ tends to $0$.
We divide $\vtheta$ by $\varepsilon_n$
to bound the tail by
\[
    \bigO(1)
    \sup_{\vtheta \in [-\varepsilon_n, \varepsilon_n]^d \setminus (-\varepsilon_n, \varepsilon_n)^d}
    \exp \left(
    - \frac{n {\varepsilon_n}^2}{2}
    ({\varepsilon_n}^{-1} \vtheta)^T \mH ({\varepsilon_n}^{-1} \vtheta)
    \right).
\]
The vector ${\varepsilon_n}^{-1} \vtheta$ has
at least one coefficient equal to $-1$ or $1$,
so it stays bounded away from $\vzero$.
Since $\mH$ is positive definite,
there exist a positive constant $c$
such that
$({\varepsilon_n}^{-1} \vtheta)^T \mH ({\varepsilon_n}^{-1} \vtheta) \geq c$,
so the tail is bounded by
\[
    \bigO(1)
    \exp \left(- c \frac{n {\varepsilon_n}^2}{2} \right).
\]
Since $n {\varepsilon_n}^2 \to +\infty$,
this converges exponentially fast to $0$,
so the tail is negligible compared to
the central part of the integral.

\paragraph{Locating the minimum of $\Real(\phi(\vtheta))$ on the tail.}
By definition of $\phi(\vtheta)$, we have
\begin{align*}
    \Real(\phi(\vtheta))
    &=
    \Real \left(
    i \vlambda^T \vtheta
    - \log \left(
        \frac{B(\vzeta e^{i \vtheta})}
            {B(\vzeta)}
    \right)
    \right)
    \\&=
    \log(B(\vzeta))
    - \log \left( \left|
    B(\vzeta e^{i \vtheta})
    \right| \right).
\end{align*}
Thus, the minima of $\Real(\phi(\vtheta))$
correspond to maxima of $|B(\vzeta e^{i \vtheta})|$.
By assumption, on the polydisc of radius $\vzeta$,
the maximum of $|B(\vz)|$
is reached only at $\vz = \vzeta$,
so the minimum of $\Real(\phi(\vtheta))$
is reached only at $\vtheta = \vzero$.
As all points of the central part tend to $\vzeta$ with $n$,
this minimum tends to $\Real(\phi(\vzero)) = 0$.
By contradiction, let us assume that
there is an infinite set of indices $n$
for which there exists a point in
$(-\pi, \pi)^d \setminus [-\varepsilon_n, \varepsilon_n]^d$
where $\Real(\phi(\vtheta))$ reaches its minimum
on the tail $(-\pi, \pi)^d \setminus (-\varepsilon_n, \varepsilon_n)^d$.
We extract a converging subsequence $\vu_n$.
The limit $\vv$ must be $\vzero$,
otherwise $\lim_n \Real(\phi(\vu_n)) = \Real(\phi(\vv)) > 0$,
which contradicts the fact that the minimum on the tail
tends to $0$.
The gradient of $\Real(\phi(\vtheta))$ vanishes
at each $\vtheta = \vu_n$.
From \cref{eq:phi_gaussian}, we deduce
that the gradient of $\Real(\phi(\vtheta))$
at $\vtheta = \vu_n$ is proportional
to $\mH \vu_n + \bigO(\|\vu_n\|^2)$.
Let $q_n$ denote the minimum of the absolute value
of all coefficients of $\vu_n$,
then for all $n$
\[
    \mH \frac{\vu_n}{q_n} + \bigO(\|\vu_n\|) = \vzero.
\]
Left-multiplying by $\mH^{-1}$,
we deduce that $\vu_n / q_n$ converges to $\vzero$,
which is impossible as at least
one of its coefficients is equal to $1$.
This concludes the proof that
for all large enough $n$,
the minima of $\Real(\phi(\vtheta))$
on $(-\pi, \pi]^d \setminus (-\varepsilon_n, \varepsilon_n)^d$
belongs to $[-\varepsilon_n, \varepsilon_n]^d \setminus (-\varepsilon_n, \varepsilon_n)^d$.
\end{proof}

    \subsection{Combining Large Powers and Singularity Analysis}\label{sec_appendix-Elie-singularity}
The following result generalizes \cref{th_large-powers-limit}
to the case where the saddle point meets singularities.
It is the analytic result behind \Cref{th_critical,th_mixed-case,th_proba-X-Y}.
This is a multivariate combination of the large powers theorem \cite[Theorem VIII.8]{flajolet2009analytic} and singularity analysis \cite[Section VI]{flajolet2009analytic}.
The presence of large powers simplifies the singularity analysis part
and we avoid the complications related to multivariate singularity analysis
\cite{pemantle2024analytic}.
We will derive a more precise version of it (asymptotic expansion)
in \Cref{sec_asymptotic_expansion}.

\begin{theorem}
\label{th_large-powers-singularity}
Consider a power series $B(z_1, \ldots, z_d)$
with nonnegative coefficients,
full rank, aperiodic
and convergent on a neighborhood $\Omega$ of $0$.
Let $K(\vt) = \log(B(e^{\vt}))$ denote
its cumulant generating function.
Let $\vlambda$ denote a vector of positive real values
such that the function $K(\vt) - \vlambda^T \vt$
does not tend to its infimum on the boundary
of $\Psi = \{ \vt \in \reals^d,\ e^{\vt} \in \Omega \}$.
Let $\vtau$ denote the unique point of $\Psi$
where $K(\vt) - \vlambda^T \vt$ reaches its minimum.
Set $\vzeta = e^{\vtau}$.

We assume $\zeta_j \leq 1$ for all $j$ and denote by $\mS$
the set of indices such that $\zeta_j < 1$,
and by $\mC$ the corresponding set for $\zeta_j = 1$.
Note that by assumption, we have $\mS \cup \mC = \{1, \ldots, d\}$.
Let $\mH$ denote the Hessian matrix of $K(\vt)$ at $\vtau$
and $\mM$ the submatrix of $\mH^{-1}$ corresponding to the rows and column in $\mC$.
Consider a power series $A(\vz)$
analytic on the torus of radius $\vzeta$
and satisfying $A(\vzeta) \neq 0$.
Consider a real value $\delta < 1/2$
and a sequence of vectors of real values
\[
    \vkappa_n = \vlambda n + \bigO(n^{\delta}).
\]
Then
\[
    [\vz^{\vkappa_n}]
    \frac{A(\vz) B(\vz)^n}
    {\prod_{j \in \mS \cup \mC} (1 - z_j)}
    \sim
    \frac{B(\vzeta)^n}{\vzeta^{\vkappa_n}}
    \frac{A(\vzeta)}
        {\prod_{j \in \mS} (1 - \zeta_j)}
    \frac{1}
        {\sqrt{ (2 \pi)^d n^{|\mS|} \det(\mH)}}
    \int_{(0, +\infty)^{|\mC|}}
    e^{- \vu^T \mM \vu / 2}
    d \vu.
\]
\end{theorem}

\begin{proof}
We write the coefficient extraction as a Cauchy integral
and apply again the change of variable $\vz = \vzeta e^{i \vtheta}$.
However, for $j \in \mC$, as $\zeta_j = 1$ and $z_j = 1$
is a singularity of the integrand, we cannot let $\theta_j$ go through $0$.
So in that case, fixing $\varepsilon_n = n^{- \beta}$
with $\beta \in (\max(\delta, 1/3), 1/2)$,
as in the proof of \cref{th_large-powers-limit},
we define the path of integration for $\theta_j$
as $(-\pi, -\varepsilon_n)$, then a half-circle centered at $0$,
starting at $- \varepsilon_n$, going through $i \varepsilon_n$
and stopping at $\varepsilon_n$, then $(\varepsilon_n, \pi]$.
For $j \in \mS$, we have $\zeta_j < 1$ so we choose
the integration path $\theta_j \in (-\pi, \pi]$.
The function $\phi(\vt)$ is defined as in \eqref{eq:phi}
and \cref{eq:central_part_large_powers_limit} becomes
\[
    [\vz^{\vkappa_n}]
    \frac{A(\vz) B(\vz)^n}
    {\prod_{j \in \mS \cup \mC} (1 - z_j)}
    =
    \frac{B(\vzeta)^n}
    {\vzeta^{\vkappa_n}}
    \frac{1}{(2 \pi)^d}
    \int
    \frac{A(\vzeta e^{i \vtheta}) e^{- n \phi(\vtheta)}}
    {\prod_{j \in \mS \cup \mC} (1 - \zeta_j e^{i \theta_j})}
    \frac{d \theta}
    {e^{i (\vkappa_n - \vlambda n)^T \vtheta}}.
\]

\paragraphproof{Central part.}
It corresponds to the following part of the domain of integration,
denoted by $\gamma_n$.
For $j \in \mC$, we choose $\theta_j$ on the half-circle described above.
For $j \in \mS$, we choose $\theta_j \in (-\varepsilon_n, \varepsilon_n)$.
On this domain, we apply a Taylor expansion of the integrand.
Note that for $j \in \mC$, as $\zeta_j = 1$,
we have $1 - \zeta_j e^{i \theta_j} \sim - i \theta_j$,
while for $j \in \mS$,
we have $1 - \zeta_j e^{i \theta_j} \sim 1 -\zeta_j$.
\begin{equation}
\label{eq:asympt_one}
    \frac{1}{(2 \pi)^d}
    \int_{\gamma_n}
    \frac{A(\vzeta e^{i \vtheta}) e^{- n \phi(\vtheta)}}
    {\prod_{j \in \mS \cup \mC} (1 - \zeta_j e^{i \theta_j})}
    \frac{d \theta}
    {e^{i (\vkappa_n - \vlambda n)^T \vtheta}}
    \sim
    \frac{1}{(2 \pi)^d}
    \int_{\gamma_n}
    \frac{A(\vzeta)
        e^{- n \vtheta^T \mH \vtheta / 2 + \bigO(n \|\vtheta\|^3)}}
    {\prod_{j \in \mC} (- i \theta_j)
    \prod_{j \in \mS} (1 - \zeta_j)}
    d \theta.
\end{equation}
Note that our choice of $\varepsilon_n$ ensures
that $\bigO(n \|\vtheta\|^3)$ tends to $0$.
The change of variable $\vt = \sqrt{n} \vtheta$ is applied.
The integral is now equivalent with
\[
    \frac{A(\vzeta)}
        {\prod_{j \in \mS} (1 - \zeta_j)}
    \frac{i^{|\mC|}}
        {(2 \pi)^d n^{|\mS|/2}}
    \int
    \frac{e^{- \vt^T \mH \vt / 2}}
        {\prod_{j \in \mC} t_j}
    d \vt,
\]
where the domain of integration is the following.
For $j \in \mC$,
$t_j$ is on the line from $-\sqrt{n} \varepsilon_n$ to $i$,
then on the line from $i$ to $\sqrt{n} \varepsilon_n$.
For $j \in \mS$,
$t_j \in (\sqrt{n} -\varepsilon_n, \sqrt{n} \varepsilon_n)$.
The tails of this integral are exponentially small,
so we add them to simplify the expression.
The central asymptotics becomes
\[
    \frac{A(\vzeta)}
        {\prod_{j \in \mS} (1 - \zeta_j)}
    \frac{i^{|\mC|}}
        {(2 \pi)^d n^{|\mS|/2}}
    \int_{(i - \infty, i + \infty)^{|\mC|} \times (-\infty, +\infty)^{|\mS|}}
    \frac{e^{- \vt^T \mH \vt / 2}}
        {\prod_{j \in \mC} t_j}
    d \vt.
\]
To further simplify the expression,
we inject the following identity,
valid for any $t$ with positive imaginary part
\[
    \frac{1}{t} =
    - i
    \int_0^{+\infty}
    e^{i t u}
    d u.
\]
Applying this identity for $t = t_j$ and $j \in \mC$ yields
\[
    \frac{A(\vzeta)}
        {(2 \pi)^d n^{|\mS|/2} \prod_{j \in \mS} (1 - \zeta_j)}
    \int_{(i - \infty, i + \infty)^{|\mC|} \times (-\infty, +\infty)^{|\mS|}}
    \int_{(0, +\infty)^{|\mC|}}
    \exp \bigg(
        - \frac{1}{2} \vt^T \mH \vt
        + i \sum_{j \in \mC} u_j t_j
    \bigg)
    d \vu
    d \vt.
\]
Let $(\vu_{\mC}, \vzero_{\mS})$ denote the vector of dimension $d$
with $u_j$ at any index $j \in \mC$
and $0$ at any index $j \in \mS$.
The matrix $\mH$ is definite positive,
so its inverse and square root are well defined.
We interchange the integrals and
apply the change of variable $\vs = \sqrt{\mH} \vt$
\begin{align*}
    \frac{A(\vzeta)}
        {\prod_{j \in \mS} (1 - \zeta_j)}
&
    \frac{1}
        {(2 \pi)^d \sqrt{n^{|\mS|} \det(\mH)}}
\\&\times
    \int_{(0, +\infty)^{|\mC|}}
    \int_{(i - \infty, i + \infty)^{|\mC|} \times (-\infty, +\infty)^{|\mS|}}
    \exp \left(
        - \frac{1}{2} \vs^T \vs
        + i (\vu_{\mC}, \vzero_{\mS})^T \mH^{-1/2} \vs
    \right)
    d \vs
    d \vu.
\end{align*}
The argument of the exponential is factorized
\begin{align*}
    \frac{A(\vzeta)}
        {\prod_{j \in \mS} (1 - \zeta_j)}
&
    \frac{1}
        {(2 \pi)^d \sqrt{n^{|\mS|} \det(\mH)}}
    \iint
\\
&    \exp \left(
        - \frac{1}{2}
        \left(\vs - i \mH^{-1/2} (\vu_{\mC}, \vzero_{\mS}) \right)^T
        \left(\vs - i \mH^{-1/2} (\vu_{\mC}, \vzero_{\mS}) \right)
        - \frac{1}{2} (\vu_{\mC}, \vzero_{\mS})^T \mH^{-1} (\vu_{\mC}, \vzero_{\mS})
    \right)
    d \vs
    d \vu.
\end{align*}
The integrals with respect to each $s_j$ are now separated,
each a Gaussian integral with value $\sqrt{2 \pi}$.
We deduce that the central part has asymptotics
\[
    \frac{A(\vzeta)}{\prod_{j \in \mS} (1 - \zeta_j)}
    \frac{1}{ \sqrt{ (2 \pi)^d n^{|\mS|} \det(\mH)} }
    \int_{(0, +\infty)^{|\mC|}}
    \exp \left(
        - \frac{1}{2} (\vu_{\mC}, \vzero_{\mS})^T \mH^{-1} (\vu_{\mC}, \vzero_{\mS})
    \right)
    d \vu.
\]

\paragraphproof{Tail.}
The proof that the tail is
exponentially small compared to the central part
is the same as in the proof of \cref{th_large-powers-limit}.
\end{proof}

    \subsection{Asymptotic Expansion}
    \label{sec_asymptotic_expansion}

We now explain how to obtain more precise asymptotic results,
applied in \Cref{sec_error-terms}.
Since \cref{th_large-powers-limit} is a particular case
of \cref{th_large-powers-singularity},
we focus on the latter.
%
In the critical and mixed cases,
we find an asymptotic expansion where each term corresponds to a power of $1/\sqrt{n}$.
This is somewhat surprising, as for singularity analysis \cite[Section VI]{flajolet2009analytic}
and large powers \cite[Theorem VIII.8]{flajolet2009analytic},
the terms of the asymptotic expansion correspond to powers of $1/n$.

\begin{theorem}
\label{th:asympt_expansion}
We keep the notations of \cref{th_large-powers-singularity}
and replace our assumption
$\vkappa_n = \vlambda n + \bigO(n^{\delta})$
for some $\delta < 1/2$
with
\[
    \vkappa_n = \vlambda n + \bigO(1).
\]
Then for any $r \geq 1$,
there exist functions $a_{k,n}$ of $\vkappa_n - \vlambda n$
(which are bounded with respect to $n$) such that
\begin{equation}
\label{eq:asympt_expansion}
    [\vz^{\vkappa_n}]
    \frac{A(\vz) B(\vz)^n}
    {\prod_{j \in \mS \cup \mC} (1 - z_j)}
    =
    \frac{B(\vzeta)^n}{\vzeta^{\vkappa_n}}
    \frac{1}{n^{|\mS|/2}}
    \left(
    \sum_{k=0}^{r-1}
    a_{k,n}
    n^{-k/2}
    + \bigO(n^{-r/2})
    \right).
\end{equation}
An algorithm computing them is provided in the proof.
In particular, $a_0 := a_{0,n}$ is independent of $n$ and equal to
\begin{align*}
    a_0
%
%
    &=
    \frac{1}{\sqrt{(2 \pi)^d \det(\mH)}}
    \int_{(0, +\infty)^{|\mC|}}
    e^{- \vu^T \mM \vu}
    d \vu
\end{align*}
and we will express $a_{1,n}$ in \cref{th:compute_a1}.
Furthermore, in the particular case $\mC = \emptyset$,
then $a_{k,n} = 0$ for all odd $k$.
\end{theorem}

\begin{proof}
We fix some $r \geq 1$.
Let us say that a sequence $u_n$
\emph{tends superpolynomially fast} to $0$
if for any integer $k$, the sequence $n^k u_n$ tends to $0$.
For example, $e^{-\sqrt{n}}$ tends superpolynomially fast to $0$.
We proved that the tails tends superpolynomially fast to $0$,
so its contribution to \cref{eq:asympt_expansion}
is negligible.
We apply the change of variable $\vt = \sqrt{n} \vtheta$
in the left-hand side of \cref{eq:asympt_one}
and deduce
\[
    [\vz^{\vkappa_n}]
    \frac{A(\vz) B(\vz)^n}
    {\prod_{j \in \mS \cup \mC} (1 - z_j)}
    =
    \frac{B(\vzeta)^n}{\vzeta^{\vkappa_n}}
    \frac{1}{(2 \pi)^d n^{d/2}}
    \int
    \frac{A(\vzeta e^{i \vt / \sqrt{n}}) e^{- n \phi(\vt / \sqrt{n})}}
    {\prod_{j \in \mS \cup \mC} (1 - \zeta_j e^{i t_j / \sqrt{n}})}
    \frac{d \vt}
    {e^{i (\vkappa_n - \vlambda n)^T \vt / \sqrt{n}}}
    \left(1 + \bigO(n^{-r/2}) \right)
\]
where the domain of integration is the following.
For $j \in \mC$,
$t_j$ is on the line from $-\sqrt{n} \varepsilon_n$ to $i$,
then on the line from $i$ to $\sqrt{n} \varepsilon_n$.
For $j \in \mS$,
$t_j \in (- \sqrt{n} \varepsilon_n, \sqrt{n} \varepsilon_n)$.
Let us set
\begin{equation}
\label{eq:Fvty}
    F(\vt, y) :=
    \frac{A(\vzeta e^{i \vt y})
    \exp \left( \frac{1}{2} \vt^T \mH \vt - \phi(\vt y) y^{-2} \right)}
    {\left(\prod_{j \in \mC} \frac{1 - e^{i t_j y}}{- i t_j y} \right)
    \prod_{j \in \mS} (1 - \zeta_j e^{i t_j y}) }
    \frac{1}{e^{i (\vkappa_n - \vlambda n)^T \vt y}}.
\end{equation}
Recalling that $\vzeta_j = 1$ for $j \in \mC$
and injecting $F(\vt, y)$ in our previous expression leads to
\begin{equation}
\label{eq:FFF}
    [\vz^{\vkappa_n}]
    \frac{A(\vz) B(\vz)^n}
    {\prod_{j \in \mS \cup \mC} (1 - z_j)}
    =
    \frac{B(\vzeta)^n}{\vzeta^{\vkappa_n}}
    \frac{i^{|\mC|}}{(2 \pi)^d n^{|\mS|/2}}
    \int
    \frac{F(\vt, n^{-1/2})}{\prod_{j \in \mC} t_j}
    e^{- \vt^T \mH \vt / 2}
    d \vt
    \left(1 + \bigO(n^{-r/2}) \right).
\end{equation}
The term
$\frac{1 - e^{i t_j y}}{- i t_j y}$ has a Taylor expansion at $y=0$.
Observe also that the valuation of
$\frac{1}{2} \vt^T \mH \vt - \phi(\vt y) y^{-2}$
in $y$ is at least $1$,
and that in the series expansion at $(\vt, y) = (\vzero, 0)$,
each monomial where the sum of the exponents of the $t_j$
is $p$ has $y$ at the power $p-2$.
This implies that $F(\vt, y)$ has a Taylor expansion at $y = 0$,
denoted by
\[
    F(\vt, y) =
    \sum_{k = 0}^{r-1}
    F_k(\vt) y^k
    + \bigO \left( (t_1 \cdots t_d)^{\beta_r} y^{r} \right)
\]
where $\beta_r$ is some sequence depending on $r$
(we could prove $\beta_r \leq 3 r$ but will not need it)
and each $F_k(\vt)$ is a polynomial in $\vt$
and in $\vkappa_n - \vlambda n$,
which is a bounded sequence in $n$.
%
With our domain of integration,
we have for any $\beta_r$
\[
    \int
    \frac{\bigO \left( (t_1 \cdots t_d)^{\beta_r} n^{-r/2} \right)}{\prod_{j \in \mC} t_j}
    e^{- \vt^T \mH \vt / 2}
    d \vt
    =
    \bigO(n^{-r/2})
\]
so, injecting in \cref{eq:FFF}
the Taylor expansion of $F(\vt, \vy)$
at $\vy = n^{-1/2}$, we obtain
\[
    [\vz^{\vkappa_n}]
    \frac{A(\vz) B(\vz)^n}
    {\prod_{j \in \mS \cup \mC} (1 - z_j)}
    =
    \frac{B(\vzeta)^n}{\vzeta^{\vkappa_n}}
    \frac{1}{n^{|\mS|/2}}
    \left(
    \sum_{k=0}^{r-1}
    \left[
    \frac{i^{|\mC|}}{(2 \pi)^d}
    \int
    \frac{F_k(\vt)}{\prod_{j \in \mC} t_j}
    e^{- \vt^T \mH \vt / 2}
    d \vt
    \right] 
    n^{-k/2}
    + \bigO(n^{-r/2})
    \right).
\]
The tails of the integrals are exponentially small,
so we add them to simplify the expression.
After deformation, the domain of integration becomes
$(i - \infty, i + \infty)^{|\mC|} \times (-\infty, +\infty)^{|\mS|}$.
Defining
\[
    a_{k,n} :=
    \frac{i^{|\mC|}}{(2 \pi)^d}
    \int_{(i - \infty, i + \infty)^{|\mC|} \times (-\infty, +\infty)^{|\mS|}}
    \frac{F_k(\vt)}{\prod_{j \in \mC} t_j}
    e^{- \vt^T \mH \vt / 2}
    d \vt,
\]
we obtain \cref{eq:asympt_expansion}.

The expression of $a_0$ was obtained and simplified
in the proof of \cref{th_large-powers-singularity}.
Let us now show how to apply the same ideas to the other coefficients.
Write for any complex $t$ with positive imaginary part
\[
    \frac{i}{t}
    =
    \int_0^{+\infty}
    e^{i t u}
    d u,
\]
then
\[
    a_{k,n} =
    \frac{1}{(2 \pi)^d}
    \iint_{(0,+\infty)^{|\mC|}}
    \vt^{\vp}
    F_k(\vt)
    e^{- \vt^T \mH \vt / 2 + i (\vu_{\mC}, \vzero_{\mS})^T \vt}
    d \vu
    d \vt.
\]
We apply the change of variable $\vs = \sqrt{\mH} \vt$
and factorize the argument of the exponential
as in the proof of \cref{th_large-powers-singularity}
\begin{align*}
    a_{k,n}
    &=
    \frac{1}{(2 \pi)^d \sqrt{\det(\mH)}}
    \iint_{(0,+\infty)^{|\mC|}}
    F_k(\mH^{-1/2} \vs)
    e^{- \vs^T \vs / 2 + i (\vu_{\mC}, \vzero_{\mS})^T \mH^{-1/2} \vs}
    d \vu
    d \vs
\\
    &=
    \frac{1}{(2 \pi)^d \sqrt{\det(\mH)}}
    \iint_{(0,+\infty)^{|\mC|}}
    F_k(\mH^{-1/2} \vs)
    e^{- (\vs - i \mH^{-1/2} (\vu_{\mC}, \vzero_{\mS})^T (\vs - i \mH^{-1/2} (\vu_{\mC}, \vzero_{\mS})) / 2 - \vu^T \mM \vu / 2}
    d \vu
    d \vs.
\end{align*}
We interchange the integrals and shift
the vector $\vs$ by $i \mH^{-1/2} (\vx, \vzero)$
\[
    a_{k,n}
    =
    \frac{1}{(2 \pi)^d \sqrt{\det(\mH)}}
    \int_{(0,+\infty)^{|\mC|}}
    \int
    F_k\left( \mH^{-1/2} \vs + i \mH^{-1} (\vu_{\mC}, \vzero_{\mS}) \right)
    e^{- \vs^T \vs / 2 - \vu^T \mM \vu / 2}
    d \vs
    d \vu.
\]
Recall that $F_k(\vt)$ is a polynomial,
so there exist polynomials $G_{k, \vp}(\vu)$
(which are also polynomials in $\vkappa_n - \vlambda n$)
and a finite set $S_k \subset \integers_{\geq 0}^d$ such that
\[
    F_k\left( \mH^{-1/2} \vs + i \mH^{-1} (\vu_{\mC}, \vzero_{\mS}) \right)
    =
    \sum_{\vp \in S_k}
    \vs^{\vp}
    G_{k, \vp}(\vu).
\]
We deduce
\[
    a_{k,n} =
    \frac{1}{(2 \pi)^d \sqrt{\det(\mH)}}
    \sum_{\vp \in S_k}
    \int_{(0,+\infty)^{|\mC|}}
    \int
    \vs^{\vp}
    G_{k, \vp}(\vu)
    e^{- \vs^T \vs / 2 - \vu^T \mM \vu / 2}
    d \vs
    d \vu.
\]
Each integral with respect to $s_j$, for $j \in \mS \cup \mC$,
is simplified using the formula
\[
    \int_{(i - \infty, i + \infty)}
    s^p e^{- s^2 / 2} d s
    =
    \begin{cases}
    0 & \text{if $p$ is odd,}\\
    \sqrt{2 \pi}
    \frac{(2 q)!}{2^q q!}
    & \text{if $p = 2 q$.}
    \end{cases}
\]
We observe that in the expression of $a_{k,n}$,
only vectors $\vp$ where all coefficients are even may contribute.
Let us define
\[
    S_k^{\star} =
    \{ \vq,\ 2 \vq \in S_k \},
\]
then
\begin{equation}
\label{eq:ak}
    a_{k,n} =
    \frac{1}{ \sqrt{(2 \pi)^d \det(\mH)} }
    \sum_{\vq \in S_k^{\star}}
    \bigg(
        \prod_{j \in \mS \cup \mC}
        \frac{(2 q_j)!}{2^{q_j} q_j!}
    \bigg)
    \int_{(0,+\infty)^{|\mC|}}
    G_{k, 2 \vq}(\vu)
    e^{ - \vu^T \mM \vu / 2}
    d \vu.
\end{equation}

In the particular case $\mC = \emptyset$, the function $F_k$
is simply composed with $\mH^{-1/2} \vs$.
If $k$ is odd, by construction, the sum of the exponents
in each monomial of $F_k(\vt)$ is odd.
This implies that each monomial of $F_k(\mH^{-1/2} \vs)$
contains some $s_j$ at an odd power,
so $a_{k,n} = 0$.

An alternative approach to compute the coefficients $a_{k,n}$
is to apply the change of variable $- \vt^T \vt / 2 = \phi(\vtheta)$
in the left-hand side of \cref{eq:asympt_one}.
We refer the interested reader to \cite[Chapter 5]{analytic2013pemantle}.
\end{proof}

Our next result provides an explicit expression
for the term $a_{1,n}$ of the asymptotic expansion
from \cref{th:asympt_expansion}.

\begin{lemma}
\label{th:compute_a1}
With $\vzeta$ defined as in \cref{th:asympt_expansion},
define the polynomial
\[
    \phi_3(\vt) :=
    [y^3]
    \log \left( \frac{B(\vzeta e^{\vt y})}{B(\vzeta)} \right).
\]
Observe that the sum of the exponents of each of its monomials is $3$.
Let $\phi_3^{\star}(\vu)$ denote the polynomial
obtained by replacing each monomial $t_j t_k t_{\ell}$
(with $j$, $k$, $\ell$ in $\mS \cup \mC$ not necessarily distinct) with
\begin{align}
    \label{eq:replace_in_phi3t}
    &- \ve_j^T \mH^{-1} (\vu_{\mC}, \vzero_{\mS})
    \ve_k^T \mH^{-1} (\vu_{\mC}, \vzero_{\mS})
    \ve_{\ell}^T \mH^{-1} (\vu_{\mC}, \vzero_{\mS})
    \\&
    + \left(
    (\mH^{-1})_{k,\ell} \ve_j
    + (\mH^{-1})_{j,\ell} \ve_k
    + (\mH^{-1})_{j,k} \ve_{\ell}
    \right)^T
    \mH^{-1} (\vu_{\mC}, \vzero_{\mS}). \nonumber
\end{align}
Let $\vb$ denote the vector of size $d$ with $b_j = -1/2$ for $j \in \mC$
and $b_j = \frac{\zeta_j}{1 - \zeta_j}$ for $j \in \mS$.
Then the term $a_{1,n}$ in the asymptotic expansion from \cref{th:asympt_expansion}
is equal to
\begin{align*}
    a_{1,n} =&
    \frac{A(\vzeta)}{\prod_{j \in \mS} (1 - \zeta_j)}
    \frac{1}{ \sqrt{ (2 \pi)^d \det(\mH) }}
\\& \times
    \int_{(0,+\infty)^{|\mC|}}
    \left(
    \phi_3^{\star}(\vu)
    -
    \left[
    \frac{\nabla_A(\vzeta)}{A(\vzeta)}
    + \vb
    - (\vkappa_n - \vlambda n)
    \right]^T \mH^{-1} (\vu_{\mC}, \vzero_{\mS})
    \right)
    e^{- \vu \mM \vu / 2}
    d \vu.
\end{align*}
\end{lemma}

\begin{proof}
Our first step is to compute $F_1(\vt) = [y^1] F(\vt, y)$,
where $F(\vt, y)$ is defined in \cref{eq:Fvty}.
Using the expression of $\phi_3(\vt)$ from the lemma,
we have the following Taylor expansions
\begin{align*}
    A(\vzeta e^{i \vt y}) &=
    A(\vzeta) + i \nabla_A(\vzeta) \vt y + \bigO(y^2),
\\
    \exp \left(
    \frac{1}{2} \vt^T \mH \vt - \phi(\vt y) y^{-2}
    \right)
    &=
    1 - i \phi_3(\vt) y + \bigO(y^2),
\\
    \frac{1}{\frac{1 - e^{i t_j y}}{-i t_j y}}
    &=
    1 - \frac{i}{2} t_j y + \bigO(y^2),
\\
    \frac{1}{1 - \zeta_j e^{i t_j y}} &=
    \frac{1}{1 - \zeta_j} + \frac{\zeta_j}{(1 - \zeta_j)^2} i t_j y + \bigO(y^2),
\\
    \frac{1}{e^{i (\vkappa_n - \vlambda n)^T \vt y}} &=
    1 - i (\vkappa_n - \vlambda n)^T \vt y + \bigO(y^2).
\end{align*}
Using the vector $\vb$ from the current lemma,
we obtain the following Taylor expansion for the products
\[
    \frac{1}{\prod_{j \in \mC} \frac{1 - e^{i t_j y}}{- i t_j y}}
    \frac{1}{\prod_{j \in \mS} (1 - \zeta_j e^{i t_j y})}
    =
    1 + i \vb^T \vt y + \bigO(y^2).
\]
We deduce
\[
    F_1(\vt) =
    \frac{- i A(\vzeta)}{\prod_{j \in \mS} (1 - \zeta_j)}
    \left(
    \phi_3(\vt)
    - \left[
    \frac{\nabla_A(\vzeta)}{A(\vzeta)}
    + \vb
    - (\vkappa_n - \vlambda n)
    \right]^T \vt
    \right).
\]

Following the proof of \cref{th:asympt_expansion},
our next step is to compute
\[
    \sum_{\vq \in S_1^{\star}}
    \bigg(
    \prod_{j \in \mS \cup \mC}
    \frac{(2 q_j)!}{2^{q_j} q_j!}
    \bigg)
    G_{1, 2 \vq}(\vu).
\]
It is obtained by the following transformation of $F_1(\vt)$.
Replace $\vt$ with $\mH^{-1/2} \vs + i \mH^{-1} (\vu_{\mC}, \vzero_{\mS})$,
then for each $j$ and $k$, replace $s_j^k$ with $0$ if $k$ is odd,
and with $\frac{(2q)!}{2^q q!}$ if $k = 2q$.
Clearly, any $t_j$ at the power $1$ will be transformed into $i \mH^{-1} (\vu_{\mC}, \vzero_{\mS})$. By construction, $\phi_3(\vt)$ is a polynomial where the sum of exponents in each momonial is $3$.
Consider three integers $j, k, \ell$ in $\mS \cup \mC$.
The monomial $t_j t_k t_{\ell}$ will be transformed into
\[
    \left( \ve_j^T \mH^{-1/2} \vs + i \ve_j^T \mH^{-1} (\vu_{\mC}, \vzero_{\mS}) \right)
    \left( \ve_k^T \mH^{-1/2} \vs + i \ve_k^T \mH^{-1} (\vu_{\mC}, \vzero_{\mS}) \right)
    \left( \ve_{\ell}^T \mH^{-1/2} \vs + i \ve_{\ell}^T \mH^{-1} (\vu_{\mC}, \vzero_{\mS}) \right)
\]
where each $s_h$ at an odd power is replaced by $0$
and each $s_h^2$ is replaced by $1$, for all $h \in \mS \cup \mC$.
The expression becomes
\begin{align*}
    - i
    \ve_j^T \mH^{-1} (\vu_{\mC}, \vzero_{\mS})
    \ve_k^T \mH^{-1} (\vu_{\mC}, \vzero_{\mS})
    \ve_{\ell}^T \mH^{-1} (\vu_{\mC}, \vzero_{\mS})
    &+ i \ve_j^T \mH^{-1} (\vu_{\mC}, \vzero_{\mS})
    \sum_{h=1}^d \ve_k^T \mH^{-1/2} \ve_h \ve_{\ell}^T \mH^{-1/2} \ve_h
    \\&+ i \ve_k^T \mH^{-1} (\vu_{\mC}, \vzero_{\mS})
    \sum_{h=1}^d \ve_j^T \mH^{-1/2} \ve_h \ve_{\ell}^T \mH^{-1/2} \ve_h
    \\&+ i \ve_{\ell}^T \mH^{-1} (\vu_{\mC}, \vzero_{\mS})
    \sum_{h=1}^d \ve_j^T \mH^{-1/2} \ve_h \ve_k^T \mH^{-1/2} \ve_h
\end{align*}
which, given the symmetry of $\mH^{-1/2}$,
is equal to \eqref{eq:replace_in_phi3t}.
With $\phi_3^{\star}(\vu)$ defined as in the lemma,
we deduce that after evaluating $F_1(\vt)$ at $\vt = \mH^{-1/2} \vs + i \mH^{-1} (\vu_{\mC}, \vzero_{\mS})$, then replacing for each $j$ any odd power of $s_j$ with $0$ and $s_j^2$ with $1$, we obtain
\[
    \sum_{\vq \in S_1^{\star}}
    \bigg(
    \prod_{j \in \mS \cup \mC}
    \frac{(2 q_j)!}{2^{q_j} q_j!}
    \bigg)
    G_{1, 2 \vq}(\vu)
=
    \frac{A(\vzeta)}{\prod_{j \in \mS} (1 - \zeta_j)}
    \left(
    \phi_3^{\star}(\vu)
    - \left[
    \frac{\nabla_A(\vzeta)}{A(\vzeta)}
    + \vb
    - (\vkappa_n - \vlambda n)
    \right]^T
    \mH^{-1} (\vu_{\mC}, \vzero_{\mS})
    \right).
\]
Injecting this in \cref{eq:ak} for $k=1$
yields the expression of $a_{1,n}$ from the current lemma.
\end{proof}

\newpage

\section{Properties and Proofs for the General Independent Culture (GIC)}\label{sec_appendix-gic}

\subsection{Satisfying the Conditions of the Asymptotic Theorems in \Cref{sec_appendix-Elie}}
\label{sec_P_theorem_conditions}

In this paper, the characteristic polynomial \( P(\vx) \) encodes a probability distribution over the preference rankings of a voter. This guarantees its convergence on a neighborhood \( \Omega \) of \( 0 \).  
Since we assume that the distribution of voter preferences over rankings is generic, \ie every ranking \( r \) has a positive probability, it follows that all coefficients of the characteristic polynomial \( P(\vx) \) from \Cref{def_def-of-P} are positive.  
The following result establishes that this assumption ensures that the remaining conditions required for the application of our asymptotic theorems (\Cref{th_large-powers-limit,th_large-powers-singularity,th:asympt_expansion}) are satisfied.

\begin{lemma}
\label{th_conditions-P-satisfied}
Consider a polynomial $P(x_1, \ldots, x_d)$
where for any $\mX \subseteq \{1, \ldots, d\}$,
the monomial $\prod_{j \in \mX} x_j$ has a positive coefficient,
while any other monomial has a zero coefficient.
As usual, define its cumulant generating function as
\[
    K : \vt \mapsto \log(P(e^{\vt})).
\]
Then $P(\vx)$ has full rank (\Cref{def:rank})
and is aperiodic (\Cref{def:periodicity}).
Furthermore, for any $\vlambda \in (0, 1)^d$
and $\vl$ on the boundary of $[-\infty, +\infty]^d$,
the function $t \mapsto K(\vt) - \vlambda^T \vt$
tends to infinity as $\vt$ tends to $\vl$.
Thus, there exists a unique vector $\vtau \in \reals^d$
where the gradient of $K$ vanishes.
\end{lemma}

\begin{proof}
By assumption, the support of $P(\vx)$
contains $\vzero$ and $\ve_j$ for each $j \in \{1, \ldots, d\}$.
\Cref{th_full_rank_aperiodic} then implies
that $P(\vx)$ has full rank and is aperiodic.

Now, consider $\vt$ converging to $\vl$.
Let $\mX$ denote the sets of indices $j$
such that $\vl_j = +\infty$.
Let $p_{\mX}$ denote the coefficient of $P(\vx)$
corresponding to the monomial of exponents in $\mX$.
Let $\vone_{\mX}$ denote the vector of dimension $d$
with value $1$ at indices from $\mX$,
and $0$ otherwise.
Then
\[
    P(e^{\vt})
    \geq
    p_{\mX} \exp \left(\vone_{\mX}^T \vt \right),
\]
so
\[
    K(\vt) - \vlambda^T \vt
    \geq
    \log(p_{\mX})
    + \vone_{\mX}^T \vt
    - \vlambda^T \vt
    \geq
    \log(p_{\mX})
    + (\vone_{\mX} - \vlambda)^T \vt.
\]
The vector $\vone_{\mX} - \vlambda$
has positive coefficients at the indices where $\vt$ tends to $+\infty$,
and negative coefficients at the indices where $\vt$ tends to a constant or $-\infty$.
Since $\vl$ is on the boundary of $[-\infty, +\infty]^d$,
it has at least one element equal to $-\infty$ or $+\infty$,
so $(\vone_{\mX} - \vlambda)^T \vt$ tends to infinity.
This implies that $K(\vt) - \vlambda^T \vt$
converges to $+\infty$ as well.

According to \Cref{th_convexity-cumulant-gf},
$K$ is strictly convex.
If its minimum is not reached
on the boundary of its domain of definition,
it must be reached in its interior and is unique.
Thus, there exists a unique point $\vtau$
where the gradient of $K$ vanishes.
\end{proof}

\subsection{Connection Between the Hessian Matrix and the Correlation Matrix}
\label{sec_link-matrices}

This section demonstrate the connection discussed in \Cref{sec_critical-case} between our Hessian matrix $\mH_{K}(\vtau)$ and the correlation matrix $R_m$ of \cite{niemi1968paradox} and \cite{krishnamoorthy2005condorcet}.

For each $j \in \mA$, let $X_j$ denote the random variable
taking value $1$ if candidate $j$ is ranked higher than candidate $m$,
and $0$ otherwise.
Set $\vY = 2 \vX - \vone$.
In this paper, we consider the covariance matrix $\mH_K(\vtau)$ of $\vX$,
whose coefficient $(j,k)$ is
\[
     \mean(X_j X_k) - \mean(X_j) \mean(X_k),
\]
while \cite{niemi1968paradox} and \cite{krishnamoorthy2005condorcet}
consider the correlation matrix $R_m$ of $\vY$,
whose coefficient $(j,k)$ is
\[
    \frac{\mean(Y_j Y_k) - \mean(Y_j) \mean(Y_k)}
    {\sqrt{(1 - \mean(Y_j)^2) (1 - \mean(Y_k)^2}}.
\]
Observe that
\begin{align*}
    \mean(Y_j Y_k) - \mean(Y_j) \mean(Y_k)
    &=
    \mean((2 X_j - 1) (2 X_k - 1)) - \mean(2 X_j - 1) \mean(2 X_k - 1)
    \\&=
    4 \left(
    \mean(X_j X_k) - \mean(X_j) \mean(X_k)
    \right).
\end{align*}
Let $D$ denote the invertible diagonal matrix
whose $j$th element is $(1 - \mean(Y_j)^2)^{-1/2}$,
then
\[
    R_m = 4 D \mH_K(\vtau) D.
\]
Thus, the change of variable $\vu = 2 D \vv$ links the two formulas
\[
    \frac{1}{\sqrt{\det(R_m)}}
    \int_{(0, +\infty)^{m-1}}
    e^{-\vv^T {R_m}^{-1} \vv / 2}
    d \vv
    =
    \frac{1}{\sqrt{\det(\mH_{K}(\vtau))}}
    \int_{(0, +\infty)^{m-1}}
    e^{-\vu^T {\mH_K(\vtau)}^{-1} \vu / 2}
    d \vu.   
\]

\subsection{Inversion Lemma for a Saddle Point Coordinate}\label{sec_appendix-flip-flop}

This section consists in proving \Cref{th_supercritical-becomes-subcritical}, rewritten below.
\begin{lemmastar}
Let $\mX$ and $\mY$ be two disjoint sets of adversaries. Let $\tilde{\vzeta}$ be the saddle point associated with $\proba\big(m \succ_{\valpha} \mathcal{X} \land m \preccurlyeq_{\valpha} \mathcal{Y} \big)$ . Let $j \in \mX$, $\mX' = \mX \setminus \{j\}$, and $\mY' = \mY \cup \{j\}$. 
Then the saddle point associated with $\proba\big(m \succ_{\valpha} \mathcal{X}' \land m \preccurlyeq_{\valpha} \mathcal{Y}' \big)$ has its $k$-th coordinate equal to $\tilde{\zeta}_k$ if $k \neq j$ and $\frac{1}{\tilde{\zeta}_j}$ if $k=j$.
\end{lemmastar}

\begin{proof}
Transferring a candidate from $\mX$ to $\mY$ corresponds to the following algebraic operation:
\begin{equation}\label{eq_polynomials-flip-flap}
P_{\mY'}^{\mX'}\big(\vx_{{}_{\mX'}},\vy_{{}_{\mY'}}\big) = 
 y_j P_{\mY}^{\mX}\big(\vx_{{}_{\mX'}}, \frac{1}{y_j},\vy_{{}_{\mY}}\big),  
\end{equation}
where $y_j$ encodes when $m$ is preferred to $j$, whereas $x_j$ encoded the opposite.

Now, Let us denote $(\vbeta,\valpha)$ the threshold vector associated with $\proba\big(m \succ_{\valpha} \mX \land m \preccurlyeq_{\valpha} \mY\big)$, and $(\hat{\vbeta},\hat{\valpha})$ the threshold vector associated with  $\proba\big(m \succ_{\valpha} \mX'\land m \preccurlyeq_{\valpha} \mY'\big)$. Remark that $\hat{\vbeta}$ is the vector $\vbeta$ where the coordinate $\beta_j$ has been removed, and $\hat{\valpha}$ is the vector $\valpha$ where the coordinate $\alpha_j$ has been added.
Finally, let us denote $\hat{\vzeta}$ the saddle point associated with $\proba\big(m \succ_{\valpha} \mX'\land m \preccurlyeq_{\valpha} \mY'\big)$. By definition of $\hat{\vzeta}$, extending the remark in \Cref{sec_saddle-point-method}, we have
\[\hat{\vzeta} = \argmin_{(\reals_{> 0})^{|\mX' \cup \mY'|}} \frac{P_{\mY'}^{\mX'}\big(\vx_{{}_{\mX'}},\vy_{{}_{\mY'}}\big)}{{\vx_{{}_{\mX'}}}^{\hat{\vbeta}} {\vy_{{}_{\mY'}}}^{\hat{\valpha}} }.\ \]
So injecting \Cref{eq_polynomials-flip-flap} we get
\[\hat{\vzeta} = \argmin_{(\reals_{> 0})^{|\mX' \cup \mY'|}} \frac{P_{\mY}^{\mX}\big(\vx_{{}_{\mX'}}, \frac{1}{y_j},\vy_{{}_{\mY'}}\big)}{{\vx_{{}_{\mX'}}}^{\hat{\vbeta}} {y_j}^{-\beta_j}{\vy_{{}_{\mY}}}^{\valpha} }.\]
Since the minimum of a function over a positive orthant remains unchanged under the inversion of a coordinate, setting \( x_j = \frac{1}{y_j} \) yields
\[\min_{(\reals_{> 0})^{|\mX' \cup \mY'|}} \frac{P_{\mY}^{\mX}\big(\vx_{{}_{\mX'}}, \frac{1}{y_j},\vy_{{}_{\mY'}}\big)}{{\vx_{{}_{\mX'}}}^{\hat{\vbeta}} {y_j}^{-\beta_j}{\vy_{{}_{\mY}}}^{\valpha} } = \min_{(\reals_{> 0})^{|\mX \cup \mY|} }\frac{P_{\mY}^{\mX}\big(\vx_{{}_{\mX}},\vy_{{}_{\mY}}\big)}{{\vx_{{}_{\mX}}}^{\vbeta} {\vy_{{}_{\mY}}}^{\valpha} }, \]
and more importantly, the transformation \( x_j = \frac{1}{y_j} \) yields
\[\hat{\vzeta}=(\tilde{\zeta}_1,\dots,\tilde{\zeta}_{j-1}, \frac{1}{\tilde{\zeta}_j}, \tilde{\zeta}_{j+1},\dots,\tilde{\zeta}_{m-1}).\]
\end{proof}

\newpage
\section{Properties and Proofs for Particular Cultures}\label{sec_appendix-ic-and-mallows}

\subsection{Saddle Point, Hessian Matrix and Terms of the Asymptotic Expansion in Impartial Culture}
\label{sec_constants_IC}

We provide the proofs of \Cref{{th_impartial_culture},th_IC_error_term},
both considering Impartial Culture on $m$ candidates.
For each $j \in \mA$, let $X_j$ denote the random variable
taking value $1$ if candidate $j$ is ranked higher than candidate $m$,
and $0$ otherwise.
The probability generating function of the vector $\vX$ is $P(\vx)$
and its cumulant generating function $\vt \mapsto \mean(e^{\vt^T \vX})$
is $K : \vt \mapsto \log(P(e^{\vt}))$.
We have $\mean(\vX) = \vone/2$.
For our analysis, we need to compute the second and third central moments of $\vX$,
which correspond to the terms of order $2$ and $3$ in the Taylor expansion of $K$.

\paragraph{Second central moment.}
The Hessian of $K$ at $\vzero$ is the covariance matrix of $\vX$,
whose coefficient $(j,k)$ is equal to
\[
    \mH_K(\vzero)_{j,k} = \mean((X_j - 1/2)(X_k - 1/2)) = \mean(X_j X_k) - \frac{1}{4}.
\]
For $j = k$, we deduce $\mH_K(\vzero)_{j,k} = 1/4$.
In a random ranking $r$,
the probability that each candidate from a given subset $\mY$ of adversaries
beats candidate $m$ is $\frac{1}{|\mY| + 1}$.
Indeed, the ranking induced by $r$ on $\mY \cup \{m\}$
is uniformly distributed, so the probability that $m$ is in the last position
is $\frac{1}{|\mY| + 1}$. This implies for $j \neq k$
\[
    \mH_K(\vzero)_{j,k} = \frac{1}{3} - \frac{1}{4} = \frac{1}{12}.
\]
Letting $I$ denote the identity matrix of dimension $m-1$
and $J$ the matrix containing only $1$s, we deduce
\[
    \mH_K(\vzero) = \frac{1}{12} J + \frac{1}{6} I.
\]
The expressions of its inverse and determinant from \Cref{eq:mHKIC} follow.

\paragraph{Third central moment.}
A useful property of the cumulant generating function
is that the third term of its Taylor expansion
$\phi_3 := \vt \mapsto [y^3] \log(P(e^{\vt y}))$
(using the notation from \Cref{th:compute_a1})
is equal to the third central moment of $\vX$.
Thus
\[
    \phi_3(\vt) =
    \sum_{j, k, \ell \in \{1, \ldots, m-1\}}
    \mean((X_j - 1/2) (X_k - 1/2) (X_{\ell} - 1/2))
    t_j t_k t_{\ell}.
\]
We develop
\[
\mean((X_j - 1/2) (X_k - 1/2) (X_{\ell} - 1/2)) =
\mean(X_j X_k X_{\ell}) - \frac{1}{2} (\mean(X_j X_k) + \mean(X_j X_{\ell}) + \mean(X_k X_{\ell})) + \frac{3}{4} \mean(X_j) - \frac{1}{8}.
\]
As we saw earlier, we have $\mean(X_j X_k X_{\ell}) = \frac{1}{4}$.
In all cases, whether $j$, $k$ and $\ell$ are distinct
or some of them are equal, we find the third central moment to be $0$
and deduce $\phi_3(\vt) = 0$.

\paragraph{Proof of \Cref{th_impartial_culture}.}
The number of permutations of $\{1, \ldots, m\}$
where the set of elements above $m$ is exactly $\mX$
is $|\mX|! (m-1-|\mX|)!$.
Thus, for the Impartial Culture,
the characteristic polynomial from \cref{def_def-of-P} is
\[
    P(\vx) =
    \sum_{\mX \subseteq \mA}
    \frac{|\mX|! (m-1-|\mX|)!}{m!}
    \prod_{j \in \mX} x_j.
\]
The saddle point and log saddle point are
$\vzeta = \vone$ and $\vtau = \vzero$,
and we have $\vbeta = \vone / 2$.
\Cref{th_impartial_culture} is obtained by application
of \Cref{th_critical} with those parameters.

\paragraph{Proof of \Cref{th_IC_error_term}.}
The term $a_0$ is computed in \Cref{th_impartial_culture}.
The term $a_{1,n}$ is obtained by application of \Cref{th:compute_a1}
with $c = d = m-1$, $\vb = - \vone / 2$, $A(\vz) = 1$ so $\nabla_A(\vzeta) = \vzero$, $\mH = \mH_K(\vzero)$,
$\phi_3^{\star}(\vu) = 0$ because $\phi_3(\vt) = 0$, and $\vkappa_n - \vlambda n = (\lceil n / 2 \rceil - 1 - n / 2) \vone$.

\subsection{Saddle Points in the Mallows Models $\mM_{m \textnormal{ last}}$ and $\mM_{m \textnormal{ first}}$}\label{sec_appendix-saddle-point-Mallows}

This section is dedicated to proving \Cref{th_saddle-point-mallows} and \Cref{th_saddle-point-mallows-m-first}. Since the proof follows the exact same steps in both cases, we present it only for the culture $\mM_{m \textnormal{ last}}$.
\begin{lemmastar}
    Under the Mallows culture $\mathcal{M}_{m \; \textnormal{last}}$, the log saddle point $\vtau$ is given by
    
    \[ \vtau = \left(\frac{-m\rho}{2},\frac{(-m+2)\rho}{2}, \dots, \frac{(m-4)\rho}{2} \right),\]
    where $\rho$ is the concentration parameter of the culture, defined in \Cref{sec_definitions-and-notations}.
\end{lemmastar}

\begin{proof}

To prove this result, we leverage the probabilistic interpretation of the saddle point, namely that in the distribution induced by the log saddle point \( \vtau \), candidate \( m \) is, in expectation, precisely at the threshold for being an \( \valpha \)-winner.

In the general case, this interpretation translates to the following equation for each $\zeta_j$
\[ \sum_{r, \ j> m} p_r \Big(\prod_{k > m \textnormal{ in } r} \zeta_k \Big)= \sum_{r, \ j < m} p_r \Big(\prod_{k > m \textnormal{ in } r} \zeta_k \Big). \]

For the Mallows culture \( \mM_{m \textnormal{ last}} \), the equilibrium condition can be formulated in terms of ranking pairs. Specifically, due to the structure of the Mallows distribution, a ranking of the form \( (\sigma, m, \sigma') \) balances with the ranking obtained by swapping \( \sigma \) and \( \sigma' \), \ie \( (\sigma', m, \sigma) \).

For instance, when \( m=5 \), the ranking \( (1,3,2,5,4) \) must be in equilibrium with \( (4,5,1,3,2) \), yielding
\[ p_{(1,3,2,5,4)} \zeta_1 \zeta_2 \zeta_3 = p_{(4,5,1,3,2)} \zeta_4.\]

Therefore, for an arbitrary \( m \), the saddle point satisfies \( m!/2 \) equilibrium equations. Indeed, for each suitable pair of rankings \( r \) and \( r' \), the saddle point satisfies
\[
p_r \prod_{k > m \textnormal{ in } r } \zeta_k = p_{r'} \prod_{k > m \textnormal{ in } r' } \zeta_k.
\]
By considering the equation
\[
p_{\{ j \}} \zeta_j =  p_{\mA \setminus \{ j \}} \prod_{k \neq j} \zeta_k
\]
for each adversary \( j \), we obtain a system of \( m-1 \) independent equations with \( m-1 \) unknowns (except when \( m = 3 \), where the system reduces to a single equation). Solving this system, while ensuring that all coordinates of the saddle point are nonzero, yields the desired result.
In the case of \( m = 3 \), the system simplifies to \( p_{\{ 1 \}} \zeta_1 = p_{\{ 2 \}} \zeta_2 \). Adding the equation \( p_{\{ 1, 2 \}} \zeta_1 \zeta_2 = p_{\emptyset} \) provides two independent equations, allowing for a unique solution.

\end{proof}

\subsection{Asymptotic Behavior in the Mallows Model $\mM_{3 \textnormal{ first}}$}\label{sec_appendix-mallows-3-first}

This section outlines the main steps leading to the asymptotic formula
\[ 1-\probacondorcetwinner{3} \underset{n \to +\infty}{\sim}  \sqrt{\frac{2}{\pi n}} \frac{e^{-\rho \lceil n/2 \rceil} }{1-e^{-\rho}}\left(\frac{2}{1+e^{-\rho}} \right)^n .
\]

First, it is needed to iteratively apply \Cref{eq_polynomials-flip-flap} on the supercritical coordinates, until every term is subcritical, which gives the following equality
\[ 
1- \probacondorcetwinner{3} =  \big( 
 \proba( 3 \preccurlyeq_{\frac{\vone}{\vtwo}} \{1\}) + \proba( 3 \preccurlyeq_{\frac{\vone}{\vtwo}} \{2\})  \big) - \proba( 3 \preccurlyeq_{\frac{\vone}{\vtwo}} \{1,2\}). 
 \]

\subsubsection*{Subcriticality}
Note that the characteristic polynomial associated with \( \proba( 3 \preccurlyeq_{\frac{\vone}{\vtwo}} \{1,2\}) \), namely \( P_{\{1,2\}}^{\emptyset}(y_1 y_2) \), is actually the characteristic polynomial of the probability \( \probacondorcetwinner{3} \) under the Mallows model \( \mM_{3 \textnormal{ last}} \). Therefore, we already know that the saddle point for this term is subcritical.

The characteristic polynomials $P_{\{1\}}^{\{2\}}(y_1,1)$ and $P_{\{2\}}^{\{1\}}(1,y_2)$ associated respectively with $\proba( 3 \preccurlyeq_{\valpha} \{1\})$ and $\proba( 3 \preccurlyeq_{\valpha} \{2\})$ are both univariate and of the form $\gamma(q + p y)$.
Denoting $\xi$ a saddle point associated with such a polynomial, a straightforward calculation yields 
\(
\xi = \frac{q}{p}.
\)
The verification of subcriticality follows directly, by simply writing explicitly  
\[ 
P^{\{2\}}_{\{1\}}(y_1,1) = \gamma  \left( ( 2e^{-2\rho} + e^{-3\rho}) +  (1+ 2e^{-\rho}) y_1  \right), 
\] and 
\[P^{\{1\}}_{\{2\}}(1,y_2) = \gamma \left( (e^{-\rho} + e^{-2\rho}+e^{-3\rho} ) + (1 + e^{-\rho} + e^{-2\rho} )y_2\right).\]

\subsubsection*{Dominating term}
The first two terms can be associated with a Mallows culture with two candidates and a concentration parameter $\tilde{\rho} = \log(\frac{q}{p})$. In the case of $ \proba( 3 \preccurlyeq_{\valpha} \{2\})$, $\tilde{\rho} = \rho$. In the case of $ \proba( 3 \preccurlyeq_{\valpha} \{1\})$ , $\tilde{\rho}$ is equal to $\rho $ multiplied by a quantity strictly less than 1. Therefore, this Mallows is less concentrated than the first one. Thus, asymptotically, $\proba( 3 \preccurlyeq_{\valpha} \{2\})$ will dominate $\proba( 3 \preccurlyeq_{\valpha} \{1\})$.

In particular, as $  \proba( 3 \preccurlyeq_{\valpha} \{1,2\})\leq  \proba( 3 \preccurlyeq_{\valpha} \{1\})$, the third term of the sum is also dominated by $\proba( 3 \preccurlyeq_{\valpha} \{2\})$.
Thus,
\[ 
1-\probacondorcetwinner{3} \underset{n \to +\infty}{\sim} \proba( 3 \preccurlyeq_{\valpha} \{2\}).
\]

Now, using \Cref{th_proba-X-Y}, we get
\[ \proba( 3 \preccurlyeq_{\valpha} \{2\}) \underset{n \to +\infty}{\sim} \frac{ P_{\{2\}}^{\{1\}}(1,\xi)^n}{(1-\xi)\xi^{\lfloor n/2 \rfloor}} \frac{1}{\sqrt{2\pi n \frac{1}{4}}},  \]
where $\xi = e^{-\rho}$, thus
$P_{\{2\}}^{\{1\}}(1,\xi) = 2\gamma e^{-\rho}(1+e^{-\rho}+e^{-2\rho})$, with $\gamma = \frac{1}{1+2e^{-\rho}+2e^{-2\rho}+e^{-3\rho}}$.

Straightforward calculations yield the final formula.

\end{document}